\documentclass{amsart}
\usepackage{fullpage,graphicx,mathpazo,color}
\usepackage{amsmath,amscd}
\usepackage{subfigure}

\begin{document}

\newtheorem{rhp}{Riemann-Hilbert Problem}
\newtheorem{theorem}{Theorem}
\newtheorem{lemma}{Lemma}
\newtheorem{proposition}{Proposition}
\newtheorem{corollary}{Corollary}
\newtheorem{conjecture}{Conjecture}
\newtheorem{remark}{Remark}

\newcommand{\F}[1]{{\color{blue}#1}}
\newcommand{\Li}[1]{{\color{red} #1}}
\newcommand{\B}[1]{{\color{cyan} #1}}

\newcommand{\ii}{\mathrm{i}}
\newcommand{\ee}{\mathrm{e}}
\newcommand{\dd}{\mathrm{d}}
\newcommand{\T}{\mathrm{T}}

\newcommand{\lm}[1]{{\color{red} #1}}
\title{Darboux transformation and solitonic solution to the coupled complex short pulse equation}
\author{Baofeng Feng$^{1}$}\email{baofeng.feng@utrgv.edu}
\author{Liming Ling$^{2}$}\email{linglm@scut.edu.cn; lingliming@qq.com}

\address{$^1$ School of Mathematical and Statistical Sciences, The University of Texas Rio Grande Valley, Edinburg TX, 78541-2999, USA}
\address{$^2$ School of Mathematics, South China University of Technology, Guangzhou 510640, China}

\date{\today}

\begin{abstract}
The Darboux transformation (DT) for the coupled complex short pulse (CCSP) equation is constructed through the loop group method. The DT is then utilized to construct various exact solutions including bright soliton, dark-soliton, breather and rogue wave solutions to the CCSP equation.  In case of vanishing boundary condition (VBC), we perform the inverse scattering analysis to understand the soliton solution better. Breather and rogue wave solutions are constructed in case of non-vanishing boundary condition (NVBC). Moreover, we conduct a modulational instability (MI) analysis based on the method of squared eigenfunctions, whose result confirms the condition for the existence of rogue wave solution.


\end{abstract}

\maketitle

\section{Introduction}
The nonlinear Schr\"odinger (NLS) equation plays a crucial role in nonlinear waves since it can be used to describe the evolution of slowly varying packets of quasi-monochromatic waves in weakly nonlinear dispersive media. {Since} the NLS equation is an integrable equation, {various exact solutions can be derived through some well-known methods such as the inverse scattering transform, Hirota bilinear method and Darboux transformation method. In case of} the vanishing boundary condition (VBC), the NLS equation admits multi-bright \cite{AblowitzC91}, breather, higher order soliton \cite{Yang10} and infinite order type solitonic solutions \cite{BilmanB-19,BilmanBW-21, BilmanM17,BilmanLM18,LingZ-21}. In the framework of the inverse scattering transform, the long time asymptotic behavior of solutions can be analyzed by the nonlinear steepest descent method, or the so-called Deift-Zhou method, \cite{Deift93,Zhou89} or D-bar method \cite{McLaughlinM06}. In other words, for arbitrary initial data in suitable functional space without the appearance of spectral singularity \cite{Klaus03}, the solutions will evolve into multi-solitons and the additional dispersive waves (the breather and higher order soliton solutions are asymptotic unstable), i.e. the soliton resolution conjecture \cite{BorgheseJM18}. For non-vanishing boundary condition (NVBC), the NLS equation admits breather \cite{Akhmediev86,Kuznetsov77,Ma79,XuGCZK19}, rogue wave solutions \cite{Peregrine83} in the focusing case and dark soliton solution in the defocusing case. The rogue wave and Akhmediev breather in the focusing case are related to the modulational instability (MI) \cite{AkhmedievAS09,BenjaminF67}. {For the defocusing case with NVBC}, the plane wave is modulational stable, so the defocusing NLS equation does not admit rogue wave or Akhmediev breather solutions.

The coupled NLS (CNLS) equation also known as Manakov system \cite{PrinariAB16}, {can be used as a model} in nonlinear birefringent optics and for two modes of Bose-Einstein condensate \cite{Yang10}. The localized wave solutions in the coupled NLS equation are more {complicated and richer than the} scalar one \cite{DegasperisLS18,LingZ19, PrinariL18,PrinariO18}. There are some differences {in compared with} the scalar one. In the focusing case of VBC, in addition to the bright solitons, there exists double-hump soliton solution \cite{StalinRSL19,QinZL-19}. If one component has a VBC and the other component possesses NVBC, there exist bright-dark soliton, rogue wave, breather solutions and the combinations thereof. When both components have NVBCs, there are two distinct cases. If they share the same wave-number, the plane wave background can be removed by a Galilean shift and in this case the breather, rogue wave and their superposition, all degenerate into solutions of the scalar NLS equation. The other case is that the two plane waves have distinct wave-numbers, in which there exist genuinely coupled breather and rogue wave solutions. {In regard to} the modulational instability, there exist two different branches which correspond to two different types of Akhemediev breather or rogue wave solutions. In addition to the dark-dark and bright-dark solitons, the defocusing coupled NLS also admits Akhmediev breather and rogue wave solutions {for special choices of the wave-number and frequency}, and in this case the background is modulationally unstable because of the effect of cross-phase modulation \cite{BaronioCDLOW14}. However the MI has only one branch.

In the regime of ultra-short pulses where the width of the optical pulse is of the order of femtoseconds ($10^{-15}$s), the propagation of ultra-short wave packets can be described by the complex short pulse (CSP) equation, which was proposed  for the focusing and defocusing case \cite{Feng15,FengLZ16}.  A coupled complex short pulse (CCSP) equation 
\begin{equation}\label{eq:ccsp}
\begin{split}
 q_{1,xt}+q_1+\frac{\kappa}{2}((|q_1|^2+|q_2|^2)q_{1,x})_x=&0, \\
 q_{2,xt}+q_2+\frac{\kappa}{2}((|q_1|^2+|q_2|^2)q_{2,x})_x=&0, \\
\end{split}
\end{equation}
was also proposed to describe the propagation of ultra short pulses in the birefringent fibers, where the symbol $\kappa=\pm1$ corresponds to the focusing or defocusing case respectively.
{In the original paper \cite{Feng15}, the author only considered the focusing case. The defocusing case can be naturally derived  along the way of the latter reference} \cite{FengLZ16}. As far as we know, the study of the CCSP equation has only a few. Most recently, an inverse scattering analysis is done for the CCSP equation with only the vanishing boundary condition \cite{GPTFccsp21}. 
Therefore, there is no systematic analysis for various soliton solutions for the CCSP equation, which motivates the present study.

The CCSP equation \eqref{eq:ccsp} admits the following Lax pair
\begin{equation}\label{eq:ccsp-lp}
\begin{split}
 \Psi_{x}=&\mathbf{X}(x,t;\lambda)\Psi,\,\,\,\,\mathbf{X}(x,t;\lambda)=\frac{1}{\lambda}\begin{bmatrix}
-\ii \mathbb{I}_2 & -\kappa\mathbf{Q}^{\dag}_x \\[5pt]
\mathbf{Q}_x & \ii\mathbb{I}_2 \\
 \end{bmatrix},\,\,\,\,\mathbf{Q}=\begin{bmatrix}
q_1(x,t)& q_2(x,t)\\[5pt]
q_2^*(x,t) & -q_1^*(x,t) \\
\end{bmatrix},  \\
 \Psi_{t}=&\mathbf{T}(x,t;\lambda)\Psi,\,\,\,\,\mathbf{T}(x,t;\lambda)=\begin{bmatrix}
-\frac{\ii}{4} \lambda\mathbb{I}_2+\frac{\ii\kappa}{2\lambda}\mathbf{Q}\mathbf{Q}^{\dag} & -\frac{\kappa}{2}\mathbf{Q}^{\dag}+\frac{1}{2\lambda}\mathbf{Q}^{\dag}\mathbf{Q}\mathbf{Q}_x^{\dag} \\[5pt]
-\frac{\ii}{2}\mathbf{Q}-\frac{\kappa}{2\lambda}\mathbf{Q}\mathbf{Q}^{\dag}\mathbf{Q}_x &
\frac{\ii}{4} \lambda\mathbb{I}_2-\frac{\ii\kappa}{2\lambda}\mathbf{Q}^{\dag}\mathbf{Q} \\
 \end{bmatrix},   \\
\end{split}
\end{equation}
where $\mathbb{I}_2$ denotes the $2\times 2$ identity matrix. The CCSP equation \eqref{eq:ccsp} admits an alternative form in conservation law
\begin{equation}\label{eq:con}
\left(\rho^{-1}\right)_t+\frac{\kappa}{2}\left((|q_1|^2+|q_2|^2)\rho^{-1}\right)_x=0\,,
\end{equation}
where $\rho^{-1}=\sqrt{1+\kappa|q_{1,x}|^2+\kappa|q_{2,x}|^2}$. Thus we can define a hodograph transformation
\begin{equation}\label{eq:hodo}
\dd y=\rho^{-1} \dd x-\frac{\kappa}{2}\rho^{-1} (|q_1|^2+|q_2|^2)\dd t,\,\,\,\,\, \dd s=-\dd t,
\end{equation}
which converts the Lax pair \eqref{eq:ccsp-lp} into $\Phi(y,s)=\Psi(x,t)$:
\begin{equation}\label{eq:laxcd}
\begin{split}
 \Phi_y&=\mathbf{U}(y,s;\lambda) \Phi,\,\,\,\,\,\, \mathbf{U}(y,s;\lambda)=-\ii\frac{\rho(y,s)}{\lambda}\Sigma_3-\frac{1}{\lambda}\Sigma_3\mathbf{V}_{0,y}, \\
 \Phi_s&=\mathbf{V}(y,s;\lambda) \Phi,\,\,\,\,\,\,\mathbf{V}(y,s;\lambda)=\frac{\ii}{4}\lambda\Sigma_3+\frac{\ii}{2}\mathbf{V}_0, \\
\end{split}
\end{equation}
where
\[
\Sigma_3={\rm diag}(1,1,-1,-1),\,\,\,\,\,\,\,\, \mathbf{V}_0=\begin{bmatrix}
\mathbf{0} & \kappa\mathbf{Q}^{\dag} \\
\mathbf{Q} & \mathbf{0} \\
\end{bmatrix},\,\,\,\,\,\,\,\, \mathbf{Q}=\begin{bmatrix}
q_1(y,s)& q_2(y,s)\\
q_2^*(y,s) & -q_1^*(y,s) \\
\end{bmatrix}.
\]
The compatibility condition  gives a two-component coupled complex dispersionless (TCCD) equation:
\begin{equation}\label{eq:cd}
\begin{split}
 q_{i,ys}&=\rho q_i,  \quad i=1,2 \\
 \rho_{s}&=-\frac{\kappa}{2}(|q_1|^2+|q_2|^2)_y.
\end{split}
\end{equation}
Conversely, by the third equation of \eqref{eq:cd}, we define the following inverse hodograph transformation:
\begin{equation}\label{eq:inv-hodo}
\dd x=\rho\,\dd y-\frac{\kappa}{2}(|q_1|^2+|q_2|^2)\dd s,\,\,\,\,\, \dd t=-\dd s,
\end{equation}
which converts the Lax pair \eqref{eq:laxcd} and equations \eqref{eq:cd} into Lax pair \eqref{eq:ccsp} and equations \eqref{eq:ccsp-lp} respectively. Thus {the CCSP equation \eqref{eq:ccsp} is equivalent to the TCCD equation \eqref{eq:cd} through the hodograph transformation \eqref{eq:hodo} and \eqref{eq:inv-hodo}, provided that $\rho>0$}.

We illustrate below how to derive the hodograph transformation \eqref{eq:hodo} by the formal scattering and inverse scattering analysis. From the Lax pair \eqref{eq:ccsp-lp}, we can suppose the matrix solutions of Lax pair \eqref{eq:ccsp-lp} possesses the following form:
\begin{equation}\label{eq:ansatz}
\Psi(x,t;\lambda)=\mathbf{m}(x,t;\lambda)\ee^{-\ii(f(x,t)\lambda^{-1}+\frac{\lambda}{4} t) \Sigma_3}
\end{equation}
where $\mathbf{m}(x,t;\lambda)$ is an invertible analytic matrix solution with respect to $\lambda$ in the neighborhood of $\infty$ and $0$: i.e.
\[
\mathbf{m}(x,t;\lambda)=\mathbf{m}_{\infty}(x,t)+\mathbf{m}_{-1}(x,t)\lambda^{-1}+O(\lambda^{-2}),\,\,\,\, \lambda\to\infty,
\]
and
\[
\mathbf{m}(x,t;\lambda)=\mathbf{m}_0(x,t)+\mathbf{m}_1(x,t)\lambda+O(\lambda^2),\,\,\,\, \lambda\to 0.
\]
Inserting the ansatz \eqref{eq:ansatz} into the Lax pair \eqref{eq:ccsp-lp} and expanding them at $\lambda=0$, we obtain
\begin{equation}\label{eq:lead-exp}
\begin{split}
 -\ii f_x\Sigma_3=&[\mathbf{m}_0(x,t)]^{-1}\begin{bmatrix}
-\ii \mathbb{I}_2 & -\kappa\mathbf{Q}^{\dag}_x \\[5pt]
\mathbf{Q}_x & \ii\mathbb{I}_2 \\
 \end{bmatrix}\mathbf{m}_0(x,t), \\
 -\ii f_t\Sigma_3=&-\frac{\kappa}{2}[\mathbf{m}_0(x,t)]^{-1}\begin{bmatrix}
-\ii\mathbf{Q}\mathbf{Q}^{\dag} & -\kappa\mathbf{Q}^{\dag}\mathbf{Q}\mathbf{Q}_x^{\dag} \\[5pt]
\mathbf{Q}\mathbf{Q}^{\dag}\mathbf{Q}_x &
\ii\mathbf{Q}^{\dag}\mathbf{Q} \\
 \end{bmatrix}\mathbf{m}_0(x,t). \\
\end{split}
\end{equation}
Taking the determinant of previous two equations \eqref{eq:lead-exp}, we arrive at
\begin{equation}\label{eq:f-x-t}
f_x^2=1+\kappa(|q_{1,x}|^2+|q_{2,x}|^2),\,\,\,\,\,\, f_t^2=-\frac{\kappa}{2}(|q_1|^2+|q_2|^2)\left[1+\kappa(|q_{1,x}|^2+|q_{2,x}|^2)\right]
\end{equation}
which shows that the function $f(x,t)$ is consistent with $y$ in \eqref{eq:inv-hodo}, i.e. $f(x,t)=y$. Thus we know that the hodograph transformation \eqref{eq:hodo} is necessary in the procedure of solving the CCSP equation \eqref{eq:ccsp}. On the other hand, the solution of Eq. \eqref{eq:ccsp} can be represented as the solutions of equations \eqref{eq:cd} and the inverse hodograph transformation \eqref{eq:inv-hodo}:
\[
x=\int_{(y_0,s_0)}^{(y,s)}\rho(y',s')\dd y'-\frac{\kappa}{2}(|q_1(y',s')|^2+|q_2(y',s')|^2)\dd s',\,\,\,\, t=-s.
\]
If $|q_1(y,s)|^2+|q_2(y,s)|^2\to \frac{1}{4}(\alpha_1^2+\alpha_2^2)$ and $\rho(y,s)\to \frac{\delta}{2}$ as $y\to\pm\infty$, then the above curvilinear integration can be simplified as:
\begin{equation}\label{eq:inv-hodo-1}
x=\frac{\delta}{2}y-\frac{\kappa}{8}(\alpha_1^2+\alpha_2^2)s+\int_{-\infty}^{y}\left[\rho(y',s)-\frac{\delta}{2}\right]\dd y',\,\,\,\, t=-s.
\end{equation}

Localized wave solutions can be {constructed} by using vector solutions of the Lax pair. By choosing different vector solutions, we can obtain the different localized wave solutions. The {key} is to analyze the structure of the localized wave solutions and their mechanism of formation. Thus, for the focusing CCSP equation \eqref{eq:ccsp} on a zero background, we firstly use the scattering and inverse scattering analysis to analyze the spectral problem \eqref{eq:laxcd}. On one hand, by using the ansatz of the Darboux matrix, we are able to obtain compact formulas for the soliton solutions. On the other hand, the inverse scattering method provides a complete spectral classification of the solutions.

For NVBC, since the spectral problem is $4\times 4$ which involves genus one algebraic curve, it is difficult to perform the scattering/inverse scattering analysis. Thus, we only construct the localized wave solutions. Especially, we establish a relation between rogue wave solution and modulational instability. Moreover, we apply a method to {conduct} the linear stability analysis based on the squared eigenfunctions and vector solutions of Lax pair. This analysis {reveals the connection between the existence of the rogue waves and the modulational instability}.

This paper is organized as follows. In Section \ref{sec2}, the Darboux transformation \cite{guo12} and B\"acklund transformation for the system \eqref{eq:laxcd} are constructed through the loop group method. In Section \ref{sec3}, we firstly analyze the Lax pair \eqref{eq:laxcd} by the Riemann-Hilbert approach. Exact solutions are constructed both by means of the Riemann-Hilbert approach, and Darboux transformation, which include the single soliton, ${\rm SU}(2)$ soliton, double-hump, breather, multi-soliton and high order soliton solutions. In Section \ref{sec4}, a modulational instability analysis method based on the squared eigenfunctions is applied for equation \eqref{eq:cd}. Rogue wave solution is constructed exactly by the Darboux transformation, and their formation is elucidated by the MI analysis. Besides the rogue waves, other types of localized wave solution are also constructed by the Darboux transformation. Section \ref{sec5} is devoted to concluding remarks.

\section{Darboux transformation}\label{sec2}
Firstly, we need to understand the structure of the Lax pair \eqref{eq:laxcd}, which is the first negative flow for the matrix nonlinear Schr\"odinger hierarchy. To this end, we review the {relevant} theory in the literature \cite{TerngU00}. Let
\begin{equation*}
  \begin{split}
  {\rm sl}(4)_{\Sigma_3}&=\left\{y\in{\rm sl}(4,\mathbb{C})| \,[\Sigma_3,y]=0  \right\}, \\
  {\rm sl}(4)_{\Sigma_3}^{\perp}&=\left\{y\in{\rm sl}(4,\mathbb{C})| {\rm tr}(zy)=0\,\, \text{for all } z\in  {\rm sl}(4)_{\Sigma_3}\right\},
  \end{split}
\end{equation*}
denote the centralizer of $\Sigma_3$ and its orthogonal complement in ${\rm sl}(4,\mathbb{C})$. Furthermore, we consider the reality condition. It is easy to see that the Lax pair \eqref{eq:laxcd} satisfies the symmetry conditions ${\rm su}(4)$--reality condition ($[\mathbf{A}(\lambda^*)]^{\dag}+\mathbf{A}(\lambda)=0$ for all $\lambda\in\mathbb{C}$) or ${\rm su}(2,2)$--reality condition ($[\mathbf{A}(\lambda^*)]^{\dag}+\Sigma_3\mathbf{A}(\lambda)\Sigma_3=0$ for all $\lambda\in\mathbb{C}$) for choosing different sign, where $\mathbf{A}(\lambda)=\sum_{k\leq n_0} \mathbf{A}_k \lambda^k$. Then we consider the ${\rm su}(4)$ or ${\rm su}(2,2)$-reality condition twisted by $\sigma: [\mathbf{A}(-\lambda^*)]^*\rightarrow \Lambda [\mathbf{A}(-\lambda^*)]^*\Lambda^{-1}$, i.e. $\sigma([\mathbf{A}(-\lambda^*)]^*)=\mathbf{A}(\lambda).$ In other words, the matrices $\mathbf{U}(\lambda)$ and $\mathbf{V}(\lambda)$ in the system \eqref{eq:laxcd} satisfy the following symmetry conditions:
\begin{equation}\label{eq:sym1}
\begin{split}
[\mathbf{U}(y,s;\lambda^*)]^{\dag}+\Sigma\mathbf{U}(y,s;\lambda)\Sigma&=0,\,\,\,\,\,\, [\mathbf{V}(y,s;\lambda^*)]^{\dag}+\Sigma\mathbf{V}(y,s;\lambda)\Sigma=0,\\
\Lambda[\mathbf{U}(y,s;-\lambda^*)]^{*}\Lambda^{-1}-\mathbf{U}(y,s;\lambda)&=0,\,\,\,\,\,\, \Lambda[\mathbf{V}(y,s;-\lambda^*)]^{*}\Lambda^{-1}-\mathbf{V}(y,s;\lambda)=0,
\end{split}
\end{equation}
where
\[
\Sigma={\rm diag}(1,1,\kappa,\kappa),\,\,\,\, \Lambda={\rm diag}(\sigma_2,\sigma_2),\,\,\,\,\, \sigma_2=\begin{bmatrix}
0&1 \\
-1&0 \\
\end{bmatrix}.
\]

Suppose we have a fundamental matrix solution $\Phi(y,s;\lambda)$ which is normalized at {the point} $(y,s)=(0,0)$ by taking $\Phi(0,0;\lambda)=\mathbb{I}.$ By the symmetry conditions \eqref{eq:sym1}, we have the following symmetry conditions for the fundamental solution matrix
\begin{equation*}
  [\Phi^{\dag}(y,s;\lambda^*)]^{-1}=\Sigma\Phi(y,s;\lambda)\Sigma,\,\,\,\,\,\,\, \Lambda \Phi(y,s;\lambda)\Lambda^{-1}=\Phi^*(y,s;-\lambda^*),
\end{equation*}
respectively. Meanwhile, the symmetry for the vector solutions is helpful to construct the Darboux matrix. If $\Phi_1(y,s;\lambda_1)$ is a vector solution for the Lax pair at $\lambda=\lambda_1$, {so is $\Lambda\Phi_1^*(y,s;\lambda_1)$ at $\lambda=-\lambda_1^*$}.

Suppose there exists a Darboux matrix $\mathbf{T}(y,s;\lambda)$ which converts the wave function $\Phi(y,s;\lambda)$ into a new wave function $\Phi^{[1]}(y,s;\lambda)=\mathbf{T}(y,s;\lambda)\Phi(y,s;\lambda)\mathbf{T}^{-1}(0,0;\lambda)$ which is normalized at $(y,s)=(0,0)$ to the identity matrix. Through the symmetry relation \eqref{eq:sym1} and existence and uniqueness theorem for ordinary differential equations, we {have}
\begin{equation*}
  \left[\Phi^{[1]\dag}(y,s;\lambda^*)\right]^{-1}=\Sigma\Phi^{[1]}(y,s;\lambda)\Sigma,\,\,\,\,\,\,\, \Lambda \Phi^{[1]}(y,s;\lambda)\Lambda^{-1}=\Phi^{[1]*}(y,s;-\lambda^*),
\end{equation*}
which induce the following symmetry relations for the Darboux matrix:
\begin{equation}\label{dt-sym}
  [\mathbf{T}(y,s;\lambda^*)]^{-1}=\Sigma\mathbf{T}(y,s;\lambda)\Sigma,\,\,\,\,\,\,\, \Lambda \mathbf{T}(y,s;\lambda)\Lambda^{-1}=\mathbf{T}(y,s;-\lambda^*).
\end{equation}
The first symmetry relation in \eqref{dt-sym} follows from the $su(4)$ or $su(2,2)$-reality condition, and the second one is from the twisted symmetry $\sigma$. By the loop group construction of Darboux transformation theory \cite{TerngU00} and the first symmetry in \eqref{dt-sym}, the elementary Darboux matrix can be represented by
\begin{equation}\label{darboux-matrix}
  \mathbf{T}(y,s;\lambda)=\mathbb{I}_4-\left[\Phi_1,\Phi_2\right]\mathbf{M}^{-1}{\rm diag}(\lambda-\lambda_1^*,\lambda-\lambda_2^*)^{-1}
\begin{bmatrix}
\Phi_1^{\dag}\\[5pt]
\Phi_2^{\dag} \\
\end{bmatrix}\Sigma
\end{equation}
where $\Phi_i=\Phi(y,s;\lambda_i)\mathbf{v}_i$ $(i=1,2)$ are vector solutions for the Lax pair \eqref{eq:laxcd} at $\lambda=\lambda_i$, and
\[
\mathbf{M}=\begin{bmatrix}
\frac{\Phi_1^{\dag}\Sigma\Phi_1}{\lambda_1-\lambda_1^*} & \frac{\Phi_1^{\dag}\Sigma\Phi_2}{\lambda_2-\lambda_1^*}\\
\frac{\Phi_2^{\dag}\Sigma\Phi_1}{\lambda_1-\lambda_2^*} & \frac{\Phi_2^{\dag}\Sigma\Phi_2}{\lambda_2-\lambda_2^*} \\
\end{bmatrix}.
\]
The above Darboux matrix \eqref{darboux-matrix} is uniquely determined by the kernel conditions
\begin{equation}\label{eq:ker-dt}
\mathbf{T}(y,s;\lambda_i)\Phi(y,s;\lambda_i)\mathbf{v}_i=0
\end{equation}
and residue conditions
\begin{equation}\label{eq:res-dt}
\underset{\lambda=\lambda_i^*}{{\rm Res}}\left(\mathbf{T}(y,s;\lambda)\Phi(y,s;\lambda)\Sigma\Sigma\mathbf{w}_i\right)=0
\end{equation}
where $\mathbf{v}_i^{\dag}\Sigma \mathbf{w}_i=0$ and ${\rm rank}(\mathbf{v}_i)=1,\, {\rm rank}(\mathbf{w}_i)=3,$ $i=1,2.$ We can {observe} that the residue conditions are determined by the ${\rm su}(4)$ or ${\rm su}(2,2)$ symmetry.

To keep the second symmetry in equation \eqref{dt-sym}, we merely need to set $\Phi_2=\Lambda \Phi_1^*(y,s;\lambda_1)$ with $\lambda_2=-\lambda_1^*.$ Thus the kernel conditions for the Darboux matrix are: ${\rm Ker}(\mathbf{T}(s,y;\lambda_1))={\rm span}\{\Phi_1\}$ and ${\rm Ker}(\mathbf{T}(s,y;-\lambda_1^*))={\rm span}\{\Lambda\Phi_1^*\}$.

Recall that $L_{+,0}(\mathbf{GL}(n,\mathbb{C}))$ is the group of holomorphic maps from $\mathbb{C}/\{0\}$ to $\mathbf{GL}(n,\mathbb{C})$, $L_{-,0}(\mathbf{GL}(n,\mathbb{C}))$ is the group of holomorphic maps $h$ from $O_{\infty}\cup O_0$ to $\mathbf{GL}(n,\mathbb{C})$ such that $h(\infty)=\mathbb{I}$, $L(\mathbf{GL}(n,\mathbb{C}))$ is the group of holomorphic maps from $\mathbb{S}^{1}=\{z\in \mathbb{C}: |z|=1\}$ to $\mathbf{GL}(n,\mathbb{C})$, where $O_{\infty}$ and $O_0$ represent the neighborhood of $\infty$ and $0$ on the compact Riemann surface $\mathbb{S}^2=\mathbb{C}\cup\{\infty\}$ respectively. We know that $\Phi(y,s;\lambda)\in L_{-,0}(\mathbf{SU}(4))$ [or $\Phi(y,s;\lambda)\in L_{-,0}(\mathbf{SU}(2,2))$] and $\mathbf{T}(y,s;\lambda)\in L_{+,0}(\mathbf{SU}(4))$ [or $\mathbf{T}(y,s;\lambda)\in L_{+,0}(\mathbf{SU}(2,2))$].

For the linear system \eqref{eq:laxcd}, the above results for the Darboux matrix can be summarized in the following theorem:
\begin{theorem}\label{thm1} Let $\Phi(y,s;\lambda)$ be a holomorphic matrix function for $\lambda\in\mathbb{S}^2\setminus\{0,\infty\}$ with $\Phi(0,0;\lambda)=\mathbb{I}$.
The Darboux matrix
\begin{equation}\label{eq:dt}
\mathbf{T}(y,s;\lambda)=\mathbb{I}-\left[\Phi_1,\Lambda\Phi_1^*\right]\left(\frac{\Phi_1^{\dag}\Sigma\Phi_1}{\lambda_1-\lambda_1^*}\right)^{-1}
\begin{bmatrix}
\frac{\Phi_1^{\dag}}{\lambda-\lambda_1^*}\\[5pt]
-\frac{\Phi_1^{\T}\Lambda}{\lambda+\lambda_1} \\
\end{bmatrix}\Sigma,
\end{equation}
where $\Phi_1=\Phi_1(y,s)=\begin{bmatrix}
\Theta_1\\
\Psi_1\\
\end{bmatrix}=\Phi(y,s;\lambda_1)\mathbf{v}_1$ is a special solution for the linear system \eqref{eq:laxcd} at $\lambda=\lambda_1$, and $\mathbf{v}_1$ is a constant vector which is independent of $y$ and $s$,
converts the Lax pair \eqref{eq:laxcd} into a new one by replacing the old potential function $\mathbf{Q}$ with the new one:
\begin{equation}\label{eq:backlund}
\begin{split}
\mathbf{Q}[1]&=\mathbf{Q}-\left(\frac{\Phi_1^{\dag}\Sigma\Phi_1}{\lambda_1-\lambda_1^*}\right)^{-1}\left(\Psi_1\Theta_1^{\dag}
-\sigma_2\Psi_1^*\Theta_1^{\T}\sigma_2\right),\\
\rho[1]&=\rho-2 \ln_{y,s}\left(\frac{\Phi_1^{\dag}\Sigma\Phi_1}{\lambda_1-\lambda_1^*}\right),\\
|q_1[1]|^2+|q_2[1]|^2&=|q_1|^2+|q_2|^2+4\kappa\ln_{ss}\left(\frac{\Phi_1^{\dag}\Sigma\Phi_1}{\lambda_1-\lambda_1^*}\right).
\end{split}
\end{equation}
\end{theorem}
\begin{remark}\label{rem1}
When $\lambda_1\in \mathbb{R}$ and $\kappa=-1$, we obtain that
\begin{equation}\label{eq:p1-dark}
      \begin{split}
       \left(\Phi_1^{\dag}(y,s;\lambda_1)\Sigma_3\Phi_i(y,s;\lambda)\right)_y=& \frac{\lambda-\lambda_1}{\lambda\lambda_1}\Phi_1^{\dag}(y,s;\lambda_1)\left(\ii \rho(y,s)\mathbb{I}_4+ \mathbf{V}_{0,y}\right)\Phi_i(y,s;\lambda), \\
       \left(\Phi_1^{\dag}(y,s;\lambda_1)\Sigma_3\Phi_i(y,s;\lambda)\right)_s=& \frac{\ii}{4}(\lambda-\lambda_1)\Phi_1^{\dag}(y,s;\lambda_1)\Phi_i(y,s;\lambda), \\
      \end{split}
      \end{equation}
where $\Phi_i(y,s;\lambda)=V_i(y,s;\lambda)\ee^{\ii(\xi_i(\lambda) y+\eta_i(\lambda)s)}$, $\xi_2(\lambda)=\xi_1^*(\lambda)$, $\eta_2(\lambda)=\eta_1^*(\lambda)$, $(\lambda\in \mathbb{R}; i=1,2)$, $V_i(y,s;\lambda)$ are bounded vector functions.
Thus \[\left(\Phi_1^{\dag}(y,s;\lambda_1)\Sigma_3\Phi_i(y,s;\lambda_1)\right)_y=\left(\Phi_1^{\dag}(y,s;\lambda_1)\Sigma_3\Phi_i(y,s;\lambda_1)\right)_s=0,\] which implies that \begin{equation}\label{eq:dark-con}
\Phi_1^{\dag}(y,s;\lambda_1)\Sigma_3\Phi_1(y,s;\lambda_1)=0,\,\,\,\,\,\,\, \Phi_1^{\dag}(y,s;\lambda_1)\Sigma_3\Phi_2(y,s;\lambda_1)=\kappa(\lambda_1),
\end{equation}
where $\kappa(\lambda_1)$ is a parameter merely depending on $\lambda_1$.
The Darboux matrix given in {\bf Theorem }\ref{thm1} should be modified in the following form:
      \begin{equation}\label{eq:dt-dark}
      \mathbf{T}(y,s;\lambda)=\mathbb{I}-\left[\Phi_1,\Lambda\Phi_1^*\right]
\begin{bmatrix}
\frac{\Phi_1^{\dag}}{\lambda-\lambda_1}\\[5pt]
-\frac{\Phi_1^{\T}\Lambda}{\lambda+\lambda_1} \\
\end{bmatrix}\frac{\Sigma_3}{\Omega_1(y,s)},
      \end{equation}
      where
      \begin{equation*}
      \begin{split}
        \Omega_1(y,s)&\equiv\lim_{\lambda\to\lambda_1}\frac{\Phi_1^{\dag}(y,s;\lambda_1)\Sigma_3\left[\Phi_1(y,s;\lambda)+\tau (\lambda-\lambda_1)\Phi_2(y,s;\lambda)\right]}{(\lambda-\lambda_1)}\\
        &=\Phi_1^{\dag}(y,s;\lambda_1)\Sigma_3\Phi_1'(y,s;\lambda_1)+\tau\kappa(\lambda_1)\\
        &=\Phi_1^{\dag}(y,s;\lambda_1)\Sigma_3\Phi'(y,s;\lambda_1)\mathbf{v}_1+\Omega_1(0,0),\,\,\,\, \Omega_1(0,0)=\mathbf{v}_1^{\dag}\Sigma_3\mathbf{v}_1'+\tau\kappa(\lambda_1)\,.
      \end{split}
      \end{equation*}
      Here the prime denotes the derivative with respect to $\lambda$ at $\lambda=\lambda_1$, $\tau$ is an appropriate parameter to keep the non-singularity of $\Omega_1(y,s)$, and $\Omega_1(y,s)^*=-\Omega_1(y,s)$. Meanwhile, the matrix $\Omega_1(y,s)$ can also be derived by the equations:
      \begin{equation}\label{eq:p1-dark-1}
      \begin{split}
       &\lim_{\lambda\to\lambda_1}\frac{\Phi_1^{\dag}(y,s;\lambda_1)\Sigma_3\Phi_1(y,s;\lambda)}{\lambda-\lambda_1}\\
       =& \int_{(y_0,s_0)}^{(y,s)}\frac{\ii}{4}\Phi_1^{\dag}(y',s';\lambda_1)\Phi_1(y',s';\lambda_1) \dd s'+\lambda_1^{-2}\Phi_1^{\dag}(y',s';\lambda_1)\left(\ii \rho(y',s')\mathbb{I}_4+ \mathbf{V}_{0,y'}\right)\Phi_1(y',s';\lambda_1)\dd y'\,.
      \end{split}
      \end{equation}
Recently, Rybkin considered the binary Darboux transformation for the KdV equation to remove or add the discrete spectrum in the framework of Riemann-Hilbert representation \cite{Rybkin-21}.
\end{remark}
\begin{remark}Let $\mathbf{C}$ be a contour in the complex plane that satisfies the Schwartz symmetry (for instance \cite{BilmanM17}, $\mathbf{C}=(-\infty,-r)\cup\{z:|z|=r\}\cup(r,\infty)$) and
suppose $\Phi(y,s;\lambda)$ is an analytic matrix on $\mathbb{S}^2\setminus(\{0\}\cup\{\infty\}\cup \mathbf{C})$, which satisfies the jump condition $\Phi_+(y,s;\lambda)=\Phi_-(y,s;\lambda)\mathbf{V}(\lambda)$ across the contour $\mathbf{C}$.  The derivative of $\Phi_{\pm}(y,s;\lambda)$ with respect to $y$ and the one with respect to $s$ satisfies the same jump condition as $\Phi$, so $\frac{\partial}{\partial y}\Phi(y,s;\lambda)\Phi^{-1}(y,s;\lambda)$ and $\frac{\partial}{\partial s}\Phi(y,s;\lambda)\Phi^{-1}(y,s;\lambda)$ are holomorphic functions on $\mathbb{S}^2\setminus\{0,\infty\}$.
Thus, from the proof of {\bf Theorem} \ref{thm1}, we can see that the normalized condition for the wave function $\Phi(y,s;\lambda)$ with $\Phi(0,0;\lambda)=\mathbb{I}$ is not necessary. It is also an obvious fact from the viewpoint of Riemann-Hilbert problems.
\end{remark}

For the Darboux matrix, we have the following proposition:
\begin{proposition}\label{prop1}
\begin{itemize}
  \item The determinant of Darboux matrix is $\det(\mathbf{T}(y,s;\lambda))=\frac{(\lambda-\lambda_1)(\lambda+\lambda_1^*)}{(\lambda-\lambda_1^*)(\lambda+\lambda_1)}.$
  \item Suppose the Darboux matrix is expanded in the neighborhood of infinity as: $\mathbf{T}(y,s;\lambda)=\mathbb{I}+\mathbf{T}_1(y,s)\lambda^{-1}+O(\lambda^{-2})$; then we have $\mathrm{Tr}(\mathbf{T}_1(y,s))=2(\lambda_1^*-\lambda_1)$.
\end{itemize}
\end{proposition}

Next, we consider the multi-fold Darboux matrix and the higher order Darboux matrix, which are iterations of the elementary Darboux matrix in {\bf Theorem} \ref{thm1}. Suppose there are $N$ different vector solutions $\Phi_i(y,s)=\Phi(y,s;\lambda_i)\mathbf{v}_i$ and $\lambda_i$, the multi-fold Darboux matrix can be constructed as follows:
\begin{equation}\label{eq:n-fold-dt}
\mathbf{T}[N](y,s;\lambda)=\mathbf{T}_N(y,s;\lambda)\mathbf{T}_{N-1}(y,s;\lambda)\cdots \mathbf{T}_{1}(y,s;\lambda)
\end{equation}
where
\[
\mathbf{T}_{i}(y,s;\lambda)=\mathbb{I}-\left[\Phi_i[i-1],\Lambda\Phi_i[i-1]^*\right]\left(\frac{\Phi_i[i-1]^{\dag}\Sigma\Phi_i[i-1]}{\lambda_i-\lambda_i^*}\right)^{-1}
\begin{bmatrix}
\frac{\Phi_i[i-1]^{\dag}}{\lambda-\lambda_i^*}\\[5pt]
-\frac{\Phi_i[i-1]^{\T}\Lambda}{\lambda+\lambda_i} \\
\end{bmatrix}\Sigma
\]
and
\[
\Phi_i[i-1](y,s;\lambda_i)=\mathbf{T}_{i-1}(y,s;\lambda_i)\mathbf{T}_{i-2}(y,s;\lambda_i)\cdots \mathbf{T}_{1}(y,s;\lambda_i)\Phi_i(y,s;\lambda_i),\,\,\,\,\, i\geq 2,\,\,\,\,
\]
$\Phi_1[0](y,s;\lambda_1)=\Phi_1(y,s)$.  {From standard complex analysis considerations}, the Darboux matrix $\mathbf{T}[N](y,s;\lambda)$ can be decomposed in the following form:
\begin{equation}\label{eq:n-dt-frac}
\mathbf{T}[N](y,s;\lambda)=\sum_{i=1}^{N}\left(\frac{1}{\lambda-\lambda_i^*}\frac{1}{2\pi \ii}\oint_{\gamma_{\lambda_i^*}}\mathbf{T}[N](y,s;\lambda) {\rm d}\lambda+\frac{1}{\lambda+\lambda_i}\frac{1}{2\pi \ii}\oint_{\gamma_{-\lambda_i}}\mathbf{T}[N](y,s;\lambda) {\rm d}\lambda\right)
\end{equation}
where $\gamma_{\lambda_i^*}$, $\gamma_{-\lambda_i}$ denote a small circle around the point $\lambda_i^*$ and $-\lambda_i$ respectively, with anti-clockwise orientation. The residue can be given explicitly:
\begin{equation}
\begin{split}
 &\frac{1}{2\pi \ii}\oint_{\gamma_{\lambda_i^*}}\mathbf{T}[N](y,s;\lambda) {\rm d}\lambda\\
 =&(\lambda_i^*-\lambda_i)\mathbf{T}_N(\lambda_i^*)\mathbf{T}_{N-1}(\lambda_i^*)\cdots \mathbf{T}_{i+1}(\lambda_i^*)\frac{\Phi_i[i-1]\Phi_i[i-1]^{\dag}\Sigma}{\Phi_i[i-1]^{\dag}\Sigma\Phi_i[i-1]}\mathbf{T}_{i-1}(\lambda_i^*)\cdots\mathbf{T}_{1}(\lambda_i^*) \\
\end{split}
\end{equation}
and
\begin{equation}
\begin{split}
 &\frac{1}{2\pi \ii}\oint_{\gamma_{-\lambda_i}}\mathbf{T}[N](y,s;\lambda) {\rm d}\lambda\\
 =&(\lambda_i-\lambda_i^*)\mathbf{T}_N(\lambda_i^*)\mathbf{T}_{N-1}(\lambda_i^*)\cdots \mathbf{T}_{i+1}(\lambda_i^*)\frac{\Lambda\Phi_i[i-1]^*\Phi_i[i-1]^{\T}\Lambda\Sigma}{\Phi_i[i-1]^{\dag}\Sigma\Phi_i[i-1]}\mathbf{T}_{i-1}(\lambda_i^*)\cdots\mathbf{T}_{1}(\lambda_i^*), \\
\end{split}
\end{equation}
both of which have rank one. Then the multi-fold Darboux matrix $\mathbf{T}[N](y,s;\lambda)$ can be rewritten in the following form:
\begin{equation}
\mathbf{T}[N](y,s;\lambda)=\mathbb{I}_4-\sum_{i=1}^{N}\left(\frac{|x_{2i-1}\rangle\langle y_{2i-1}|}{\lambda-\lambda_i^*}+\frac{|x_{2i}\rangle\langle y_{2i}|}{\lambda+\lambda_i}\right)\Sigma.
\end{equation}
The inverse Darboux matrix can be rewritten as
\begin{equation}
\mathbf{T}[N]^{-1}(y,s;\lambda)=\Sigma\mathbf{T}[N]^{\dag}(y,s;\lambda^*)\Sigma=\mathbb{I}_4-\sum_{i=1}^{N}\left(\frac{|y_{2i-1}\rangle\langle x_{2i-1}|}{\lambda-\lambda_i}+\frac{|y_{2i}\rangle\langle x_{2i}|}{\lambda+\lambda_i^*}\right)\Sigma,
\end{equation}
where $|y_{k}\rangle^{\dag}=\langle y_{k}|$ and $|x_{k}\rangle^{\dag}=\langle x_{k}|$, $k=1,2,\cdots, 2N$.
Since $\mathbf{T}[N](y,s;\lambda)\mathbf{T}[N]^{-1}(y,s;\lambda)=\mathbb{I}_4$, we know that
\begin{equation}
\begin{split}
\frac{1}{2\pi \ii}\oint_{\gamma_{\lambda_i}}\mathbf{T}[N](y,s;\lambda)\mathbf{T}[N]^{-1}(y,s;\lambda) {\rm d}\lambda&=\mathbf{T}[N](y,s;\lambda_i)|y_{2i-1}\rangle \langle x_{2i-1}|\Sigma=0,\\
\frac{1}{2\pi \ii}\oint_{\gamma_{-\lambda_i^*}}\mathbf{T}[N](y,s;\lambda)\mathbf{T}[N]^{-1}(y,s;\lambda) {\rm d}\lambda&=\mathbf{T}[N](y,s;-\lambda_i^*)|y_{2i}\rangle \langle x_{2i}|\Sigma=0.
\end{split}
\end{equation}
{It is easy to show}
\begin{equation}\label{eq:kernel}
\mathbf{T}[N](y,s;\lambda_i)|y_{2i-1}\rangle=0,\,\,\,\,\,\, \mathbf{T}[N](y,s;-\lambda_i^*)|y_{2i}\rangle=0,\,\,\,\,\, i=1,2,\cdots, N.
\end{equation}
On the other hand, we know the kernel of the multi-Darboux matrix is such that ${\rm Ker}(\mathbf{T}(s,y;\lambda_i))={\rm span}\{\Phi_i\}$ and ${\rm Ker}(\mathbf{T}(s,y;-\lambda_i^*))={\rm span}\{\Lambda\Phi_i^*\}$. Thus the vectors $|y_{2i-1}\rangle$ and $|y_{2i}\rangle$ can be chosen as $\Phi_i$ and $\Lambda\Phi_i^*$ respectively. {By using equation \eqref{eq:kernel}, one obtains}
\begin{equation}\label{eq:kernel1}
\left[|x_1\rangle, |x_2\rangle, \cdots, |x_{2N-1}\rangle, |x_{2N}\rangle\right]=\left[\Phi_1, \Lambda\Phi_1^*, \cdots, \Phi_N, \Lambda\Phi_N^*\right]\mathbf{M}_{N}^{-1},
\end{equation}
where $\mathbf{M}_{N}$ is the block matrix
\begin{equation}\label{eq:Mn}
\mathbf{M}_{N}=\left(\begin{bmatrix}
                        \frac{\Phi_i^{\dag}\Sigma\Phi_j}{\lambda_j-\lambda_i^*} & \frac{\Phi_i^{\dag}\Lambda\Sigma\Phi_j^*}{-\lambda_j^*-\lambda_i^*} \\[8pt]
                        \frac{-\Phi_i^{\T}\Lambda\Sigma\Phi_j}{\lambda_j+\lambda_i} & \frac{\Phi_i^{\T}\Sigma\Phi_j^*}{-\lambda_j^*+\lambda_i} \\
                     \end{bmatrix}
\right)_{1\leq i,j\leq N}.
\end{equation}

Finally, the $N$-fold Darboux matrix $\mathbf{T}[N](y,s;\lambda)$ is given by the following theorem:
\begin{theorem}\label{thm2}
The $N$-fold Darboux matrix
\begin{equation}\label{eq:n-fold-dt-1}
\mathbf{T}[N](y,s;\lambda)=\mathbb{I}_4-\mathbf{Y}
\mathbf{M}_{N}^{-1}(\lambda\mathbb{I}_{2N}-\mathbf{D}_{2N})^{-1}\mathbf{Y}^{\dag}\Sigma
\end{equation}
converts the system \eqref{eq:laxcd} into the system with new potential functions $\rho[N](y,s)$ and $\mathbf{Q}[N]$, where
\[
\mathbf{Y}=\left[\Phi_1,\Lambda\Phi_1^*,\cdots,\Phi_N,\Lambda\Phi_N^*\right]=\begin{bmatrix}
                                                                               \mathbf{Y}_1 \\
                                                                               \mathbf{Y}_2 \\
                                                                             \end{bmatrix}
,\,\,\,\,\,\,\, \mathbf{D}_{2N}={\rm diag}\left(\lambda_1^*,-\lambda_1,\cdots,\lambda_N^*,-\lambda_N\right),
\]
$\mathbf{Y}_{1}$ and $\mathbf{Y}_{2}$ are $2\times (2N)$ matrices, the matrix $\mathbf{M}_{N}$ is given in \eqref{eq:Mn}. The B\"acklund transformation between old potential functions and new ones is given by \begin{equation}\label{eq:backlund1}
\begin{split}
\mathbf{Q}[N]&=\mathbf{Q}-\mathbf{Y}_2\mathbf{M}_N^{-1}\mathbf{Y}_1^{\dag},\\
\rho[N]&=\rho-\ln_{y,s}\left(\det(\mathbf{M}_N)\right),\\
|q_1[N]|^2+|q_2[N]|^2&=|q_1|^2+|q_2|^2+2\kappa\ln_{ss}\left(\det(\mathbf{M}_N)\right).
\end{split}
\end{equation}
\end{theorem}
\begin{proof}
The multi-fold Darboux matrix $\mathbf{T}[N](y,s;\lambda)$ can be derived from the theory of standard Darboux transformation. Here we just prove the B\"acklund transformation \eqref{eq:backlund1}. By the theory of the Darboux matrix, we know that
\[
\mathbf{T}[N]_{y}+\mathbf{T}[N]\mathbf{U}=\mathbf{U}[N]\mathbf{T}[N],\,\,\,\,\,\,\, \mathbf{T}[N]_{s}+\mathbf{T}[N]\mathbf{V}=\mathbf{V}[N]\mathbf{T}[N].
\]
If we expand the multi-fold Darboux matrix $\mathbf{T}[N](y,s;\lambda)$ in the neighborhood of $\infty$
\[
\mathbf{T}[N]=\mathbb{I}_4+\mathbf{T}[N]^{[1]}(y,s)\lambda^{-1}+O(\lambda^{-2}),\,\,\,\,\,\,\, \mathbf{T}[N]^{[1]}(y,s)=-\mathbf{Y}
\mathbf{M}_{N}^{-1}\mathbf{Y}^{\dag}\Sigma,
\]
we then have
\begin{equation}
\begin{split}
\mathbf{Q}[N]&=\mathbf{Q}-\mathbf{Y}_2\mathbf{M}_N^{-1}\mathbf{Y}_1^{\dag},\\
\rho[N]&=\rho+\ii \left(\mathbf{T}[N]^{[1]}_{1,1}\right)_y,
\end{split}
\end{equation}
where $\mathbf{T}[N]^{[1]}_{1,1}$ represents the $(1,1)$ element of the matrix $\mathbf{T}[N]^{[1]}$. Actually, from the symmetry relation $\sigma_3\mathbf{Y}_1
\mathbf{M}_{N}^{-1}\mathbf{Y}_1^{\dag}\sigma_3=\mathbf{Y}_1
\mathbf{M}_{N}^{-1}\mathbf{Y}_1^{\dag}$ and $-\sigma_2\mathbf{Y}_1
\mathbf{M}_{N}^{-1}\mathbf{Y}_1^{\dag}\sigma_2=\mathbf{Y}_1
\mathbf{M}_{N}^{-1}\mathbf{Y}_1^{\dag}$, where $\sigma_3={\rm diag}(1,-1)$, it follows that $-\mathbf{Y}_1
\mathbf{M}_{N}^{-1}\mathbf{Y}_1^{\dag}=\mathbf{T}[N]^{[1]}_{1,1}\mathbb{I}_2$. Suppose there exists a matrix solution, which can be expanded in the deleted neighborhood of $\infty$ on $\mathbb{S}^2$:
\begin{equation}\label{eq:theta-psi}
\begin{bmatrix}
\Theta\\
\Psi\\
\end{bmatrix}=\left(\begin{bmatrix}
                      \mathbb{I}_2 \\
                      0 \\
                    \end{bmatrix}+\sum_{i=1}^{\infty}\Phi_i
\right)\ee^{\frac{\ii}{4}\lambda s\mathbb{I}_2}.
\end{equation}
{Substituting the solution \eqref{eq:theta-psi} into the Lax pair \eqref{eq:laxcd}, one obtains}
\[
\left(\mathbf{Q}^{\dag}\Psi\Theta^{-1}\right)_s=\mathbf{Q}_s^{\dag}\Psi\Theta^{-1}+\frac{\ii}{2}(|q_1|^2+|q_2|^2)\mathbb{I}_2
-\frac{\ii\lambda}{2}\mathbf{Q}^{\dag}\Psi\Theta^{-1}-\frac{\ii\kappa}{2}(\mathbf{Q}^{\dag}\Psi\Theta^{-1})^2,
\]
which yields the following expansion
\[
\mathbf{Q}^{\dag}\Psi\Theta^{-1}=(|q_1|^2+|q_2|^2)\mathbb{I}_2 \lambda^{-1}+O(\lambda^{-2}).
\]
From the Lax pair \eqref{eq:laxcd}, we have
\begin{equation}\label{eq:thetas}
\Theta_s=\left(\frac{\ii}{4}\lambda\mathbb{I}_2+\frac{\ii\kappa}{2}(|q_1|^2+|q_2|^2)\mathbb{I}_2 \lambda^{-1}+O(\lambda^{-2})\right)\Theta.
\end{equation}
Inserting the asymptotic expansion
\[
\Theta=\left(\mathbb{I}_2+\sum_{i=1}^{\infty}\Theta^{[i]}\lambda^{-i}\right)\ee^{\frac{\ii}{4}\lambda s\mathbb{I}_2}
\]
{in \eqref{eq:theta-psi} into the above equation \eqref{eq:thetas}, one arrives at}
\[
\frac{\ii\kappa}{2}(|q_1|^2+|q_2|^2)\mathbb{I}_2=\left[\Theta^{[1]}_{1,1}\mathbb{I}_2\right]_s
\]
where $\Theta^{[1]}_{1,1}$ represents the $(1,1)$ element of matrix $\Theta^{[1]}$.
If we apply the multi-fold Darboux matrix on the matrix solution \eqref{eq:theta-psi}, a similar result can be obtained:
\[
\frac{\ii\kappa}{2}(|q_1[N]|^2+|q_2[N]|^2)\mathbb{I}_2=\left[\left(\mathbf{T}[N]^{[1]}_{1,1}+\Theta^{[1]}_{1,1}\right)\mathbb{I}_2\right]_s.
\]
Thus we have
\[
\frac{\ii\kappa}{2}(|q_1[N]|^2+|q_2[N]|^2)\mathbb{I}_2=\frac{\ii\kappa}{2}(|q_1|^2+|q_2|^2)\mathbb{I}_2+\left[\mathbf{T}[N]^{[1]}_{1,1}\mathbb{I}_2\right]_s.
\]
{What is left} is to prove that the expression $\mathbf{T}[N]^{[1]}_{1,1}$ can be rewritten in a compact form. {Due to the fact that}
\[
\Phi_{i,s}=\mathbf{V}(y,s;\lambda_i)\Phi_i,\,\,\,\,\, -\Phi_{j,s}^{\dag}\Sigma=\Phi_j^{\dag}\Sigma\mathbf{V}(y,s;\lambda_j^*),
\]
we have
\[
-4\ii\left(\frac{\Phi_i\Sigma\Phi_j^{\dag}}{\lambda_i-\lambda_j^*}\right)_s=\Phi_i\Sigma\Sigma_3\Phi_j^{\dag}.
\]
Together with {\bf Proposition} \ref{prop1}, we have
\[
\left[\mathbf{T}[N]^{[1]}_{1,1}\right]_s=\frac{1}{4}\left[2\mathbf{T}[N]^{[1]}_{1,1}-2\mathbf{T}[N]^{[1]}_{3,3}\right]_s
=\frac{1}{4}\left[-\mathbf{Y}_1
\mathbf{M}_{N}^{-1}\mathbf{Y}_1^{\dag}+\kappa \mathbf{Y}_2
\mathbf{M}_{N}^{-1}\mathbf{Y}_2^{\dag}\right]_s=\ii\ln_{ss}\det(\mathbf{M}_N).
\]
{Similarly, we can derive}
\[
\left[\mathbf{T}[N]^{[1]}_{1,1}\right]_y=\ii\ln_{ys}\det(\mathbf{M}_N),
\]
which {completes the proof}.
\end{proof}

Before considering the higher order Darboux matrix and the general one, we rewrite the matrix $\mathbf{M}_N$ \eqref{eq:Mn} in the following form
\begin{equation}
\mathbf{M}_N=\mathbf{K}_N\mathbf{S}_{N}\mathbf{K}_N^{\dag}
\end{equation}
where
\[
\mathbf{K}_N=\begin{bmatrix}
\Phi_1^{\dag} &0&0&0&\cdots&0&0 \\
0&-\Phi_1^{\T}\Lambda &0&0&\cdots&0&0 \\
\vdots &\vdots &\vdots&\vdots&\cdots&\vdots&\vdots \\
0&0 &0&0&\cdots&\Phi_N^{\dag}&0 \\
0&0 &0&0&\cdots&0&-\Phi_N^{\T}\Lambda \\
\end{bmatrix},\,\,\,\,\,\,\, \mathbf{S}_{N}=\left(\begin{bmatrix}
                        \frac{\Sigma}{\lambda_j-\lambda_i^*} & \frac{\Sigma}{-\lambda_j^*-\lambda_i^*} \\[8pt]
                        \frac{\Sigma}{\lambda_j+\lambda_i} & \frac{\Sigma}{-\lambda_j^*+\lambda_i} \\
                     \end{bmatrix}
\right)_{1\leq i,j\leq N}.
\]
Suppose we have the following expansions:
\begin{equation}
\Phi_k=\sum_{i=0}^{\infty}\Phi_k^{[i]}(\lambda-\lambda_k)^{i},\,\,\,\,\,\,\, k=1,2,\cdots, l,
\end{equation}
through the standard limit technique for the multi-fold Darboux matrix \cite{guo12}, the general Darboux matrix can be obtained:
\begin{equation}\label{eq:general-dt}
\mathbf{T}[N](y,s;\lambda)=\mathbb{I}_4-\mathbf{Y}
\mathbf{M}_{N}^{-1}\mathbf{D}_{2N}\mathbf{Y}^{\dag}\Sigma,
\end{equation}
where
\begin{equation}\label{eq:general-y12}
\mathbf{Y}=\left[\mathbf{Y}^{[1]},\mathbf{Y}^{[2]},\cdots, \mathbf{Y}^{[l]}\right]=\begin{bmatrix}
\mathbf{Y}_1\\
\mathbf{Y}_2\\
\end{bmatrix},\,\,\,\,\,\,\,\,
\mathbf{Y}^{[k]}=\left[\Phi_k^{[0]},\Lambda\Phi_k^{[0]*},\cdots, \Phi_k^{[n_k]},\Lambda\Phi_k^{[n_k]*}\right],
\end{equation}
and
\[
\mathbf{D}_{2N}={\rm diag}\left(\mathbf{D}_{1},\mathbf{D}_{2},\cdots,\mathbf{D}_{l}\right),\,\,\,\,\,\,
\mathbf{D}_{k}=\begin{bmatrix}
\frac{1}{\lambda-\lambda_k^*} &0& \cdots &0&0 \\
0&\frac{1}{\lambda+\lambda_k}&\cdots &0&0 \\
\cdots \\
\frac{1}{(\lambda-\lambda_k^*)^{s_k+1}}&0&\cdots&\frac{1}{(\lambda-\lambda_k^*)}&0 \\
0&\frac{(-1)^{s_k}}{(\lambda+\lambda_k)^{s_k+1}} &\cdots &0&\frac{1}{\lambda+\lambda_k}\\
\end{bmatrix},
\]
\begin{equation}\label{eq:general-mn}
\mathbf{M}_N=\hat{\mathbf{K}}_N\mathbf{S}_{N}\hat{\mathbf{K}}_N^{\dag}.
\end{equation}
{Here $\sum_{k=1}^{l}s_k=N$, and}
\[
\hat{\mathbf{K}}_N=\begin{bmatrix}
\mathbf{K}_1 &0&\cdots&0 \\
0& \mathbf{K}_2 &\cdots&0 \\
\vdots &\vdots &\cdots&\vdots \\
0&0&\cdots&\mathbf{K}_l \\
\end{bmatrix},\,\,\,\,\,\,\,\,\, \mathbf{S}_{N}=\left(\mathbf{S}_{i,j}\right)_{1\leq i,j\leq l},
\]
\[
\mathbf{K}_k=\begin{bmatrix}
\Phi_k^{[0]\dag} &0&\cdots&0&0 \\
0&-\Phi_k^{[0]\T}\Lambda &\cdots&0&0 \\
\vdots &\vdots &\cdots&\vdots&\vdots \\
\Phi_k^{[s_k-1]\dag}&0&\cdots&\Phi_k^{[0]\dag}&0 \\
0&-\Phi_k^{[s_k-1]\T}\Lambda&\cdots&0&-\Phi_k^{[0]\T}\Lambda \\
\end{bmatrix},\]
\[ \mathbf{S}_{i,j}=\begin{bmatrix}
\binom{0}{0}\frac{\Sigma}{\lambda_j-\lambda_i^*} &\binom{0}{0}\frac{\Sigma}{-\lambda_j^*-\lambda_i^*}&\cdots&\binom{s_j-1}{0}\frac{(-1)^{s_j-1}\Sigma}{(\lambda_j-\lambda_i^*)^{s_j}}
&\binom{s_j-1}{0}\frac{\Sigma}{(-\lambda_j^*-\lambda_i^*)^{s_j}} \\
\binom{0}{0}\frac{\Sigma}{\lambda_j+\lambda_i} & \binom{0}{0}\frac{\Sigma}{-\lambda_j^*+\lambda_i} &\cdots&\binom{s_j-1}{0}\frac{(-1)^{s_j-1}\Sigma}{(\lambda_j+\lambda_i)^{s_j}} & \binom{s_j-1}{0}\frac{\Sigma}{(-\lambda_j^*+\lambda_i)^{s_j}} \\
\vdots &\vdots &\ddots&\vdots&\vdots \\
\binom{s_i-1}{s_i-1}\frac{\Sigma}{(\lambda_j-\lambda_i^*)^{s_i}} &\binom{s_i-1}{s_i-1}\frac{\Sigma}{(-\lambda_j^*-\lambda_i^*)^{s_i}}&\cdots&\binom{s_i+s_j-2}{s_i-1}\frac{(-1)^{s_j-1}\Sigma}{(\lambda_j-\lambda_i^*)^{s_i+s_j-1}}&\binom{s_i+s_j-2}{s_i-1}\frac{\Sigma}{(-\lambda_j^*-\lambda_i^*)^{s_i+s_j-1}} \\
\binom{s_i-1}{s_i-1}\frac{(-1)^{s_i-1}\Sigma}{(\lambda_j+\lambda_i)^{s_i}} & \binom{s_i-1}{s_i-1}\frac{(-1)^{s_i-1}\Sigma}{(-\lambda_j^*+\lambda_i)^{s_i}} &\cdots&\binom{s_i+s_j-2}{s_i-1}\frac{(-1)^{s_i+s_j-2}\Sigma}{(\lambda_j+\lambda_i)^{s_i+s_j-1}} & \binom{s_i+s_j-2}{s_i-1}\frac{(-1)^{s_i-1}\Sigma}{(-\lambda_j^*+\lambda_i)^{s_i+s_j-1}} \\
\end{bmatrix}.
\]
The B\"acklund transformation has the same form as in \eqref{eq:backlund1} with new matrices $\mathbf{Y}_1$, $\mathbf{Y}_2$ and $\mathbf{M}_N$ given in \eqref{eq:general-y12} and \eqref{eq:general-mn}.

For the nonsingular finite gap solution \cite{Belokos94}, there exists an analytic matrix solution $\Phi(y,s;\lambda)$ normalised $\Phi(0,0;\lambda)=\mathbb{I}_4$  by the {\bf Proposition} 2.1 \cite{BilmanM17}. So we can construct the localized wave solutions on the background of finite gap solutions. In this work, we merely consider the background of zero and plane wave solutions.

Therefore, through the above theorems and analysis, the regular solutions of the CCSP equation \eqref{eq:ccsp} can be represented as
\begin{equation}\label{eq:backlund1-csp}
\begin{split}
\mathbf{Q}[N]&=\mathbf{Q}-\mathbf{Y}_2\mathbf{M}_N^{-1}\mathbf{Y}_1^{\dag},\\
\rho[N]&=\rho-\ln_{y,s}\left(\det(\mathbf{M}_N)\right)>0,\\
x&=\int_{(y_0,s_0)}^{(y,s)}\rho(y',s')\dd y'-\frac{\kappa}{2}(|q_1(y',s')|^2+|q_2(y',s')|^2)\dd s'-\ln_{s}\left(\det(\mathbf{M}_N)\right),\,\,\,\, t=-s,
\end{split}
\end{equation}
where $\mathbf{Y}_i$ ($i=1,2$) and $\mathbf{M}_N$ are given in \eqref{eq:general-y12} and \eqref{eq:general-mn}, respectively.

\section{Soliton solution with VBC}\label{sec3}
In this section, we consider how to derive the soliton and multi-soliton solutions with VBC. Since the exact solution is singular in the defocusing case, we only consider the focusing case $\kappa=1$. Exact solutions can be constructed directly by the Darboux transformation, and the soliton interaction can be analyzed directly from the multi-soliton solution. In order to get a better understanding for the solutions from a spectral viewpoint, we give a simple inverse scattering analysis for this problem \cite{AblowitzC91}.

\subsection{Scattering and inverse scattering analysis}
Consider the spectral problem
\begin{equation}\label{eq:spec}
\Phi_y=\frac{1}{\lambda}\begin{bmatrix}
-\ii\rho\mathbb{I}_2& -\mathbf{Q}_y^{\dag}\\
\mathbf{Q}_y&\ii \rho\mathbb{I}_2\\
\end{bmatrix}\Phi,
\end{equation}
with the boundary condition $\rho\to \frac{\delta}{2}>0$, $\mathbf{Q}\to \mathbf{0}$ as $y\to \pm\infty.$ Suppose the potential functions satisfy the following conditions
\begin{equation}\label{eq:bc-zero}
\int_{\mathbb{R}}\left|\rho-\frac{\delta}{2}\right|{\rm d}y< \infty,\,\,\,\,\,\, \int_{\mathbb{R}}\left(\left|q_{1,y}\right|+\left|q_{2,y}\right|\right){\rm d}y<\infty.
\end{equation}
Performing a simple transformation $\Phi=\mathbf{m}\exp\left[-\ii\frac{\delta}{2\lambda}\Sigma_3y\right]$ into the spectral problem \eqref{eq:spec}, we obtain
\begin{equation}\label{eq:spec-1}
\mathbf{m}_y=-\frac{\ii\delta}{2\lambda}\left(\Sigma_3\mathbf{m}-\mathbf{m}\Sigma_3\right)+\frac{1}{\lambda}\Delta \mathbf{U}\mathbf{m},\,\,\,\,\,\,\, \Delta\mathbf{U}(y)=\begin{bmatrix}
-\ii(\rho-\frac{\delta}{2})\mathbb{I}_2& -\mathbf{Q}_y^{\dag}\\[5pt]
\mathbf{Q}_y&\ii (\rho-\frac{\delta}{2})\mathbb{I}_2\\
\end{bmatrix},
\end{equation}
which yields the following integral equations
\begin{equation}\label{eq:int-eq}
\mathbf{m}_{\pm}(y;\lambda)=\mathbb{I}_4+\frac{1}{\lambda}\int_{\pm\infty}^{y}\ee^{\frac{\ii\delta}{2\lambda}(y'-y)\Sigma_3}\Delta\mathbf{U}(y')\mathbf{m}_{\pm}(y';\lambda)\ee^{-\frac{\ii\delta}{2\lambda}(y'-y)\Sigma_3}\dd y'.
\end{equation}
By the standard Neumann series method,  $\mathbf{m}^{up}=(\mathbf{m}_{1,-},\mathbf{m}_{2,+})$ is analytic in the upper half plane $\mathbb{C}^{+}$; $\mathbf{m}^{lw}=(\mathbf{m}_{1,+},\mathbf{m}_{2,-})$ is analytic in the lower half plane $\mathbb{C}^{-}$, where $\mathbf{m}_{\pm,1}$, $\mathbf{m}_{\pm,2}$ are the first two and last two columns of the matrices $\mathbf{m}_{\pm}$ respectively. $\Phi_+$ and $\Phi_-$ are fundamental solutions of the linear system of differential equations \eqref{eq:spec}, and therefore we can define the scattering matrix $\mathbf{S}(\lambda)$:
\begin{equation}\label{eq:scattering}
\Phi_+(y,s;\lambda)=\Phi_-(y,s;\lambda)\mathbf{S}(\lambda),\,\,\,\,\,\,\, \Phi_{\pm}(y,s;\lambda)=
\mathbf{m}_{\pm}(y;\lambda)\exp\left[-\ii\frac{\delta}{2\lambda}\Sigma_3y\right],
\end{equation}
which yields $\det(\mathbf{S}(\lambda))=1.$

The symmetry relations for the solution $\Phi_{\pm}(x,t;\lambda)$ can be obtained by the symmetries of the solutions $\Phi_\pm$. Since
\begin{equation}\label{eq:sym-phi1}
\Phi_{\pm}(y,s;\lambda)\Phi_{\pm}^{\dag}(y,s;\lambda^*)=\mathbb{I}_4,
\end{equation}
and
\begin{equation}\label{eq:sym-phi2}
\Lambda\Phi_{\pm}(y,s;\lambda)\Lambda^{-1}=\Phi_{\pm}^{*}(y,s;-\lambda^*)
\end{equation}
the scattering matrix satisfies
\begin{equation}\label{eq:sym-conj}
\mathbf{S}(\lambda)\mathbf{S}^{\dag}(\lambda^*)=\mathbb{I}_4,
\end{equation}
and
\begin{equation}\label{eq:sym-inver}
\Lambda\mathbf{S}(\lambda)\Lambda^{-1}=\mathbf{S}^*(-\lambda^*).
\end{equation}
The above symmetry relations \eqref{eq:sym-conj} and \eqref{eq:sym-inver} allow to write
\begin{equation}
\mathbf{S}(\lambda)=\begin{bmatrix}
\mathbf{a}(\lambda) & \mathbf{c}(\lambda)  \\
\mathbf{b}(\lambda)& \mathbf{d}(\lambda) \\
\end{bmatrix},\,\,\,\,\,\, \mathbf{S}^{-1}(\lambda)=\mathbf{S}^{\dag}(\lambda^*)=\begin{bmatrix}
\mathbf{a}^{\dag}(\lambda^*) & \mathbf{b}^{\dag}(\lambda^*)  \\
\mathbf{c}^{\dag}(\lambda^*)& \mathbf{d}^{\dag}(\lambda^*) \\
\end{bmatrix}
\end{equation}
where
\begin{equation}
\begin{split}
 \mathbf{a}(\lambda)&=\begin{bmatrix}
a_{1}(\lambda) & -a_{2}^*(-\lambda^*) \\
a_{2}(\lambda) & a_{1}^*(-\lambda^*) \\
\end{bmatrix},\,\,\,\,\,\,\, \mathbf{b}(\lambda)=\begin{bmatrix}
b_{1}(\lambda) & -b_{2}^*(-\lambda^*) \\
b_{2}(\lambda) & b_{1}^*(-\lambda^*) \\
\end{bmatrix},\\
\mathbf{c}(\lambda)&=\begin{bmatrix}
c_{1}(\lambda) & -c_{2}^*(-\lambda^*) \\
c_{2}(\lambda) & c_{1}^*(-\lambda^*) \\
\end{bmatrix},\,\,\,\,\,\,\,\mathbf{d}(\lambda)=\begin{bmatrix}
d_{1}(\lambda) & -d_{2}^*(-\lambda^*) \\
d_{2}(\lambda) & d_{1}^*(-\lambda^*) \\
\end{bmatrix}.
\end{split}
\end{equation}
Next, we derive the jump condition and the scattering data. {To this end}, we assume that the determinant $\det(\mathbf{a}(\lambda))$ on the real line is nonzero. {If the determinant $\det(\mathbf{a}(\lambda))$ is zero somewhere on the real line, the corresponding point is called a spectral singularity \cite{Zhou89}. Thus, the purpose of this assumption is to avoid spectral singularities}.

Taking the limit $x\to -\infty$ in \eqref{eq:scattering}, we obtain that
\begin{equation}\label{eq:a-d}
\mathbf{a}(\lambda)=\lim_{y\to -\infty} \mathbf{m}_{1,1,+}(y,s;\lambda),\,\,\,\,\,\,\,\, \mathbf{d}(\lambda)=\lim_{y\to -\infty} \mathbf{m}_{2,2,+}(y,s;\lambda)
\end{equation}
where $\mathbf{m}_{1,1,+}(y,s;\lambda)$ denotes the first two rows of $\mathbf{m}_{1,+}(y,s;\lambda)$, $\mathbf{m}_{2,2,+}(y,s;\lambda)$ denotes the last two rows of $\mathbf{m}_{2,+}(y,s;\lambda)$. Then equation \eqref{eq:a-d} shows that $\mathbf{a}(\lambda)$ is analytic in the lower half plane, and $\mathbf{d}(\lambda)$ is analytic in the upper half plane. {We then} normalize the analytic matrices $\mathbf{m}^{\rm up}$ or $\mathbf{m}^{\rm lw}$ with unimodular:
\begin{equation}\label{eq:m+}
\mathbf{m}\equiv\mathbf{m}^+=\mathbf{m}^{up}{\rm diag}\left(\mathbb{I}_2,[\mathbf{d}(\lambda)]^{-1}\right)=\left(
\mathbf{m}_{1,-},\mathbf{m}_{2,-}\right)\begin{bmatrix}
\mathbb{I}_2& \mathbf{c}(\lambda)[\mathbf{d}(\lambda)]^{-1}\ee^{-\frac{\delta\ii}{\lambda}y\mathbb{I}_2}\\
0 & \mathbb{I}_2 \\
\end{bmatrix},\,\,\,\,\,\, \lambda\in \mathbb{C}^+
\end{equation}
and
\begin{equation}\label{eq:m-}
\mathbf{m}\equiv\mathbf{m}^-=\mathbf{m}^{lw}{\rm diag}\left([\mathbf{a}(\lambda)]^{-1},\mathbb{I}_2\right)=\left(
\mathbf{m}_{1,-},\mathbf{m}_{2,-}\right)\begin{bmatrix}
\mathbb{I}_2&0\\
\mathbf{b}(\lambda)[\mathbf{a}(\lambda)]^{-1} \ee^{\frac{\delta\ii}{\lambda}y\mathbb{I}_2} & \mathbb{I}_2 \\
\end{bmatrix},\,\,\,\,\,\, \lambda\in \mathbb{C}^-.
\end{equation}
Thus, the jump condition on the real line is
\begin{equation}
\mathbf{m}^+=\mathbf{m}^-\begin{bmatrix}
\mathbb{I}_2 &\mathbf{c}(\lambda)[\mathbf{d}(\lambda)]^{-1}\ee^{-\frac{\delta\ii}{\lambda}y\mathbb{I}_2} \\
-\mathbf{b}(\lambda)[\mathbf{a}(\lambda)]^{-1}\ee^{\frac{\delta\ii}{\lambda}y\mathbb{I}_2} &\mathbb{I}_2-\mathbf{b}(\lambda)[\mathbf{a}(\lambda)]^{-1}\mathbf{c}(\lambda)[\mathbf{d}(\lambda)]^{-1} \\
\end{bmatrix},\,\,\,\,\,\,\,\,\, \lambda\in\mathbb{R}.
\end{equation}
Moreover, due to the symmetry relation $\mathbf{S}^{\dag}(\lambda^*)\mathbf{S}(\lambda)=\mathbb{I}_4$, we obtain $\mathbf{a}^{\dag}(\lambda^*)\mathbf{c}(\lambda)+\mathbf{b}^{\dag}(\lambda^*)\mathbf{d}(\lambda)=0$, which implies $\mathbf{c}(\lambda)\mathbf{d}^{-1}(\lambda)=-(\mathbf{b}(\lambda^*)\mathbf{a}^{-1}(\lambda^*))^{\dag}$. Thus the jump condition can be simplified as
\begin{equation}
\mathbf{m}^+=\mathbf{m}^-\begin{bmatrix}
\mathbb{I}_2 &\mathbf{r}^{\dag}(\lambda)\ee^{-\frac{\delta\ii}{\lambda}y\mathbb{I}_2} \\
\mathbf{r}(\lambda)\ee^{\frac{\delta\ii}{\lambda}y\mathbb{I}_2} &\mathbb{I}_2+\mathbf{r}(\lambda)\mathbf{r}^{\dag}(\lambda) \\
\end{bmatrix},\,\,\,\,\,\,\,\,\, \lambda\in\mathbb{R},
\end{equation}
where \[\mathbf{r}(\lambda)=-\mathbf{b}(\lambda)\mathbf{a}^{-1}(\lambda)=\begin{bmatrix}
r_1(\lambda) & -r_2(-\lambda) \\
r_2(\lambda) & r_1(-\lambda) \\
\end{bmatrix}.
 \]
Here $\mathbf{r}(\lambda)$ is called the matrix reflection coefficient.

\subsection{Scattering data and scattering map}

To proceed to {the case of} the discrete spectrum, we assume the elements of scattering matrix $\mathbf{a}(\lambda)$ and $\mathbf{d}(\lambda)$ can be analytically extended to the corresponding complex plane.
{Since} \[\det(\mathbf{S}(\lambda))=\det(\mathbf{a}\mathbf{d})\det
\left(\mathbb{I}_2-\mathbf{b}\mathbf{a}^{-1}\mathbf{c}\mathbf{d}^{-1}\right)=\det(\mathbf{a}(\lambda)\mathbf{a}
^{\dag}(\lambda^*))\det
\left(\mathbb{I}_2+\mathbf{r}(\lambda)\mathbf{r}^{\dag}(\lambda^*)\right),\]
from the  fact of  $\mathbf{S}(\lambda)=1$, one has $|\alpha(\lambda)|^2=[\det\left(\mathbb{I}_2+\mathbf{r}(\lambda)\mathbf{r}^{\dag}(\lambda^*)\right)]^{-1}$, where $\alpha(\lambda)=\det(\mathbf{a}(\lambda)).$

In what follows, we consider the discrete scattering data following the ideas in \cite{FaddeevT87}. Since the scattering matrix $\mathbf{a}(\lambda)$ is $2\times 2$, there are two different cases: one is single zeros of $\alpha(\lambda)$; the otherone is double zeros of $\alpha(\lambda)$. Firstly we consider the single zeros of $\det(\mathbf{a}(\lambda))$. Assume that the determinant of $\det(\mathbf{a}(\lambda))$ has the simple zeros $\lambda_1$ and $-\lambda_1^*$ located in the lower half plane, since \[\det\left(\left[\Phi_{1,+}(y,s;\lambda),\Phi_{2,-}(y,s;\lambda)\right]\right)=
\alpha(\lambda)\equiv\frac{(\lambda-\lambda_1)(\lambda+\lambda_1^*)}
{(\lambda-\lambda_1^*)(\lambda+\lambda_1)}\widehat{\alpha(\lambda)},
\] the eigenfunctions satisfy the relation
\begin{equation}\label{eq:proportion}
\begin{split}
\Phi_{1,+}(y,s;\lambda_1)\begin{bmatrix}
\alpha_1(s;\lambda_1)\\
\alpha_2(s;\lambda_1)\\
\end{bmatrix}&=\Phi_{2,-}(y,s;\lambda_1)\begin{bmatrix}
\alpha_3(s;\lambda_1)\\
\alpha_4(s;\lambda_1)\\
\end{bmatrix},\\
\Phi_{1,+}(y,s;-\lambda_1^*)\begin{bmatrix}
-\alpha_2^*(s;\lambda_1)\\
\alpha_1^*(s;\lambda_1)\\
\end{bmatrix}&=\Phi_{2,-}(y,s;-\lambda_1^*)\begin{bmatrix}
-\alpha_4^*(s;\lambda_1)\\
\alpha_3^*(s;\lambda_1)\\
\end{bmatrix},
\end{split}
\end{equation}
where $\alpha_i$, $i=1,2,3,4$ are the coefficients of proportionality.  In order to get a closed form, we need another linear relation which comes from the degenerate property of the matrix $\mathbf{a}(\lambda)$. We consider the meromorphic function
\[
\Phi^-(y,s;\lambda)=\left[\frac{\Phi_{1,+}(y,s;\lambda){\rm adj}(\mathbf{a}(\lambda))}{\det(\mathbf{a}(\lambda))},\Phi_{2,-}(y,s;\lambda)\right]
\]
which can be expanded by
\[
\Phi^-(y,s;\lambda)=\frac{\left[\mathbf{K}_1^{[-1]}(y,s),0\right]}{\lambda-\lambda_1}
+\left[\mathbf{K}_1^{[0]}(y,s),\mathbf{K}_2^{[0]}(y,s)\right]+\left[\mathbf{K}_1^{[1]}(y,s),\mathbf{K}_2^{[1]}(y,s)\right](\lambda-\lambda_1)+\cdots.
\]
Since $\det(\Phi^-(y,s;\lambda))=1$, we have
\begin{equation}
\begin{split}
\det\left(\left[\mathbf{K}_1^{[-1]}(y,s),\mathbf{K}_2^{[0]}(y,s)\right]\right)&=\underset{\lambda=\lambda_1}{\rm Res}\left(\frac{1}{\alpha(\lambda)}\right)
\det\left(\left[\Phi_{1,+}(y,s;\lambda_1),\Phi_{2,-}(y,s;\lambda_1)\right]\right)
\det\left({\rm adj}(\mathbf{a}(\lambda_1))\right)\\
&=\underset{\lambda=\lambda_1}{\rm Res}\left(\frac{1}{\alpha(\lambda)}\right)\det\left(\mathbf{a}(\lambda_1)\right)
\det\left({\rm adj}(\mathbf{a}(\lambda_1))\right)=0,
\end{split}
\end{equation}
which implies that $\left[\mathbf{K}_1^{[-1]}(y,s),\mathbf{K}_2^{[0]}(y,s)\right]$ has a second order zero at the point $\lambda=\lambda_1$.  The second order zero is a consequence of the matrix $\mathbf{a}(\lambda)$ being degenerate. The corresponding linear relations are \[
\mathbf{K}_1^{[-1]}(y,s)\begin{bmatrix}
\beta_1(\lambda_1)\\
\beta_2(\lambda_1)\\
\end{bmatrix}
=0,\]
and
\[
\mathbf{K}_1^{[-1]}(y,s)\left(\mathbf{a}(\lambda_1)\begin{bmatrix}
\alpha_1(\lambda_1)\\
\alpha_2(\lambda_1)\\
\end{bmatrix}\right)=\mathbf{K}_2^{[0]}(y,s)\begin{bmatrix}
\alpha_3(s;\lambda_1)\\
\alpha_4(s;\lambda_1)\\
\end{bmatrix}
\]
where the nonzero property of vector $\mathbf{a}(\lambda_1)\begin{bmatrix}
\alpha_1(\lambda_1)\\
\alpha_2(\lambda_1)\\
\end{bmatrix}\ne 0$ can be derived from the asymptotic behavior of $\Phi_{1,+}(y,s;\lambda)$.
We have a similar relation at the point $\lambda=-\lambda_1^*$.

Secondly, we consider the double zeros of $\alpha(\lambda),$ i.e. \[\alpha(\lambda)=\frac{(\lambda-\lambda_1)^2(\lambda+\lambda_1^*)^2}{
(\lambda+\lambda_1)^2(\lambda-\lambda_1^*)^2}\widehat{\alpha(\lambda)}.\]
The matrix function $\mathbf{a}(\lambda)$ can be decomposed in the following form
\[
\mathbf{a}(\lambda)=\frac{(\lambda-\lambda_1)(\lambda+\lambda_1^*)}{(\lambda+\lambda_1)(\lambda-\lambda_1^*)}\widehat{\mathbf{a}(\lambda)},
\]
where $\widehat{\mathbf{a}(\lambda)}$ is a non-degenerate matrix in the lower half plane.
In this case, there are two eigenfunctions at $\lambda=\lambda_1$, {which satisfy the following relations}
\begin{equation}\label{eq:proportion1}
\begin{split}
\Phi_{1,+}(y,s;\lambda_1)\begin{bmatrix}
\alpha_1(s;\lambda_1)\\
\alpha_2(s;\lambda_1)\\
\end{bmatrix}&=\Phi_{2,-}(y,s;\lambda_1)\begin{bmatrix}
\alpha_3(s;\lambda_1)\\
\alpha_4(s;\lambda_1)\\
\end{bmatrix},\\
\Phi_{1,+}(y,s;\lambda_1)\begin{bmatrix}
\beta_1(s;\lambda_1)\\
\beta_2(s;\lambda_1)\\
\end{bmatrix}&=\Phi_{2,-}(y,s;\lambda_1)\begin{bmatrix}
\beta_3(s;\lambda_1)\\
\beta_4(s;\lambda_1)\\
\end{bmatrix}\,,\\
\Phi_{1,+}(y,s;-\lambda_1^*)\begin{bmatrix}
-\alpha_2^*(s;\lambda_1)\\
\alpha_1^*(s;\lambda_1)\\
\end{bmatrix}&=\Phi_{2,-}(y,s;-\lambda_1^*)\begin{bmatrix}
-\alpha_4^*(s;\lambda_1)\\
\alpha_3^*(s;\lambda_1)\\
\end{bmatrix},\\
\Phi_{1,+}(y,s;-\lambda_1^*)\begin{bmatrix}
-\beta_2^*(s;\lambda_1)\\
\beta_1^*(s;\lambda_1)\\
\end{bmatrix}&=\Phi_{2,-}(y,s;-\lambda_1^*)\begin{bmatrix}
-\beta_4^*(s;\lambda_1)\\
\beta_3^*(s;\lambda_1)\\
\end{bmatrix},
\end{split}
\end{equation}
where $\alpha_i$, $\beta_i$, $i=1,2,3,4$ are the coefficients of proportionality, and the vectors $(\alpha_1,\alpha_2,\alpha_3,\alpha_4)$ and $(\beta_1,\beta_2,\beta_3,\beta_4)$ are linear independent. {To find a closed form}, we need another linear relation which comes from the degeneracy of the matrix $\mathbf{a}(\lambda)$. In this case, we consider the meromorphic function
\[
\Phi^-(y,s;\lambda)=\left[\Phi_{1,+}(y,s;\lambda)\widehat{\mathbf{a}(\lambda)}^{-1}\frac{
(\lambda+\lambda_1)(\lambda-\lambda_1^*)}{(\lambda-\lambda_1)(\lambda+\lambda_1^*)},
\Phi_{2,-}(y,s;\lambda)\right]
\]
which can be expanded by
\[
\Phi^-(y,s;\lambda)=\frac{\left[\mathbf{K}_1^{[-1]}(y,s),0\right]}{\lambda-\lambda_1}
+\left[\mathbf{K}_1^{[0]}(y,s),\mathbf{K}_2^{[0]}(y,s)\right]+\left[\mathbf{K}_1^{[1]}(y,s),\mathbf{K}_2^{[1]}(y,s)\right](\lambda-\lambda_1)+\cdots.
\]
Since $\det(\Phi^-(y,s;\lambda))=1$,  one has
\begin{equation}
\begin{split}
\det\left(\left[\mathbf{K}_1^{[-1]}(y,s),\mathbf{K}_2^{[0]}(y,s)\right]\right)&=\frac{
2\lambda_1(\lambda_1-\lambda_1^*)}{\lambda_1+\lambda_1^*}\left(\det\left(\widehat{\mathbf{a}(\lambda_1)}\right)\right)^{-1}
\det\left(\left[\Phi_{1,+}(y,s;\lambda_1),\Phi_{2,-}(y,s;\lambda_1)\right]\right)\\
&=\frac{
2\lambda_1(\lambda_1-\lambda_1^*)}{\lambda_1+\lambda_1^*}\left(\det\left(\widehat{\mathbf{a}(\lambda_1)}\right)\right)^{-1}\det\left(\mathbf{a}(\lambda_1)\right)
=0,
\end{split}
\end{equation}
which implies that $\left[\mathbf{K}_1^{[-1]}(y,s),\mathbf{K}_2^{[0]}(y,s)\right]$ still possesses a second order zero at the point $\lambda=\lambda_1$. The double zero point is  a consequence of the degeneracy of the matrix $\mathbf{a}(\lambda).$ The corresponding linear relations between them are
\[
\frac{\lambda_1+\lambda_1^*}{2\lambda_1(\lambda_1-\lambda_1^*)}\mathbf{K}_1^{[-1]}(y,s) \left(\widehat{\mathbf{a}(\lambda_1)}\begin{bmatrix}
\alpha_1(\lambda_1)\\
\alpha_2(\lambda_1)\\
\end{bmatrix}\right)=\mathbf{K}_2^{[0]}(y,s)\begin{bmatrix}
\alpha_3(s;\lambda_1)\\
\alpha_4(s;\lambda_1)\\
\end{bmatrix}
\]
and
\[
\frac{\lambda_1+\lambda_1^*}{2\lambda_1(\lambda_1-\lambda_1^*)}\mathbf{K}_1^{[-1]}(y,s) \left(\widehat{\mathbf{a}(\lambda_1)}\begin{bmatrix}
\beta_1(\lambda_1)\\
\beta_2(\lambda_1)\\
\end{bmatrix}\right)=\mathbf{K}_2^{[0]}(y,s)\begin{bmatrix}
\beta_3(s;\lambda_1)\\
\beta_4(s;\lambda_1)\\
\end{bmatrix}\,.
\]
{Similar relations exist at the point $\lambda=-\lambda_1^*$, which are omitted here}.

To establish the symmetry relation, we rewrite the relation of scattering matrix in the following form to find the discrete spectrum:
\begin{equation}
\left[\Phi^{\pm}(y,s;\lambda^*)\right]^{\dag}\Phi^{\mp}(y,s;\lambda)=\mathbb{I}_4, \,\,\,\,\,\,\, \Phi^{\pm}(y,s;\lambda)=\mathbf{m}^{\pm}(y,s;\lambda)\exp\left[-\ii\frac{\delta}{2\lambda}\Sigma_3y \right].
\end{equation}
By the way, we obtain the determinant relation $\det(\mathbf{d}(\lambda))=\det(\mathbf{a}^{\dag}(\lambda^*)).$
At the points $\lambda=\lambda_1^*,-\lambda_1$, there are  coefficients of proportionality between $\Phi_{1,-}$ and $\Phi_{2,+}$. It is easy to see that the meromorphic solution $\Phi^+(y,s;\lambda)$ in the upper half plane can be determined from the symmetry relation $\left[\Phi^-(y,s;\lambda^*)\right]^{\dag}\Phi^+(y,s;\lambda)=\Phi^+(y,s;\lambda)\left[\Phi^-(y,s;\lambda^*)\right]^{\dag}=\mathbb{I}_4.$
The matrix functions $\Phi^{+}(y,s;\lambda)$ and $\left[\Phi^-(y,s;\lambda^*)\right]^{\dag}$ can be expanded at the singular point $\lambda=\lambda_1$ by
\begin{equation}
\begin{split}
 \Phi^{+}(y,s;\lambda)=&\left[\Phi_0^{+}(y,s;\lambda_1^*)+\Phi_1^{+}(y,s;\lambda_1^*)(\lambda-\lambda_1^*)+\cdots\right] {\rm diag}\left(1,\frac{1}{\lambda-\lambda_1^*}\right) \\
 \left[\Phi^-(y,s;\lambda^*)\right]^{\dag}=&{\rm diag}\left(\frac{1}{\lambda-\lambda_1^*},1\right)\left[\left[\Phi_0^{-}(y,s;\lambda_1)\right]^{\dag}+\left[\Phi_1^{-}(y,s;\lambda_1)\right]^{\dag}(\lambda-\lambda_1^*)+\cdots\right] \\
\end{split}
\end{equation}
which implies that
\begin{equation}\label{eq:im-ker-eq}
\begin{split}
\Phi_0^{+}(y,s;\lambda_1^*)\left[\Phi_0^{-}(y,s;\lambda_1)\right]^{\dag}=\left[\Phi_0^{-}(y,s;\lambda_1)\right]^{\dag}\Phi_0^{+}(y,s;\lambda_1^*)&=0,\\
\Phi_0^{+}(y,s;\lambda_1^*)\left[\Phi_1^{-}(y,s;\lambda_1)\right]^{\dag}+\Phi_1^{+}(y,s;\lambda_1^*)\left[\Phi_0^{-}(y,s;\lambda_1)\right]^{\dag}&\\
=\left[\Phi_0^{-}(y,s;\lambda_1)\right]^{\dag}\Phi_1^{+}(y,s;\lambda_1^*)+\left[\Phi_1^{-}(y,s;\lambda_1)\right]^{\dag}\Phi_0^{+}(y,s;\lambda_1^*)&=\mathbb{I}_4.
\end{split}
\end{equation}
{From \eqref{eq:im-ker-eq}, we have}
\begin{equation}\label{eq:im-ker-1}
{\rm Im}\left(\left[\Phi_0^{-}(y,s;\lambda_1)\right]^{\dag}\right)={\rm Ker}\left(\Phi_0^{+}(y,s;\lambda_1^*)\right),\,\,\,\,\,\,\,\, {\rm Ker}\left(\left[\Phi_0^{-}(y,s;\lambda_1)\right]^{\dag}\right)={\rm Im}\left(\Phi_0^{+}(y,s;\lambda_1^*)\right).
\end{equation}
{Similarly, we obtain}
\begin{equation}
{\rm Im}\left(\Phi_0^{-}(y,s;\lambda_1^*)\right)={\rm Ker}\left(\left[\Phi_0^{+}(y,s;\lambda_1)\right]^{\dag}\right),\,\,\,\,\,\,\,\, {\rm Ker}\left(\Phi_0^{-}(y,s;\lambda_1)\right)={\rm Im}\left(\left[\Phi_0^{+}(y,s;\lambda_1^*)\right]^{\dag}\right).
\end{equation}
We conclude that ${\rm Ker}\left(\Phi_0^{-}(y,s;\lambda_1)\right)\perp {\rm Ker}\left(\Phi_0^{+}(y,s;\lambda_1^*)\right)$. Actually, for arbitrary vectors $\mathbf{u}\in {\rm Ker}\left(\Phi_0^{-}(y,s;\lambda_1)\right)$ and $\mathbf{v}\in {\rm Ker}\left(\Phi_0^{+}(y,s;\lambda_1^*)\right)$, on account of equation \eqref{eq:im-ker-1}, we obtain $\mathbf{v}^{\dag}\mathbf{u}=\mathbf{w}^{\dag}\Phi_0^{-}(y,s;\lambda_1)\mathbf{u}=0.$ On the other hand
 \begin{equation}
\begin{split}
&{\rm dim}\left({\rm Ker}\left(\Phi_0^-(y,s;\lambda_1)\right)\right)+{\rm dim}\left({\rm Im}\left(\Phi_0^-(y,s;\lambda_1)\right)\right)\\
=&{\rm dim}\left({\rm Ker}\left(\Phi_0^-(y,s;\lambda_1)\right)\right)+{\rm dim}\left({\rm Im}\left(\left[\Phi_0^-(y,s;\lambda_1)\right]^{\dag}\right)\right) \\
=&{\rm dim}\left({\rm Ker}\left(\Phi_0^-(y,s;\lambda_1)\right)\right)+{\rm dim}\left({\rm Ker}\left(\Phi_0^-(y,s;\lambda_1^*)\right)\right)=4
\end{split}
\end{equation}
which implies
\begin{equation}
{\rm Ker}\left(\Phi_0^{-}(y,s;\lambda_1)\right)\oplus {\rm Ker}\left(\Phi_0^{+}(y,s;\lambda_1^*)\right)=\mathbb{C}^4.
\end{equation}
The spaces ${\rm Ker}\left(\Phi_0^{-}(y,s;\lambda_1)\right)$ and ${\rm Ker}\left(\Phi_0^{+}(y,s;\lambda_1^*)\right)$ represent the discrete scattering data.

Before the discussion of multiple zeros case, we introduce the generalized eigenfunctions. 
For $i=1,\dots,n$, functions $\phi_i\in L^2(\mathbb{R})$ satisfying the equations
\begin{equation}
\begin{split}
(\mathbf{L}-\lambda_1)\phi&=0, \\
(\mathbf{L}-\lambda_1)\phi_{1}&=\phi\\
\cdots \\
(\mathbf{L}-\lambda_1)\phi_n&=\phi_{n-1}\\
\end{split}
\end{equation}
are called generalized eigenfunctions. Actually, it is obvious to see that generalized eigenfunctions can be obtained by taking derivatives of an eigenfunction $\phi$ with respect to $\lambda$, i.e.
\begin{equation}\label{eq:generalized-function}
\phi_i=\frac{{\rm d}^i\phi}{{\rm d}\lambda^i}|_{\lambda=\lambda_1}\in L^2(\mathbb{R}).
\end{equation}

In general, {under the assumption}
\[\alpha(\lambda)=\prod_{i=1}^{n}\frac{(\lambda-\lambda_i)^{k_i}(\lambda+\lambda_i^*)^{k_i}}{(\lambda+\lambda_i)^{k_i}(\lambda-\lambda_i^*)^{k_i}}\widehat{\alpha(\lambda)},
\,\,\,\,\,\,\, k_i=\nu_i+\tau_i\] the generalized eigenfunctions can be determined by the following relations:
\begin{equation}\label{eq:gene-eigen}
\begin{split}
\frac{{\rm d}^{j_i}\Phi_{1,+}(y,s;\lambda)}{{\rm d}\lambda^{j_i}}\begin{bmatrix}
\alpha_1(s;\lambda)\\
\alpha_2(s;\lambda)\\
\end{bmatrix}\big|_{\lambda=\lambda_i}&=\frac{{\rm d}^{j_i}\Phi_{2,-}(y,s;\lambda)}{{\rm d}\lambda^{j_i}}\begin{bmatrix}
\alpha_3(s;\lambda)\\
\alpha_4(s;\lambda)\\
\end{bmatrix}\big|_{\lambda=\lambda_i},\\
\frac{{\rm d}^{l_i}\Phi_{1,+}(y,s;\lambda)}{{\rm d}\lambda^{l_i}}\begin{bmatrix}
\beta_1(s;\lambda)\\
\beta_2(s;\lambda)\\
\end{bmatrix}\big|_{\lambda=\lambda_i}&=\frac{{\rm d}^{l_i}\Phi_{2,-}(y,s;\lambda)}{{\rm d}\lambda^{l_i}}\begin{bmatrix}
\beta_3(s;\lambda)\\
\beta_4(s;\lambda)\\
\end{bmatrix}\big|_{\lambda=\lambda_i},
\end{split}
\end{equation}
where $j_i=1,2,\cdots, \nu_i-1$; $l_i=1,2,\cdots, \tau_i-1$; $\alpha_{\nu}(s;\lambda_i)$, $\beta_{\nu}(s;\lambda_i)$ and their derivative with respect to $\lambda$, $\nu=1,2,3,4$ are the coefficients of proportionality. The relations \eqref{eq:gene-eigen} imply that the functions $$\frac{{\rm d}^k\Phi_{1,+}(y,s;\lambda)}{{\rm d}\lambda^k}\big|_{\lambda=\lambda_i},\,\,\,\,\,\,\,\frac{{\rm d}^k\Phi_{2,-}(y,s;\lambda)}{{\rm d}\lambda^k}\big|_{\lambda=\lambda_i} $$
tend to zero as $y\to\pm\infty$. So they are generalized eigenfunctions. The other generalized eigenfunctions can be defined by the symmetry relationships \eqref{eq:sym-phi1} and \eqref{eq:sym-phi2}.

Thus, in the absence of the spectral singularities, the scattering map can be represented as
\begin{equation}
(q_1(y,0),q_2(y,0))\to \left\{\mathbf{r}(\lambda),\,\,\lambda\in\mathbb{R}; \pm\lambda_i,\pm\lambda_i^*, \frac{\dd^{j_i}}{\dd \lambda^{j_i}}\alpha_l(0;\lambda)|_{\lambda=\lambda_i},\, \frac{\dd^{j_i}}{\dd \lambda^{l_i}}\beta_l(0;\lambda)|_{\lambda=\lambda_i}\right\},
\end{equation} $i=1,2,\cdots,n$, $l=1,2,3,4$, $j_i=0,1,\cdots, \nu_i-1$, $l_i=0,1,\cdots, \tau_i-1.$

\subsection{Evolution of scattering data}
 We consider the evolution of scattering data. Suppose we have the Jost functions $\Phi_{\pm}(y,s;\lambda)$ for each $s$ and the potential function $\mathbf{Q}$ decaying to zero when $y\to \pm\infty$, then $\mathbf{W}_{\pm}(y,s;\lambda)=\Phi_{\pm}(y,s;\lambda)\exp\left(\frac{\ii}{4}\lambda s\Sigma_3\right)$
satisfies the Lax pair \eqref{eq:laxcd}. Actually, since $\Phi_{\pm}(y,s;\lambda)$ satisfies the spectral problem \eqref{eq:spec}, we can write $\mathbf{W}_{\pm}(y,s;\lambda)=\Phi_{\pm}(y,s;\lambda)\mathbf{C}^{\pm}(s;\lambda)$. Inserting $\mathbf{W}_{\pm}(y,s;\lambda)$ into evolution part of the Lax pair \eqref{eq:laxcd}, we obtain $\mathbf{C}_s^{\pm}(s;\lambda)=\frac{\ii}{4}\lambda\Sigma_3\mathbf{C}^{\pm}(s;\lambda)$ by the decay properties of the potential functions $\mathbf{Q}$, which implies that $\mathbf{C}^{\pm}(s;\lambda)=\exp\left(\frac{\ii}{4}\lambda\Sigma_3s\right)$ to keep the normalization at $\pm\infty$. Thus \[\frac{{\rm d}}{{\rm d}s}\Phi_{\pm}(y,s;\lambda)=-\frac{\ii}{4}\lambda\Phi_{\pm}(y,s;\lambda)\Sigma_3+\mathbf{V}(y,s;\lambda)\Phi_{\pm}(y,s;\lambda),\] which yields the evolution of the scattering matrix \[\frac{{\rm d}}{{\rm d}s}\mathbf{S}(s;\lambda)=\frac{{\rm d}}{{\rm d}s}\left(\Phi_{-}^{-1}(y,s;\lambda)\Phi_{+}(y,s;\lambda)\right)=\frac{\ii}{4}\lambda\left[\Sigma_3,\mathbf{S}(s;\lambda)\right]\]
i.e. $\mathbf{a}(s;\lambda)=\mathbf{a}(0;\lambda)$, $\mathbf{d}(s;\lambda)=\mathbf{d}(0;\lambda)$, $\mathbf{b}(s;\lambda)=\exp(-\frac{\ii}{2}\lambda s\mathbb{I}_2)\mathbf{b}(0;\lambda)$ and $\mathbf{c}(s;\lambda)=\exp(\frac{\ii}{2}\lambda s\mathbb{I}_2)\mathbf{c}(s;\lambda).$ So the reflection coefficients $\mathbf{r}(s;\lambda)=\exp(-\frac{\ii}{2}\lambda s\mathbb{I}_2)\mathbf{r}(0;\lambda)$ and the zeros of determinants $\mathbf{a}(\lambda)$ and $\mathbf{d}(\lambda)$ are invariant since the matrix $\mathbf{a}(\lambda)$ and $\mathbf{d}(\lambda)$ can be analytically extended in the lower/upper half plane respectively.

Now we consider the evolution of the coefficients of proportionality. Firstly, we have
\[
\frac{{\rm d}}{{\rm d}s}\Phi_{1,+}(y,s;\lambda)=-\frac{\ii}{4}\lambda \Phi_{1,+}(y,s;\lambda)+\mathbf{V}(y,s;\lambda)\Phi_{1,+}(y,s;\lambda),
\]
and
\[
\frac{{\rm d}}{{\rm d}s}\Phi_{2,-}(y,s;\lambda)=\frac{\ii}{4}\lambda \Phi_{2,-}(y,s;\lambda)+\mathbf{V}(y,s;\lambda)\Phi_{2,-}(y,s;\lambda).
\]
Taking the derivative of \eqref{eq:proportion} with respect to $s$
\begin{equation}\frac{{\rm d}}{{\rm d}s}\begin{bmatrix}
\alpha_1(s;\lambda_i)\\
\alpha_2(s;\lambda_i)\\
\end{bmatrix}=0,\,\,\,\,\,\,\frac{{\rm d}}{{\rm d}s}\begin{bmatrix}
\alpha_3(s;\lambda_i)\\
\alpha_4(s;\lambda_i)\\
\end{bmatrix}=-\frac{\ii}{2}\lambda_i\begin{bmatrix}
\alpha_3(s;\lambda_i)\\
\alpha_4(s;\lambda_i)\\
\end{bmatrix},
\end{equation}
i.e.
\begin{equation}
\begin{bmatrix}
\alpha_1(s;\lambda_i)\\
\alpha_2(s;\lambda_i)\\
\end{bmatrix}=\begin{bmatrix}
\alpha_1(0;\lambda_i)\\
\alpha_2(0;\lambda_i)\\
\end{bmatrix},\,\,\,\,\,\,\begin{bmatrix}
\alpha_3(s;\lambda_i)\\
\alpha_4(s;\lambda_i)\\
\end{bmatrix}=\ee^{-\frac{\ii}{2}\lambda_is}\begin{bmatrix}
\alpha_3(0;\lambda_i)\\
\alpha_4(0;\lambda_i)\\
\end{bmatrix}.
\end{equation}
In general, we obtain the evolution of generalized scattering data:
\begin{equation}\label{eq:gene-eigen-t}
\begin{split}
\frac{{\rm d}^{j_i}\Phi_{1,+}(y,s;\lambda)}{{\rm d}\lambda^{j_i}}\begin{bmatrix}
\alpha_1(0;\lambda)\\
\alpha_2(0;\lambda)\\
\end{bmatrix}\big|_{\lambda=\lambda_i}&=\frac{{\rm d}^{j_i}\Phi_{2,-}(y,s;\lambda)\ee^{-\frac{\ii}{2}\lambda s}}{{\rm d}\lambda^{j_i}}\begin{bmatrix}
\alpha_3(0;\lambda)\\
\alpha_4(0;\lambda)\\
\end{bmatrix}\big|_{\lambda=\lambda_i},\\
\frac{{\rm d}^{l_i}\Phi_{1,+}(y,s;\lambda)}{{\rm d}\lambda^{l_i}}\begin{bmatrix}
\beta_1(0;\lambda)\\
\beta_2(0;\lambda)\\
\end{bmatrix}\big|_{\lambda=\lambda_i}&=\frac{{\rm d}^{l_i}\Phi_{2,-}(y,s;\lambda)\ee^{-\frac{\ii}{2}\lambda s}}{{\rm d}\lambda^{l_i}}\begin{bmatrix}
\beta_3(0;\lambda)\\
\beta_4(0;\lambda)\\
\end{bmatrix}\big|_{\lambda=\lambda_i},
\end{split}
\end{equation}
where $j_i=1,2,\cdots, \nu_i-1$; $l_i=1,2,\cdots, \tau_i-1$; $\alpha_{\nu}(0;\lambda_i)$, $\beta_{\nu}(0;\lambda_i)$ and their derivative with respect to $\lambda$, $\nu=1,2,3,4$ are the coefficients of proportionality.

\subsection{Inverse scattering}
Now we consider the inverse scattering transform to obtain the solutions of equations \eqref{eq:cd}. Firstly, we represent the above well-defined meromorphic function $\mathbf{m}(y,s;\lambda)$ in equations \eqref{eq:m+} and \eqref{eq:m-} together with the evolved scattering data as a Riemann-Hilbert problem \cite{Deift99}:
\begin{rhp}\label{eq:rhp1} Let $(y,s)\in\mathbb{R}^2$ be arbitrary parameters. The meromorphic function $\mathbf{m}(y,s;\lambda)$ possesses the following properties:
\begin{description}
               \item[Analyticity] The meromorphic function $\mathbf{m}(y,s;\lambda)=(\mathbf{m}_1(y,s;\lambda),\mathbf{m}_2(y,s;\lambda))$ is analytic in $\lambda$ for $\lambda\in \mathbb{C}\backslash \left(\mathbb{R}\cup\{\pm \lambda_i,\pm\lambda_i^*| i=1,2,\cdots,n\}\right)$;
               \item[Poles] The principal part of $\mathbf{m}_1(y,s;\lambda)$ in the lower half plane can be represented as
                   \begin{equation}\label{eq:prin1}
                   \mathbf{m}_1(y,s;\lambda)=\sum_{i=1}^{n}\sum_{k=0}^{j_i}\left(
                   \frac{\mathbf{m}_1^{[k]}(y,s;\lambda_i)}{(\lambda-\lambda_i)^{k+1}}+\frac{\mathbf{m}_1^{[k]}(y,s;-\lambda_i^*)}{(\lambda+\lambda_i^*)^{k+1}}\right)+\mathbf{m}_{1a}(y,s;\lambda);
                   \end{equation}
                   and the principal part of $\mathbf{m}_2(y,s;\lambda)$ in the upper half plane can be represented as
                   \begin{equation}\label{eq:prin2}
                   \mathbf{m}_2(y,s;\lambda)=\sum_{i=1}^{n}\sum_{k=0}^{j_i}\left(
                   \frac{\mathbf{m}_2^{[k]}(y,s;\lambda_i^*)}{(\lambda-\lambda_i^*)^{k+1}}+\frac{\mathbf{m}_2^{[k]}(y,s;-\lambda_i)}{(\lambda+\lambda_i)^{k+1}}\right)+\mathbf{m}_{2a}(y,s;\lambda);
                   \end{equation}
                   where $j_i=\max\{\nu_i,\tau_i\}$, $\mathbf{m}_{1a}$ and $\mathbf{m}_{2a}$ are the analytic part of $\mathbf{m}_{1}$ and $\mathbf{m}_{2}$ respectively; the principal part of $\mathbf{m}_{1}$ and $\mathbf{m}_{2}$ are linked by equations \eqref{eq:gene-eigen-t} and their symmetric equations with the aid of equations \eqref{eq:sym-phi1} and \eqref{eq:sym-phi2}.
               \item[Jump condition] The jump conditions on the real line \begin{equation}
\mathbf{m}^+(y,s;\lambda)=\mathbf{m}^-(y,s;\lambda)\mathbf{v}(y,s;\lambda),\,\,\,\,\,\,\, \mathbf{v}(y,s;\lambda)=\begin{bmatrix}
\mathbb{I}_2 &\mathbf{r}^{\dag}(\lambda)\ee^{-\ii(\frac{\delta}{\lambda}y-\frac{\lambda s}{2})\mathbb{I}_2} \\[5pt]
\mathbf{r}(\lambda)\ee^{\ii(\frac{\delta}{\lambda}y-\frac{\lambda s}{2})\mathbb{I}_2} &\mathbb{I}_2+\mathbf{r}(\lambda)\mathbf{r}^{\dag}(\lambda) \\
\end{bmatrix},\,\,\,\,\,\,\,\,\, \lambda\in\mathbb{R},
\end{equation}
               \item[Normalization] The normalization condition $\mathbf{m}(y,s;\lambda)\to \mathbb{I}_4,$ as $\lambda\rightarrow\infty.$
             \end{description}
\end{rhp}
The above Riemann-Hilbert problem can be solved by the following algebraic-integral equations
\begin{multline}\label{eq:rhp1-int}
\mathbf{m}(y,s;\lambda)=\mathbb{I}_4\\+\sum_{i=1}^{n}\sum_{k=0}^{j_i}\left(
                   \frac{\left(\mathbf{m}_1^{[k]}(y,s;\lambda_i),0\right)}{(\lambda-\lambda_i)^{k+1}}+\frac{\left(\mathbf{m}_1^{[k]}
                   (y,s;-\lambda_i^*),0\right)}{(\lambda+\lambda_i^*)^{k+1}}+\frac{\left(0,\mathbf{m}_2^{[k]}(y,s;\lambda_i^*)\right)}{(\lambda-\lambda_i^*)^{k+1}}
                   +\frac{\left(0,\mathbf{m}_2^{[k]}(y,s;-\lambda_i)\right)}{(\lambda+\lambda_i)^{k+1}}\right)\\
                   +\frac{1}{2\pi \ii}\int_{\mathbb{R}}\frac{\mathbf{m}^-(y,s;\zeta)(\mathbf{v}(y,s;\zeta)-\mathbb{I}_4){\rm d}\zeta}{\zeta-\lambda},
\end{multline}
the algebraic equations are given in equations \eqref{eq:gene-eigen-t} and their symmetry equations as given in the {\bf Riemann-Hilbert Problem} \ref{eq:rhp1}. Insert the expansion
\begin{equation}\label{eq:expan-m}
\mathbf{m}(y,s;\lambda)=\mathbb{I}_4+\mathbf{M}_1(y,s)\lambda^{-1}+O(\lambda^{-2})
\end{equation}
into equation \eqref{eq:spec-1}, the potential function can be recovered as
\begin{equation}\label{eq:pot-rec}
\begin{bmatrix}
-\ii(\rho-\frac{\delta}{2})\mathbb{I}_2& -\mathbf{Q}_y^{\dag}\\[5pt]
\mathbf{Q}_y&\ii (\rho-\frac{\delta}{2})\mathbb{I}_2\\
\end{bmatrix}=\frac{\dd}{\dd y}\mathbf{M}_1(y,s).
\end{equation}

In general, the algebraic-integral equations \eqref{eq:rhp1-int} do not have a closed form. {However, in a special case, i.e. the reflectionless case $\mathbf{r}(\lambda)=0$, the potential functions can be obtained explicitly from equations \eqref{eq:gene-eigen-t} and their symmetric equations}.

To solve the linear system, we use the ansatz from the Darboux transformation given in {\bf Section} \ref{sec2}. Since the reflection coefficient $\mathbf{r}(\lambda)=0$, the jump condition on the line is the identity matrix. The meromorphic function $\mathbf{m}(y,s;\lambda)$ is an analytic function with poles located at the points $\pm\lambda_i$ and $\pm\lambda_i^*$, $i=1,2,\cdots, n.$ From equations \eqref{eq:rhp1-int}, we see that the first two columns of the matrix $\mathbf{m}(y,s;\lambda)$ are different form of the last two columns of matrix $\mathbf{m}(y,s;\lambda)$. Thus, we redefine a new analytic matrix
\begin{equation}
\widetilde{\mathbf{m}}(y,s;\lambda)=\mathbf{m}(y,s;\lambda){\rm diag}\left(\mathbf{a}(\lambda),\mathbb{I}_2\right)
\end{equation}
which is analytic in the lower half plane and has the poles in the upper half plane. Actually, the new analytic matrix $\widetilde{\mathbf{m}}(y,s;\lambda)$ is consistent with the general Darboux matrix $\mathbf{T}[n](y,s;\lambda)$ \eqref{eq:general-dt} with the seed solution $q_1=q_2=0.$ In the following, we give the exact solutions and their dynamics.

\subsection{Single soliton solutions}
As shown in previous analysis, there are two different types of eigenvalues: simple zero of $\det(\mathbf{a}(\lambda))$; and double zeros of $\det(\mathbf{a}(\lambda))$. Indeed, each simple zero corresponds to a ${\rm su}(2)$-type soliton,
{and each double zero corresponds to a breather solution. }

To find the single soliton solution through the formula \eqref{eq:backlund}, we need to give the exact special solution $\Phi_1$ of the Lax pair at $\lambda=\lambda_1=a_1+\ii b_1$, $b_1<0$ and $q_1=q_2=0$, $\rho=\delta/2$:\begin{equation}\label{eq:phi1-sin}
\Phi_1=\begin{bmatrix}
\theta_1 \\
\chi_1 \\
\phi_1 \\
\psi_1 \\
\end{bmatrix}=\begin{bmatrix}
c_1\ee^{\ii\left(\frac{\lambda_1}{4}s-\frac{\delta}{2\lambda_1}y+\vartheta_1\right)} \\
c_2\ee^{\ii\left(\frac{\lambda_1}{4}s-\frac{\delta}{2\lambda_1}y+\vartheta_2\right)} \\
d_1\ee^{-\ii\left(\frac{\lambda_1}{4}s-\frac{\delta}{2\lambda_1}y+\phi_1\right)} \\
d_2\ee^{-\ii\left(\frac{\lambda_1}{4}s-\frac{\delta}{2\lambda_1}y+\phi_2\right)} \\
\end{bmatrix}.
\end{equation}
{For the sake of convenience}, we rewrite the B\"acklund transformation in the following form:
\begin{equation}\label{eq:bt-single}
\begin{split}
q_1[1]&=q_1-(\lambda_1-\lambda_1^*)\frac{\phi_1\theta_1^*+\psi_1^*\chi_1}{|\theta_1|^2|+\chi_1|^2+|\phi_1|^2+|\psi_1|^2},  \\
q_2[1]&=q_2-(\lambda_1-\lambda_1^*)\frac{\phi_1\chi_1^*-\psi_1^*\theta_1}{|\theta_1|^2|+\chi_1|^2+|\phi_1|^2+|\psi_1|^2},  \\
\rho[1]&=\rho-2 \ln_{y,s}\left(\frac{|\theta_1|^2+|\chi_1|^2+|\phi_1|^2+|\psi_1|^2}{\lambda_1-\lambda_1^*}\right). \\
\end{split}
\end{equation}
Inserting \eqref{eq:phi1-sin} into \eqref{eq:bt-single}, we obtain the single soliton solution
\begin{equation}\label{eq:single-soliton}
\begin{split}
\begin{bmatrix}
q_1[1]\\[8pt]
q_2[1]\\
\end{bmatrix}&=-\ii b_1 \,{\rm sech}(A_1) \,\mathbf{C}\begin{bmatrix}
\frac{d_2}{|\mathbf{d}|}\ee^{\ii(B_1-\phi_2)} \\[8pt]
\frac{d_1}{|\mathbf{d}|}\ee^{-\ii(B_1-\phi_1)} \\
\end{bmatrix},\,\,\, \mathbf{C}=\begin{bmatrix}
\frac{c_2\ee^{\ii\vartheta_2}}{|\mathbf{c}|} & \frac{c_1\ee^{-\ii\vartheta_1}}{|\mathbf{c}|} \\[8pt]
-\frac{c_1\ee^{\ii\vartheta_1}}{|\mathbf{c}|} & \frac{c_2\ee^{-\ii\vartheta_2}}{|\mathbf{c}|} \\
\end{bmatrix},\\
\rho[1]&=\delta\left(\frac{1}{2}-\frac{b_1^2}{a_1^2+b_1^2}{\rm sech}^2(A_1)\right),\\
x&=\frac{\delta}{2}y-b_1\tanh(A_1),\,\,\,\, t=-s,
\end{split}
\end{equation}
where $A_1=b_1\left(\frac{s}{2}+\frac{\delta y}{a_1^2+b_1^2}\right)+\ln|\mathbf{d}|-\ln|\mathbf{c}|$, $B_1=a_1\left(\frac{s}{2}-\frac{\delta y}{a_1^2+b_1^2}\right)$, $|\mathbf{c}|=\sqrt{c_1^2+c_2^2}$, $|\mathbf{d}|=\sqrt{d_1^2+d_2^2}$.
{The dynamics of a single soliton in Fig. \ref{su2} exhibits a beating effect \cite{Zhao18}.} The beating effect comes from the multi-component system with Hermitian symmetry, i.e. if $(q_1,q_2)$ is the solution, then $(q_1,q_2)\mathbf{H}$ is also solution, for any Hermitian matrix $\mathbf{H}$. It is easy to see that the matrix $\mathbf{C}$ is Hermitian. If $|a_1|>|b_1|$, the soliton is  smooth;
if $|a_1|=|b_1|$, the soliton is a cuspon-type; if $|a_1|<|b_1|$, the soliton is a loop-type. If $t\to\pm\infty$, then $x\pm b_1=\frac{\delta}{2}y$. The center of soliton is located on the curve $A_1=0$ which approaches the line $\frac{(x\pm b_1)}{a_1^2+b_1^2}-\frac{t}{4}+\frac{\ln|\mathbf{d}|-\ln|\mathbf{c}|}{2b_1}=0$ as $t\to \pm\infty$. Meanwhile, the velocity of the soliton tends to $\frac{a_1^2+b_1^2}{4}$.

\begin{figure}[tbh]
\centering
\includegraphics[height=50mm,width=140mm]{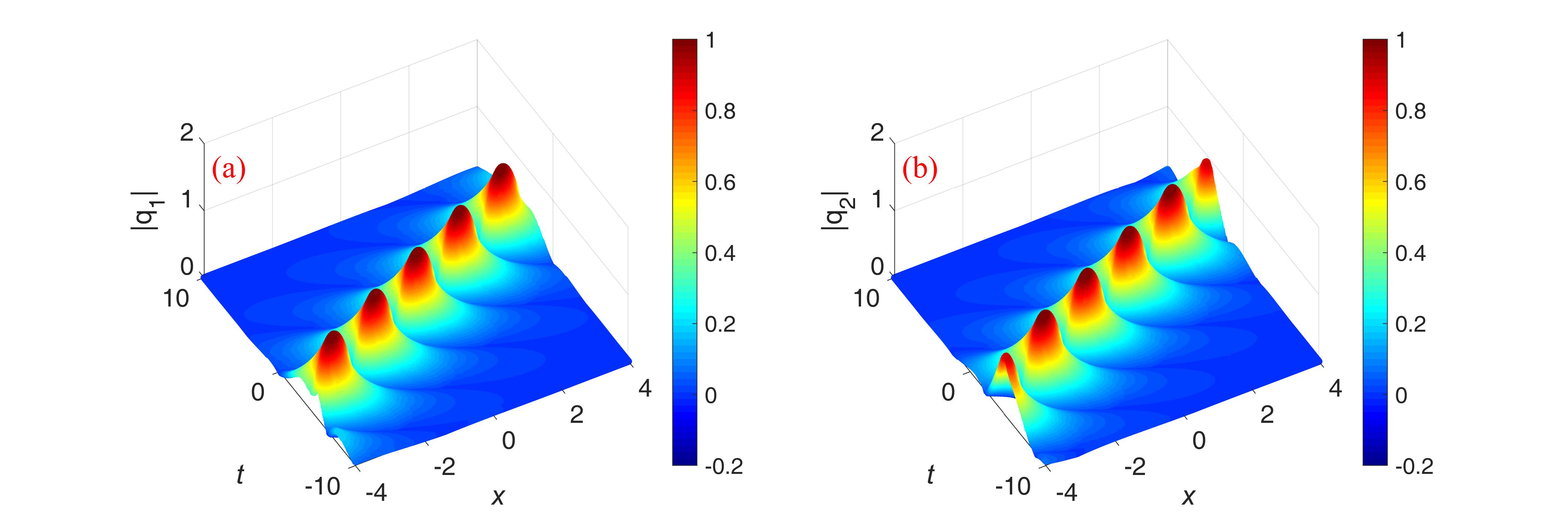}
\caption{(color on line) ${\rm SU}(2)$ solitons. Parameters: $c_1=c_2=1$, $d_1=d_2=1$, $a_1=\frac{3}{2}$, $b_1=-1$, $\delta=1$, $\vartheta_1=\vartheta_2=0$, $\phi_1=\phi_2=0$}
\label{su2}
\end{figure}

Now we consider the scattering data of above single soliton. As $y\to +\infty$, then the Darboux matrix \eqref{eq:dt} is 
\[
\lim_{y\to+\infty}\mathbf{T}(y,s;\lambda)= \mathbf{T}_{[+\infty]}=\mathbb{I}_4-\frac{\lambda_1-\lambda_1^*}{|\mathbf{c}|^2}
\begin{bmatrix}
\frac{\mathbf{c}\mathbf{c}^{\dag}}{\lambda-\lambda_1^*}+\frac{\sigma_2\mathbf{c}^*\mathbf{c}^{\T}\sigma_2^{-1}}{\lambda-\lambda_1^*}& \mathbf{0}\\[8pt]
\mathbf{0}&\mathbf{0} \\
\end{bmatrix},\,\,\,\,\, \mathbf{c}=\begin{bmatrix}
c_1\ee^{\ii\vartheta_1} \\
c_2\ee^{\ii\vartheta_2} \\
\end{bmatrix};
\]
Correspondingly, as $y\to -\infty$, the Darboux matrix \eqref{eq:dt} is 
\[
\lim_{y\to-\infty}\mathbf{T}(y,s;\lambda)= \mathbf{T}_{[-\infty]}=\mathbb{I}_4-\frac{\lambda_1-\lambda_1^*}{|\mathbf{d}|^2}
\begin{bmatrix}
\mathbf{0}& \mathbf{0}\\[8pt]
\mathbf{0}&\frac{\mathbf{d}\mathbf{d}^{\dag}}{\lambda-\lambda_1^*}+\frac{\sigma_2\mathbf{d}^*\mathbf{d}^{\T}\sigma_2^{-1}}{\lambda-\lambda_1^*} \\
\end{bmatrix},\,\,\,\,\, \mathbf{d}=\begin{bmatrix}d_1\ee^{-\ii\phi_1} \\
d_2\ee^{-\ii\phi_2} \\\end{bmatrix}.
\]
By the linear property of the Darboux matrix, we can rescale the Darboux matrix
$\mathbf{T}^{[{\rm N}]}=\mathbf{T}(y,s;\lambda)\mathbf{T}_{[-\infty]}^{-1}$. Taking the limit $y\to+\infty$, we have \[
\mathbf{a}(\lambda)=\mathbb{I}_2-\frac{\lambda_1-\lambda_1^*}{|\mathbf{c}|^2}\left(\frac{\mathbf{c}\mathbf{c}^{\dag}}{\lambda-\lambda_1^*}+\frac{\sigma_2\mathbf{c}^*\mathbf{c}^{\T}\sigma_2^{-1}}{\lambda-\lambda_1^*}\right),\,\,\,\,\,\,
\mathbf{d}(\lambda)^{-1}=\mathbb{I}_2-\frac{\lambda_1-\lambda_1^*}{|\mathbf{d}|^2}\left(\frac{\mathbf{d}\mathbf{d}^{\dag}}{\lambda-\lambda_1^*}+\frac{\sigma_2\mathbf{d}^*\mathbf{d}^{\T}\sigma_2^{-1}}{\lambda-\lambda_1^*}\right).
\]
The determinant of $\mathbf{a}(\lambda)$ equals to $\frac{\lambda-\lambda_1}{\lambda-\lambda_1^*}\frac{\lambda+\lambda_1^*}{\lambda+\lambda_1}$. The eigenfunction at the point $\lambda=\lambda_1$ can be easily constructed through the Darboux matrix:
\begin{equation}
\mathbf{T}(y,s;\lambda_1)\begin{bmatrix}
c_1\ee^{\ii\left(\frac{\lambda_1}{4}s-\frac{\delta}{2\lambda_1}y+\vartheta_1\right)} \\
c_2\ee^{\ii\left(\frac{\lambda_1}{4}s-\frac{\delta}{2\lambda_1}y+\vartheta_2\right)} \\
0 \\
0 \\
\end{bmatrix}=\mathbf{T}(y,s;\lambda_1)\begin{bmatrix}
0 \\
0 \\
-d_1\ee^{-\ii\left(\frac{\lambda_1}{4}s-\frac{\delta}{2\lambda_1}y+\phi_1\right)} \\
-d_2\ee^{-\ii\left(\frac{\lambda_1}{4}s-\frac{\delta}{2\lambda_1}y+\phi_2\right)} \\
\end{bmatrix}\in [\mathbf{L}^2(\mathbb{R})]^4.
\end{equation}
{Similarly, we can construct other eigenfunctions at points $\lambda=\pm\lambda_1^*$ and $\lambda=-\lambda_1$}.

We now turn to study the double geometric zeros of $\mathbf{a}(\lambda)$. In this case, the Darboux matrix can be constructed with the following form:
\begin{equation}\label{eq:dt-2}
\mathbf{T}_2(y,s;\lambda)=\mathbb{I}_4-\mathbf{Y}\mathbf{M}_2^{-1}(\lambda\mathbb{I}_4-\mathbf{D}_4)^{-1}\mathbf{Y}^{\dag}
\end{equation}
where
\[
\mathbf{Y}=\left[\Phi_1,\Lambda\Phi_1^*,\Phi_2,\Lambda\Phi_2^*\right],\,\,\,\,\,\, \mathbf{D}_4={\rm diag}\left(\lambda_1^*,-\lambda_1,\lambda_1^*,-\lambda_1\right),\,\,\,\,\, \mathbf{M}_2=\begin{bmatrix}
\mathbf{M}_{11} & \mathbf{M}_{12} \\
\mathbf{M}_{21} & \mathbf{M}_{22} \\
\end{bmatrix}
\]
and
\[
\mathbf{M}_{11}=\frac{\Phi_1^{\dag}\Phi_1}{\lambda_1-\lambda_1^*}\mathbb{I}_2,\,\,\, \mathbf{M}_{22}=\frac{\Phi_2^{\dag}\Phi_2}{\lambda_1-\lambda_1^*}\mathbb{I}_2,\,\,\,
\mathbf{M}_{12}=\begin{bmatrix}
\frac{\Phi_1^{\dag}\Phi_2}{\lambda_1-\lambda_1^*} & \frac{\Phi_1^{\dag}\Lambda\Phi_2^*}{-\lambda_1^*-\lambda_1^*} \\
-\frac{\Phi_1^{\T}\Lambda\Phi_2}{\lambda_1+\lambda_1} &\frac{\Phi_1^{\T}\Phi_2^*}{-\lambda_1^*+\lambda_1} \\
\end{bmatrix},\,\,\, \mathbf{M}_{21}=\begin{bmatrix}
\frac{\Phi_2^{\dag}\Phi_1}{\lambda_1-\lambda_1^*} & \frac{\Phi_2^{\dag}\Lambda\Phi_1^*}{-\lambda_1^*-\lambda_1^*} \\
-\frac{\Phi_2^{\T}\Lambda\Phi_1}{\lambda_1+\lambda_1} &\frac{\Phi_2^{\T}\Phi_1^*}{-\lambda_1^*+\lambda_1} \\
\end{bmatrix}
\]
Actually, this Darboux matrix can be {viewed} as the iteration of Darboux matrix \eqref{eq:dt}.
To obtain the solutions {which are different from}  above single soliton solution, we choose
\[
\Phi_1=\begin{bmatrix}
\ee^{\ii\left(\frac{\lambda_1}{4}s-\frac{\delta}{2\lambda_1}y\right)} \\
0 \\
d_{1,1}\ee^{-\ii\left(\frac{\lambda_1}{4}s-\frac{\delta}{2\lambda_1}y+\phi_{1,1}\right)} \\
d_{2,1}\ee^{-\ii\left(\frac{\lambda_1}{4}s-\frac{\delta}{2\lambda_1}y+\phi_{2,1}\right)} \\
\end{bmatrix},\,\,\,\,\, \Phi_2=\begin{bmatrix}
0 \\
\ee^{\ii\left(\frac{\lambda_1}{4}s-\frac{\delta}{2\lambda_1}y\right)}  \\
d_{1,2}\ee^{-\ii\left(\frac{\lambda_1}{4}s-\frac{\delta}{2\lambda_1}y+\phi_{1,2}\right)} \\
d_{2,2}\ee^{-\ii\left(\frac{\lambda_1}{4}s-\frac{\delta}{2\lambda_1}y+\phi_{2,2}\right)} \\
\end{bmatrix},
\]
where $d_{i,j}$, $\phi_{i,j}$, $i,j=1,2$, are real parameters. The corresponding solution from the formula \eqref{eq:backlund1} is a breather
\begin{equation}\label{eq:zero-breather}
\begin{split}
q_1[2]=&\frac{F_1}{D_1}, \,\,\,\,\,
q_2[2]=\frac{F_2}{D_1}, \,\,\,\,\,
\rho[2]=\frac{\delta}{2}-2\ln_{ys}(D_1) \\
x=&\frac{\delta}{2}y-2\ln_{s}(D_1),\,\,\,\, t=-s,
\end{split}
\end{equation}
where
\[
D_1=\frac{\Phi_1^{\dag}\Phi_1\Phi_2^{\dag}\Phi_2-\Phi_2^{\dag}\Phi_1\Phi_1^{\dag}\Phi_2}{(\lambda_1-\lambda_1^*)^2}
-\frac{\Phi_2^{\dag}\Lambda\Phi_1^*\Phi_1^{\T}\Lambda\Phi_2}{4|\lambda_1|^2},
\]
and
\[
F_j=\mathbf{Y}_{3,1}\mathbf{M}_{12}\mathbf{Y}_{j,2}^{\dag}+\mathbf{Y}_{3,2}\mathbf{M}_{21}\mathbf{Y}_{j,1}^{\dag}-(\mathbf{Y}_{3,1}\mathbf{M}_{22}\mathbf{Y}_{j,1}^{\dag}+\mathbf{Y}_{3,2}\mathbf{M}_{11}\mathbf{Y}_{j,2}^{\dag}),\,\,\,\, j=1,2,
\]
where $\mathbf{Y}_{j,1}$ represents the $j-$th row and first two column of $\mathbf{Y}$, $\mathbf{Y}_{j,2}$ represents the $j-$th row and last two column of $\mathbf{Y}$. If $t\to\pm\infty$, then $x\pm b_1=\frac{\delta}{2}y$. The breather propagates along the curve $\frac{s}{2}-\frac{\delta y}{a_1^2+b_1^2}={\rm const}$ which approaches the line $\frac{(x\pm 2b_1)}{a_1^2+b_1^2}-\frac{t}{4}={\rm const}$. The velocity of the breather tends to $\frac{a_1^2+b_1^2}{4}$.
{Fig. \ref{zero-breather} displays the dynamics of one breather for two sets of parameters.}

\begin{figure}[tbh]
\centering
\includegraphics[height=50mm,width=140mm]{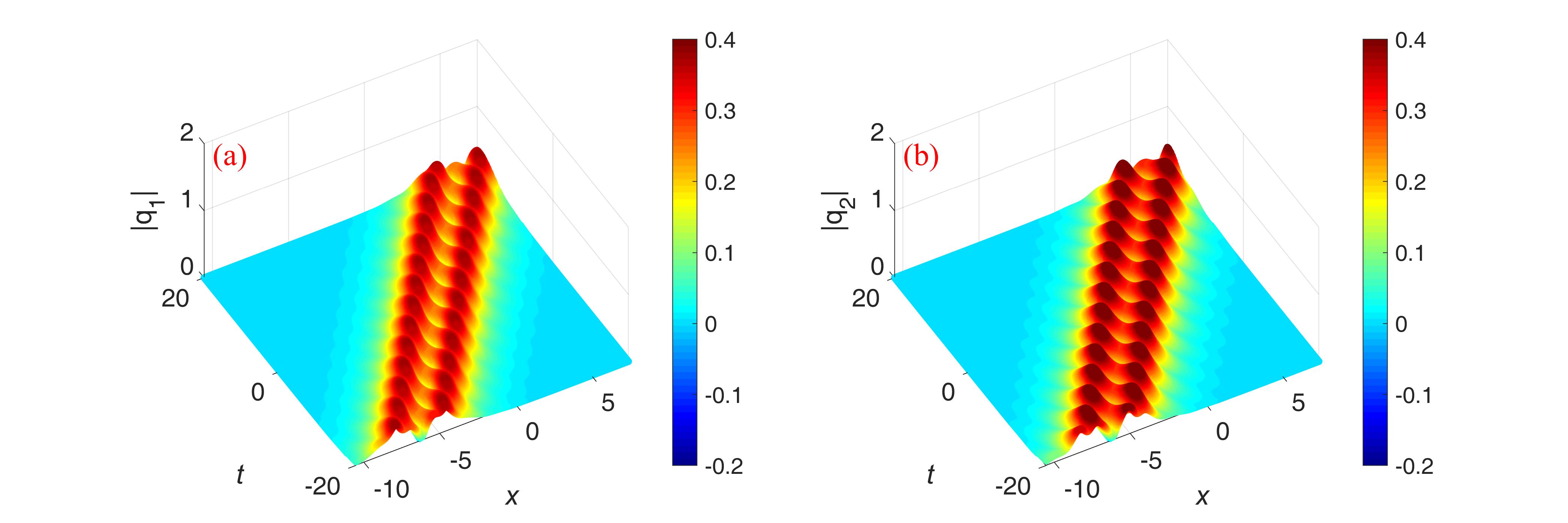}
\caption{(color on line) Breather on a zero background. Parameters: $a_1=1$, $b_1=-\frac{1}{2}$, $\delta=1$, $d_{1,1}=10$, $d_{1,2}=2$, $d_{2,1}=-10$, $d_{2,2}=-1$, $\phi_{i,j}=0$, $i,j=1,2.$}
\label{zero-breather}
\end{figure}

As shown in the simple zero case, the matrix functions $\mathbf{a}(\lambda)$, $\mathbf{d}(\lambda)$ can be obtained by taking the limit of the Darboux matrix as $y\to+\infty$ \[
\mathbf{a}(\lambda)=\frac{\lambda-\lambda_1}{\lambda-\lambda_1^*}\frac{\lambda+\lambda_1^*}{\lambda+\lambda_1}\mathbb{I}_2,\,\,\,\,\,\,\, \mathbf{d}(\lambda)=\frac{\lambda-\lambda_1^*}{\lambda-\lambda_1}\frac{\lambda+\lambda_1}{\lambda+\lambda_1^*}\mathbb{I}_2.
\]
The corresponding two eigenfunctions at $\lambda=\lambda_1$ are
\[
\mathbf{T}_2(y,s;\lambda_1)\begin{bmatrix}
\ee^{\ii\left(\frac{\lambda_1}{4}s-\frac{\delta}{2\lambda_1}y\right)} \\
0 \\
0 \\
0 \\
\end{bmatrix},\,\,\,\,\, \mathbf{T}_2(y,s;\lambda_1)\begin{bmatrix}
0 \\
\ee^{\ii\left(\frac{\lambda_1}{4}s-\frac{\delta}{2\lambda_1}y\right)} \\
0 \\
0 \\
\end{bmatrix}.
\]
\subsection{Multi-soliton solutions}
In what follows, we consider the multi-soliton solution. Among them, there is an interesting case--the double hump soliton---which corresponds to two solitons of the same speed. In the scalar case, two solitons with the same speed will resonate and thus form a breather. However, for the two component system, the nonlinear superposition of $q_1={\text{single\,\, soliton}}$, $q_2=0$ and $q_1=0$, $q_2={\text{single\,\,soliton}}$ with the same speed, generates a double hump soliton In the following, we {discuss} these three types of fundamental soliton solutions in detail.

Suppose we have the following linear independent vectors:
\begin{equation}\label{eq:vectors}
\Phi_{i,j_i}=\begin{bmatrix}
\mathbf{c}_{i,j_i}\ee^{\ii\left(\frac{\lambda_{i,j_i}}{4}s-\frac{\delta}{2\lambda_{i,j_i}}y\right)}\\
\mathbf{d}_{i,j_i}\ee^{-\ii\left(\frac{\lambda_{i,j_i}}{4}s-\frac{\delta}{2\lambda_{i,j_i}}y\right)}\\
\end{bmatrix}
\end{equation}
where $\mathbf{c}_{i,j}$ and $\mathbf{d}_{i,j}$ are constant column vectors, $|\lambda_{i,j_1}|=|\lambda_{i,j_2}|$, $j_1\neq j_2$, with at most two equal parameters among $\lambda_{i,j_i}$, $j_i=1,2,\cdots, k_i$. Inserting the vectors \eqref{eq:vectors} into the formula \eqref{eq:n-fold-dt-1}, the multi-soliton solution can be derived through the formula \eqref{eq:backlund1}.

The norm  $|\lambda_{i,j_i}|$ will determine the velocity of soliton. {If two solitons propagate with the same speed, they will form a breather. When the distance of the two solitons is large enough, the ``breath" effect will be feeble.}
{The choice of a set of linear independent vector with $|\lambda_{i,1}|=|\lambda_{i,j_i}|$, $j_i=1,2,\cdots, k_i$ will form a bound state}. As mentioned above, in the multi-component system, in addition to the breathers, there is an interesting double-hump soliton that can be obtained for a special choice of parameters. In what follows, we give its exact formula. Firstly, we choose the vectors
\begin{equation}\label{eq:vector2}
\Phi_{j}=\begin{bmatrix}
\mathbf{c}_{j}\\[5pt]
\mathbf{d}_{j}\ee^{\omega_j}\\
\end{bmatrix},\,\,\,\omega_j=\ii\left(\frac{\delta}{\lambda_{j}}y-\frac{\lambda_{j}}{2}s\right), \,\,\, j=1,2,\,\,\,\, |\lambda_{1}|=|\lambda_{2}|,\,\,\,\, \lambda_{1}\neq\lambda_{2},
\end{equation}
where
\[\mathbf{c}_{1}=\begin{bmatrix}
0\\
1\\
\end{bmatrix},\,\,\,\mathbf{c}_{2}=\begin{bmatrix}
0\\
1\\
\end{bmatrix},\,\,\, \mathbf{d}_{1}=\begin{bmatrix}
d_{1}\\
0\\
\end{bmatrix},\,\,\, \mathbf{d}_{2}=\begin{bmatrix}
0 \\
d_{2}\\
\end{bmatrix}
\]
Plugging the vectors \eqref{eq:vector2} into \eqref{eq:backlund1} and simplifying it, we can obtain the double-hump soliton:
\begin{equation}\label{eq:double-hump}
\begin{split}
q_1[2]=&\frac{F_1}{D_2}, \,\,\,\,\,
q_2[2]=\frac{F_2}{D_2}, \,\,\,\,\,
\rho[2]=\frac{\delta}{2}-2\ln_{ys}(D_2) \\
x=&\frac{\delta}{2}y-2\ln_{s}(D_2),\,\,\,\, t=-s,
\end{split}
\end{equation}
where
\[
D_2=\frac{(1+|d_1|^2\ee^{\omega_1+\omega_1^*})(1+|d_2|^2\ee^{\omega_2+\omega_2^*})}{(\lambda_1-\lambda_1^*)(\lambda_2-\lambda_2^*)}
+\frac{|d_1|^2|d_2|^2\ee^{\omega_1+\omega_2+\omega_1^*+\omega_2^*}}{|\lambda_1+\lambda_2|^2}+\frac{1}{|\lambda_2-\lambda_1^*|^2},
\]
and
\begin{equation*}
\begin{split}
F_1&=d_2^*\ee^{\omega_2^*}\left(\frac{|d_1|^2\ee^{\omega_1+\omega_1^*}}{-\lambda_2^*-\lambda_1^*}-
\frac{(1+|d_1|^2\ee^{\omega_1+\omega_1^*})}{\lambda_1-\lambda_1^*}+\frac{1}{\lambda_2-\lambda_1^*}\right),
\\
F_2&=d_1\ee^{\omega_1}\left(\frac{|d_2|^2\ee^{\omega_2^*+\omega_2}}{\lambda_1+\lambda_2}-
\frac{(1+|d_2|^2\ee^{\omega_2+\omega_2^*})}{\lambda_2-\lambda_2^*}+\frac{1}{\lambda_2-\lambda_1^*}\right).
\end{split}
\end{equation*}
The shape of the double-hump soliton (Fig. \ref{double-hump}) is {determined} by the parameters $\lambda_1$, $\lambda_2$ and $|d_1|$, $|d_2|$. The soliton velocity approaches $\frac{|\lambda_1|^2}{4}$.

\begin{figure}[tbh]
\centering
\includegraphics[height=50mm,width=140mm]{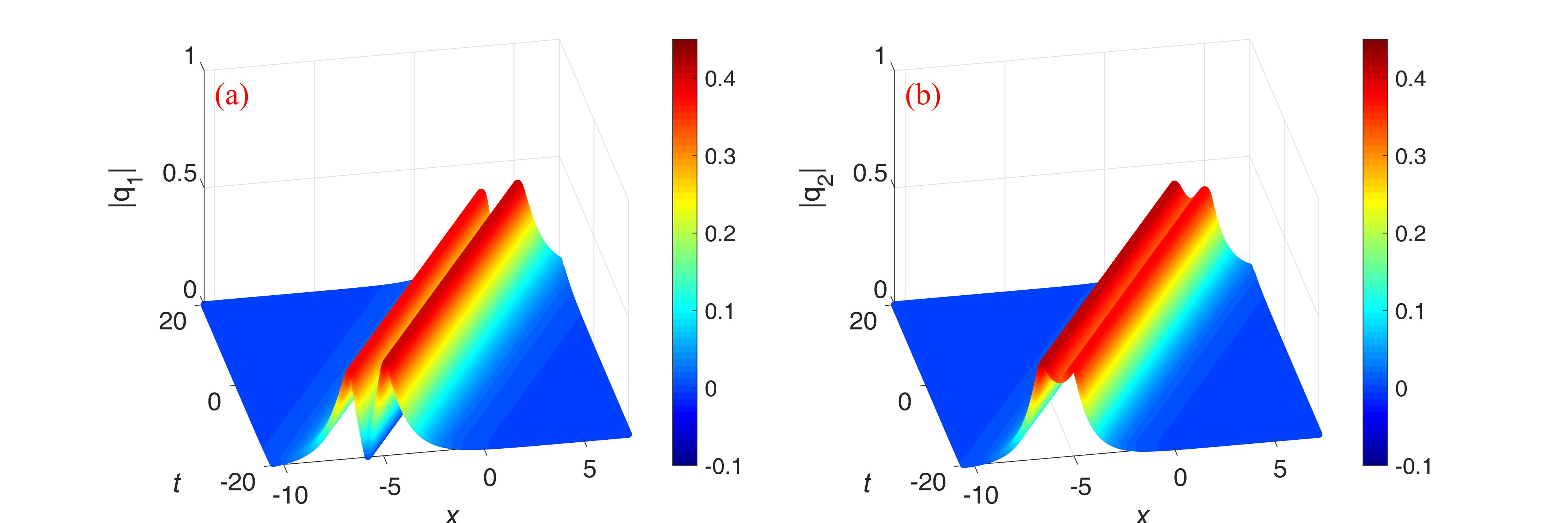}
\caption{(color on line) Double hump soliton on the zero background. Parameters: $\delta=1$, $\lambda_1=\frac{4}{5}-\ii \frac{3}{5}$, $\lambda_2=\frac{\sqrt{3}}{2}-\frac{\ii}{2}$, $d_1=1$, $d_2=0.91$}
\label{double-hump}
\end{figure}

{The dynamics of two-soliton  is shown in Fig. \ref{two-soliton}}.

\begin{figure}[tbh]
\centering
\includegraphics[height=50mm,width=140mm]{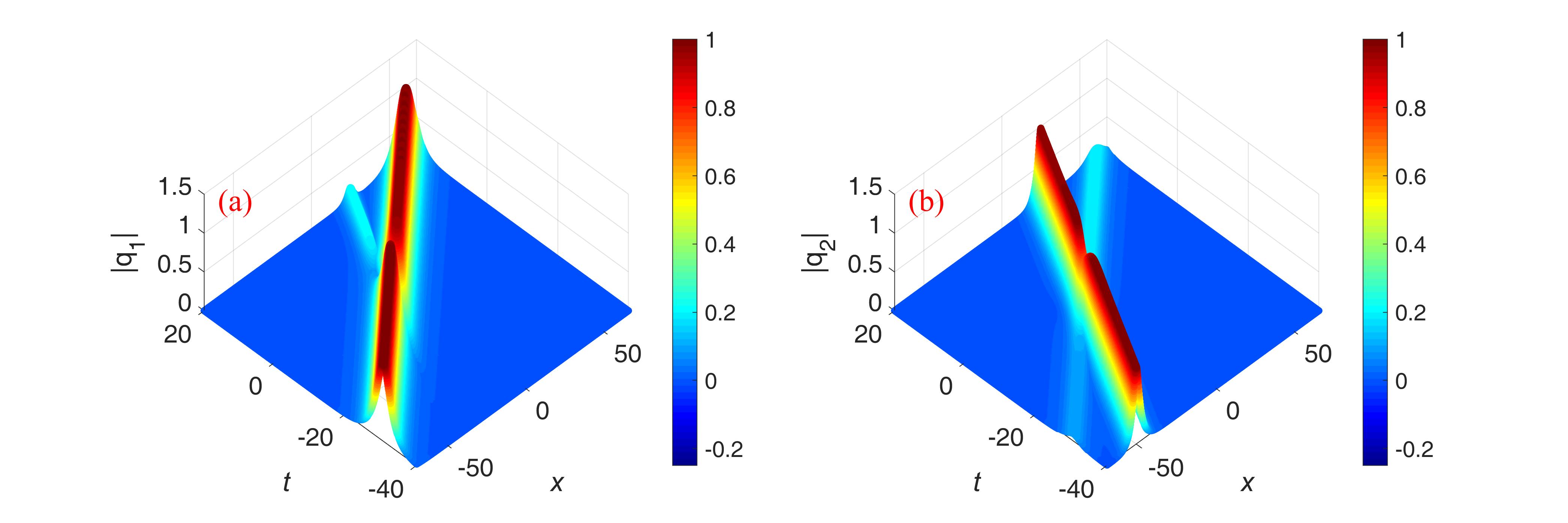}
\caption{(color on line) Two-soliton solution on a zero background. Parameters: $\delta=1$, $\lambda_1=3-\ii$, $\lambda_2=2-\ii$, $c_{1,1}=c_{1,2}=0$, $c_{2,1}=c_{2,2}=1$, $d_{1,1}=d_{1,2}=1$, $d_{2,1}=5$, $d_{2,2}=0$.}
\label{two-soliton}
\end{figure}

\subsection{Higher order soliton solutions}
To find the higher order soliton solution, the function $\ee^{\ii\left(\frac{\delta}{\lambda}y-\frac{\lambda}{2}s\right)}$ is expanded at $\lambda=\lambda_1$ by
\begin{equation}\label{eq:expansion}
 \ee^{\ii\left(\frac{\delta}{\lambda}y-\frac{\lambda}{2}s\right)}=\ee^{\ii\left(\frac{\delta}{\lambda_1}y-\frac{\lambda_1}{2}s\right)}\sum_{n=0}^{\infty}
 E_n(y,s)(\lambda-\lambda_1)^n,\,\,\,\,\, E_n(y,s)=\left(\frac{-1}{\lambda_1}\right)^n\sum_{\|\alpha\|_n=n}\frac{(\ii\delta y\lambda_1^{-1})^{\alpha_1+\alpha_2+\cdots+\alpha_n}}{\alpha_1!\alpha_2!\cdots\alpha_n!}+\frac{\left(-\frac{\ii s}{2}\right)^n}{n!},
\end{equation}
where $\|\alpha\|_{n}=\sum_{j=1}^{n}j\alpha_j$.
Correspondingly, the vector $\Phi_1$ can be expanded at $\lambda=\lambda_1$:
\begin{equation}\label{eq:high-vector}
\Phi_1(\lambda)=\begin{bmatrix}
\mathbf{c}_1 \\[5pt]
\mathbf{d}_1\ee^{\ii\left(\frac{\delta}{\lambda}y-\frac{\lambda}{2}s\right)} \\
\end{bmatrix}=\sum_{n=0}^{\infty}\Phi_1^{[n]}(\lambda-\lambda_1)^n,\,\,\,\,\, \Phi_1^{[n]}=\begin{bmatrix}
{\displaystyle \mathbf{c}_1^{[n]}} \\[5pt]
{\displaystyle \ee^{\ii\left(\frac{\delta}{\lambda_1}y-\frac{\lambda_1}{2}s\right)}\sum_{i=0}^{n}\mathbf{d}_1^{[i]}E_{n-i}(y,s)} \\
\end{bmatrix}
\end{equation}
where
\[
\mathbf{c}_1=\sum_{n=0}^{\infty}\mathbf{c}_1^{[n]}(\lambda-\lambda_1)^n,\,\,\,\,\, \mathbf{d}_1=\sum_{n=0}^{\infty}\mathbf{d}_1^{[n]}(\lambda-\lambda_1)^n.
\]
Inserting the vectors \eqref{eq:high-vector} into the formula \eqref{eq:general-dt}, the higher order soliton are obtained:
\begin{equation}\label{eq:high-order}
\begin{split}
 q_1[N]=&\frac{\det(\mathbf{M}_N^{(1)})}{\det(\mathbf{M}_N)},
 \,\,\,\,\,\,
 q_2[N]=\frac{\det(\mathbf{M}_N^{(2)})}{\det(\mathbf{M}_N)}, \\
 \rho[N]=&\frac{\delta}{2}-\ln_{ys}(\mathbf{M}_N),\\
 x=&\frac{\delta y}{2}-\ln_{s}(\mathbf{M}_N)+{\rm const},\,\,\,\, t=-s,
\end{split}
\end{equation}
where
\begin{equation}
\mathbf{M}_N=\mathbf{K}_N\mathbf{S}_{N}\mathbf{K}_N^{\dag},\,\,\,\,\, \mathbf{M}_N^{(i)}=\begin{bmatrix}
\mathbf{M}_N & \mathbf{Y}^{(i)\dag} \\[5pt]
\mathbf{Y}^{(3)} & 0\\
\end{bmatrix},
\end{equation}
and
\[
\mathbf{Y}=\left[\Phi_1^{[0]},\Lambda\Phi_1^{[0]*},\cdots,\Phi_1^{[N-1]},\Lambda\Phi_1^{[N-1]*}\right]
=\begin{bmatrix}
\mathbf{Y}^{(1)}\\
\mathbf{Y}^{(2)}\\
\mathbf{Y}^{(3)}\\
\mathbf{Y}^{(4)}\\
\end{bmatrix}
\]
\[
\mathbf{K}_N=\begin{bmatrix}
\Phi_1^{[0]\dag} &0&\cdots&0&0 \\
0&-\Phi_1^{[0]\T}\Lambda &\cdots&0&0 \\
\vdots &\vdots &\cdots&\vdots&\vdots \\
\Phi_1^{[N-1]\dag}&0&\cdots&\Phi_1^{[0]\dag}&0 \\
0&-\Phi_1^{[N-1]\T}\Lambda&\cdots&0&-\Phi_1^{[0]\T}\Lambda \\
\end{bmatrix},\]
\[ \mathbf{S}_{N}=\left(\begin{bmatrix}
 \binom{i+j-2}{i-1}\frac{(-1)^{j-1}\mathbb{I}_4}{(\lambda_1-\lambda_1^*)^{i+j-1}}& \binom{i+j-2}{i-1}\frac{\mathbb{I}_4}{(-\lambda_1^*-\lambda_1^*)^{i+j-1}} \\[8pt]
\binom{i+j-2}{i-1}\frac{(-1)^{j-1+i-1}\mathbb{I}_4}{(\lambda_1+\lambda_1)^{i+j-1}} & \binom{i+j-2}{i-1}\frac{(-1)^{i-1}\mathbb{I}_4}{(-\lambda_1^*+\lambda_1)^{i+j-1}} \\
\end{bmatrix}\right)_{1\leq i,j\leq N}.
\]
In particular, by choosing the vectors
\[
\Phi_1^{[0]}=\begin{bmatrix}
0\\
1 \\
d_1\ee^{A_1+\ii B_1} \\
d_2\ee^{A_1+\ii B_1} \\
\end{bmatrix},\,\,\,\,\,\, \Phi_1^{[1]}=\begin{bmatrix}
0\\
0 \\
d_1\left(C_1+\ii D_1\right)\ee^{A_1+\ii B_1} \\
d_2\left(C_2+\ii D_2\right)\ee^{A_1+\ii B_1} \\
\end{bmatrix}
\]
where $\lambda_1=a_1+\ii b_1$, $A_1=b_1\left(\frac{\delta y}{|\lambda_1|^2}+\frac{s}{2}\right)$, $B_1=a_1\left(\frac{\delta y}{|\lambda_1|^2}-\frac{s}{2}\right)$, $C_i=e_i-\frac{2a_1b_1\delta y}{|\lambda_1|^4}$, $D_i= f_i-\left(\frac{(a_1^2-b_1^2)\delta y}{|\lambda_1|^4}+\frac{s}{2}\right)$, $i=1,2$, $e_i$, $f_i$ are arbitrary real constants, we have the second order soliton solution:
\begin{equation}\label{eq:sec-order}
\begin{split}
 q_1[2]=&\frac{ d_2 F_2{\rm e}^{A_{1}-\ii B_{1}}}{4 {b_{1}}^{3}\left( \ii
b_{1}-a_{1} \right)F_1},
 \,\,\,\,\,\,
 q_2[2]=\frac{ d_1 F_3{\rm e}^{A_{1}+\ii B_{1}}}{4{b_{ 1}}^{3} \left( a_{ 1}+\ii b_{ 1} \right) F_1}, \\
 \rho[2]=&\frac{\delta}{2}-2\ln_{ys}(F_1),\\
 x=&\frac{\delta y}{2}-2\ln_{s}(F_1),\,\,\,\, t=-s,
\end{split}
\end{equation}
where
\begin{multline*}
F_1=\frac {
 \left[\left( -4\,{d_{1}}^2 \left(\left(C_1-C_2 \right)^2+ \left(D_1-D_2\right)^2 \right) {d_{ 2}
}^ 2{a_{ 1}}^ 2-{d}^{4} \right) {b_{ 1}}^ 2-{a_{ 1}}^ 2{d}^{4}
 \right] {{\rm e}^{4\,A_{ 1}}}}{16{b_{ 1}}^{4} \left( {a_{ 1}}^ 2+{b
_{ 1}}^ 2 \right) }\\-\frac { \left[  \left(  \left( {C_{ 1}}^ 2+{D_{ 1}}^ 2
 \right) {d_{ 1}}^ 2+{d_{ 2}}^ 2 \left( {C_{ 2}}^ 2+{D_{ 2}}^{2
} \right)  \right) {b_{ 1}}^ 2+{d}^ 2/2\right] {{\rm e}^{2\,A_{ 1}}}}{4{b_{ 1}}^{4}}-{\frac  1{16{b_{ 1}}^{4}}},
\end{multline*}
\begin{multline*}
F_2=\left[  \left( -2\,{d_1}^ 2 \left(\left(D_{1}-
\ii C_{1}\right) \left(C_{2}-\ii D_{2}\right)+\ii{C_1}^2+\ii{D_1}^ 2 \right) a_{1}+ \left( \ii C_{2}-2
\,D_{1}+D_{2} \right) {d_{1}}^2+{d_{2}}^2 \left( \ii C_{2}-
D_{2} \right)  \right) {b_1}^ 2
 \right]{{\rm e}^{2 A_1}}\\
+\left[\left(  \left(  \left( -2\ii D_
1+\ii D_2-C_2 \right) {d_ 1}^ 2-{d_ 2}^ 2 \left( C_{
2}+iD_{2} \right)  \right) a_{1}-{d}^2
 \right) b_{1}-\ii a_{1}{d}^2
 \right]{{\rm e}^{2 A_1}}
\\+\left( \ii C_{2}+D_{2} \right) {b_{1}}^2+ \left( \left( \ii D_{
 2}-C_{2} \right) a_1-1 \right) b_ 1-\ii a_1,
\end{multline*}
\begin{multline*}
F_3=-\left[  \left(2 {d_{2}}^ 2 \left(  \left( \ii C_{ 2}+D_{ 2} \right) \left(C_{1}+\ii D_{1}\right)-i{C_{2}}^ 2-\ii {D_{2}}^ 2 \right) a_{1}+ \left( \ii C_{1}-D_{1}+2D_{2} \right) {d_{2}}^ 2+{d_{ 1}}^2 \left( iC_{ 1}+D_{ 1} \right)  \right) {b_{ 1}}^2
 \right] {{\rm e}^{2\,A_
{ 1}}}\\
-\left[\left(\left(\left( \ii D_{ 1}-2\ii D_{ 2}+C_{1} \right) {d_{
 2}}^ 2-{d_{ 1}}^ 2 \left( iD_{ 1}-C_{ 1} \right)  \right)
a_{ 1}+{d}^2\right) b_{ 1}-\ii a_{1
}d^2  \right] {{\rm e}^{2\,A_
{ 1}}}
\\-\left[  \left( \ii C_{1} -D_{ 1}\right) {b_{ 1}}^ 2+ \left( \left( \ii D_{ 1}+C_{1
} \right) a_{ 1}+1\right) b_{1}-\ii a_{1} \right],
\end{multline*}
and $d^2=d_1^2+d_2^2$.
{Fig. \ref{second-order} and  Fig. \ref{third-order} display a second order soliton and a third order soliton, respectively, while Fig. \ref{second-second-order} shows a  superposition of two second order solitons}. It is seen that the two solitons in distinct components have different velocities.


\begin{figure}[tbh]
\centering
\includegraphics[height=50mm,width=140mm]{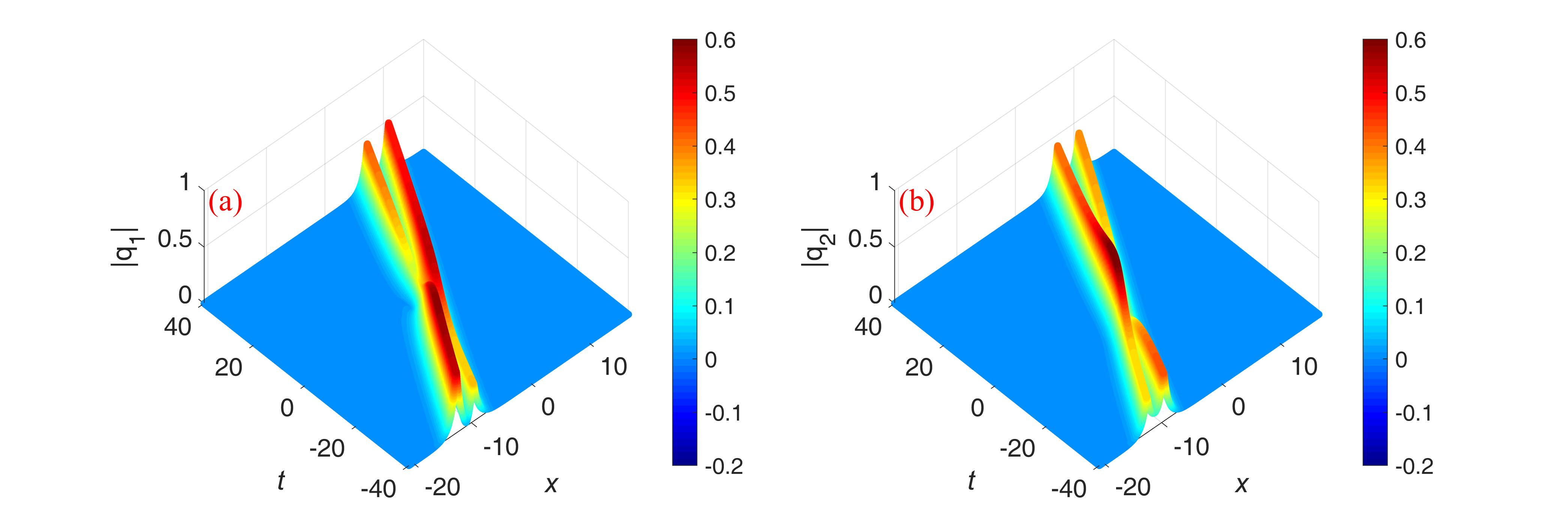}
\caption{The second order soliton. Parameters: $\delta=1$, $a_1=\frac{4}{5}$, $b_1=-\frac{3}{5}$, $d_{1}=d_2=1$, $e_{1}=2$, $e_{2}=10$, $f_1=f_2=0$}
\label{second-order}
\end{figure}
\begin{figure}[tbh]
\centering
\includegraphics[height=50mm,width=140mm]{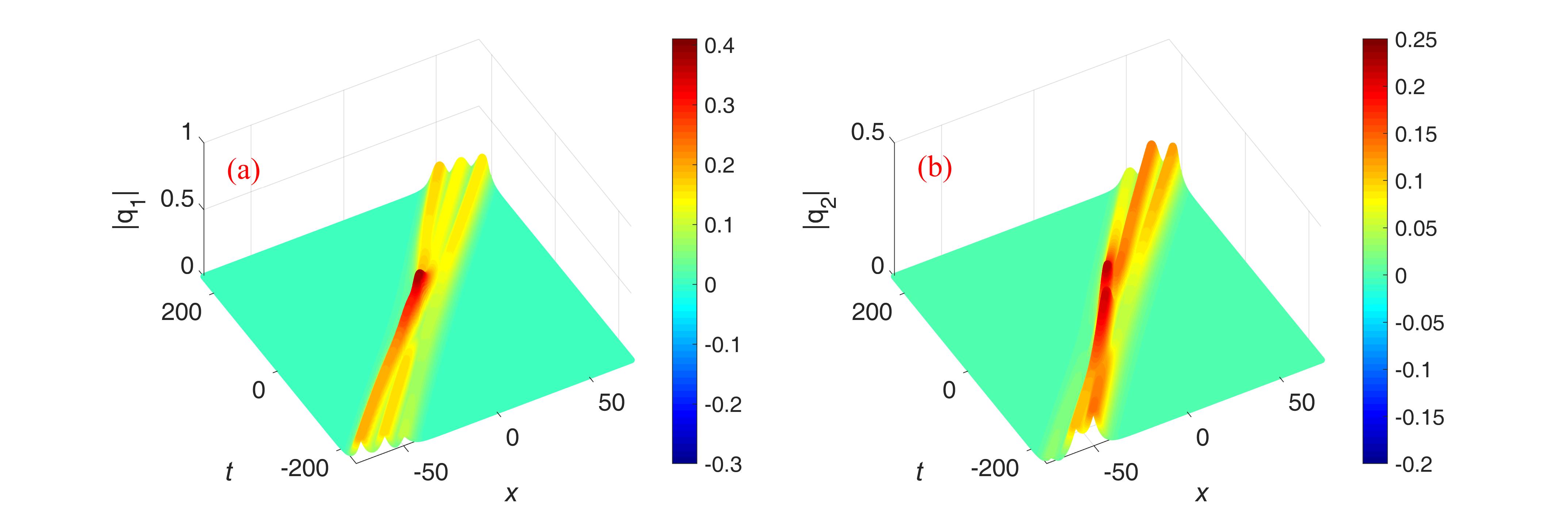}
\caption{The third order soliton. Parameters: $\delta=1$, $\lambda_1=1-\frac{\ii}{5}$, $c_{1}=0$, $c_{2}=1$, $d_1=d_2=1$, $d_1^{[1]}=2$, $d_2^{[1]}=10$,}
\label{third-order}
\end{figure}
\begin{figure}[tbh]
\centering
\includegraphics[height=50mm,width=140mm]{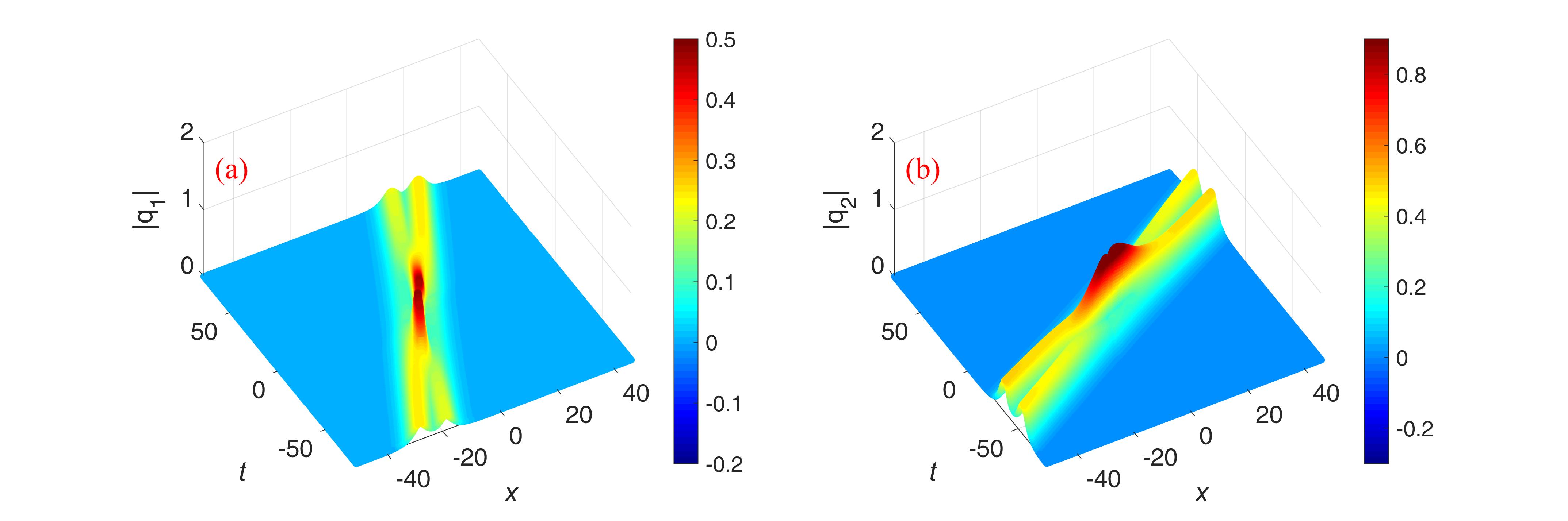}
\caption{The second-second order soliton. Parameters: $\delta=1$, $\lambda_1=1-\frac{\ii}{4}$, $\lambda_2=2-\frac{\ii}{2}$, $c_{2,1}=c_{2,2}=1$, $d_{1,1}=d_{2,2}=1$, $d_{1,2}=d_{2,1}=0$, $d_{i,j}^{[1]}=0$, $i,j=1,2$.}
\label{second-second-order}
\end{figure}

\section{Modulational instability analysis and rogue waves}\label{sec4}

\subsection{Modulational instability analysis of the TCCD equation}
In this section, we use the squared eigenfunctions method to derive the linear stability equation for the TCCD equations. Suppose we have the following stationary zero curvature equation:
\begin{equation}\label{eq:station-eq}
\mathbf{L}_y(y,s;\lambda)=\left[\mathbf{U}(y,s;\lambda), \mathbf{L}(y,s;\lambda)\right],\,\,\,\,\,\, \mathbf{L}_s(y,s;\lambda)=\left[\mathbf{V}(y,s;\lambda), \mathbf{L}(y,s;\lambda)\right]
\end{equation}
where
\[
\mathbf{L}(y,s;\lambda)=\begin{bmatrix}
\mathbf{A}(y,s;\lambda) & \mathbf{B}(y,s;\lambda) \\
\mathbf{C}(y,s;\lambda) & \mathbf{D}(y,s;\lambda) \\
\end{bmatrix}
\]
which can be rewritten as
\begin{equation}\label{eq:station-s}
\begin{split}
\mathbf{A}_s=&\frac{\ii}{2}(\kappa\mathbf{Q}^{\dag}\mathbf{C}-\mathbf{B}\mathbf{Q}),\,\,\,\,\,\,\, \mathbf{D}_s=\frac{\ii}{2}(\mathbf{Q}\mathbf{B}-\kappa\mathbf{C}\mathbf{Q}^{\dag}), \\
\mathbf{B}_s=&\frac{\ii}{2}\lambda \mathbf{B}+\frac{\ii}{2}\kappa(\mathbf{Q}^{\dag}\mathbf{D}-\mathbf{A}\mathbf{Q}^{\dag}),\,\,\,\,\,\, \mathbf{C}_s=-\frac{\ii}{2}\lambda \mathbf{C}+\frac{\ii}{2}(\mathbf{Q}\mathbf{A}-\mathbf{D}\mathbf{Q}), \\
\end{split}
\end{equation}
and
\begin{equation}\label{eq:station-y}
\begin{split}
\lambda\mathbf{A}_y=&-(\kappa\mathbf{Q}_y^{\dag}\mathbf{C}+\mathbf{B}\mathbf{Q}_y),\,\,\,\,\,\,\, \lambda\mathbf{D}_y=\mathbf{Q}_y\mathbf{B}+\kappa\mathbf{C}\mathbf{Q}_y^{\dag}, \\
\lambda\mathbf{B}_y=&-2\ii\rho \mathbf{B}-\kappa(\mathbf{Q}_y^{\dag}\mathbf{D}-\mathbf{A}\mathbf{Q}_y^{\dag}),\,\,\,\,\,\, \lambda\mathbf{C}_y=2\ii\rho \mathbf{C}+\mathbf{Q}_y\mathbf{A}-\mathbf{D}\mathbf{Q}_y. \\
\end{split}
\end{equation}
Taking the derivative of equations \eqref{eq:station-s} with respect to $y$, together with equations \eqref{eq:station-y}, we obtain that
\begin{equation}\label{eq:linear-stability}
\begin{split}
\mathbf{A}_{sy}=&\frac{\ii}{2}(\kappa\mathbf{Q}^{\dag}\mathbf{C}-\mathbf{B}\mathbf{Q})_y,\,\,\,\,\,\,\, \mathbf{D}_{sy}=\frac{\ii}{2}(\mathbf{Q}\mathbf{B}-\kappa\mathbf{C}\mathbf{Q}^{\dag})_y, \\
\mathbf{B}_{sy}=&\rho\mathbf{B}+\frac{\ii}{2}\kappa(\mathbf{Q}^{\dag}\mathbf{D}_y-\mathbf{A}_y\mathbf{Q}^{\dag}),\,\,\,\,\,\, \mathbf{C}_{sy}=\rho\mathbf{C}+\frac{\ii}{2}(\mathbf{Q}\mathbf{A}_y-\mathbf{D}_y\mathbf{Q}), \\
\end{split}
\end{equation}
which are linear differential equation for the matrices $\mathbf{A}$, $\mathbf{B}$, $\mathbf{C}$ and $\mathbf{D}$, which depend on $\mathbf{Q}$, $\mathbf{Q}^{\dag}$ and $\rho$ and are independent of $\lambda$. Moreover, if the matrices $\mathbf{A}$, $\mathbf{B}$, $\mathbf{C}$ and $\mathbf{D}$ satisfy the symmetry relations: \begin{equation}\label{eq:symmetry-station}
\mathbf{A}=\ii h_1(y,s)\mathbb{I}_2,\,\,\,\, \mathbf{D}=\ii h_2(y,s)\mathbb{I}_2,\,\,\,\,\, \mathbf{A}-\mathbf{D}=-\frac{\ii h(y,s)}{2}\mathbb{I}_2,\,\,\,\, \mathbf{B}=-\kappa\mathbf{C}^{\dag},\,\,\,\, \sigma_2\mathbf{C}^*\sigma_2^{-1}=-\mathbf{C},\end{equation}
where $h_1(y,s)$, $h_2(y,s)$ and $h(y,s)$ are real functions,
then  \eqref{eq:linear-stability} can be simplified into
\begin{equation}\label{eq:linear-stability-1}
\begin{split}
h_{sy}=&-\frac{\kappa}{2}(q_1^*c_1+q_1c_1^*+q_2c_2^*+q_2^*c_2)_y, \\
c_{1,sy}=&\rho c_1+q_1 h_{y},\,\,\,\,\, c_{2,sy}=\rho c_2+q_2 h_{y}, \\
\end{split}
\end{equation}
where
\[
\mathbf{C}(y,s)=\begin{bmatrix}
c_1(y,s) & c_2(y,s) \\
c_2^*(y,s) &-c_1^*(y,s)  \\
\end{bmatrix}.
\]
Obviously \eqref{eq:linear-stability-1} is the linearized version of  \eqref{eq:cd}. Therefore, the solution of  \eqref{eq:linear-stability-1} can be constructed from \eqref{eq:station-eq}. {In this way, we can avoid solving the linearized equations \eqref{eq:linear-stability-1} directly. The general solution for the linearized equations \eqref{eq:linear-stability-1}--\eqref{eq:station-eq} can be found by combining the inverse scattering analysis \cite{BilmanM17}}.

{Here, we only analyze} the spectral stability of the plane wave solution for the equations \eqref{eq:cd}:
\begin{equation}\label{eq:cd-bc}
q_i=\frac{\alpha_i}{2}\ee^{\theta_i},\,\,\,\,\, \theta_i=\ii\left(\frac{\beta_i s}{2}-\frac{\delta y}{\beta_i}\right),\,\,\,\,\, \rho=\frac{\delta}{2}.
\end{equation}
Inserting the plane wave solution \eqref{eq:cd-bc} into \eqref{eq:laxcd}, we can solve the linear system:
\begin{equation}\label{eq:fund-cd}
\Gamma_i(y,s;\lambda)\equiv\mathbf{D}\mathbf{V}_i(\lambda,\xi^{[i]})\ee^{\frac{\ii}{4}\xi^{[i]}\left(s+\frac{\delta \xi^{[i]}y}{\beta_1\beta_2\lambda}\right)-\frac{\ii\delta(\kappa|\alpha|^2+|\beta|^2+\lambda^2)y}{4\beta_1\beta_2\lambda}},\,\,\,\, |\alpha|^2=\alpha_1^2+\alpha_2^2,\,\,\,\, |\beta|^2=\beta_1^2+\beta_2^2,
\end{equation}
where $i=1,2,3,4,$
\[
\mathbf{D}={\rm diag}\left(\ee^{-\frac{\theta_1-\theta_2}{2}},\ee^{\frac{\theta_1-\theta_2}{2}},\ee^{\frac{\theta_1+\theta_2}{2}} ,\ee^{-\frac{\theta_1+\theta_2}{2}}\right),
\]
if $\alpha_1\alpha_2\neq 0$, we have
\begin{equation}\label{eq:vector}
\mathbf{V}_i(\lambda,\xi^{[i]})=\begin{bmatrix}
\kappa(\beta_1-\beta_2-\lambda+\xi^{[i]})[(\beta_1+\beta_2)^2-(\lambda+\xi^{[i]})^2]+(\alpha_1^2-\alpha_2^2)(\beta_1+\beta_2)+|\alpha|^2(\lambda+\xi^{[i]})\\[5pt]
2\alpha_1\alpha_2(\beta_1+\beta_2)\\[5pt]
\alpha_1[|\alpha|^2+\kappa(\lambda-\beta_1)^2-\kappa(\xi^{[i]}-\beta_2)^2]\\[5pt]
\alpha_2[|\alpha|^2+\kappa(\lambda+\beta_2)^2-\kappa(\xi^{[i]}+\beta_1)^2]\\
\end{bmatrix}
\end{equation}
where $\xi^{[i]}$s are the roots of the following characteristic equation:
\begin{equation}\label{eq:chara-eq}
\det\left(\begin{bmatrix}
\beta_1-\beta_2+\lambda-\xi^{[i]} &0&\kappa \alpha_1&\kappa\alpha_2 \\
0&\beta_2-\beta_1+\lambda-\xi^{[i]}&\kappa\alpha_2&-\kappa\alpha_1 \\
\alpha_1&\alpha_2& -\beta_1-\beta_2-\lambda-\xi^{[i]}&0 \\
\alpha_2&-\alpha_1&0&\beta_1+\beta_2-\lambda-\xi^{[i]} \\
\end{bmatrix}\right)=0;
\end{equation}
if $\alpha_1=0$ and $\alpha_2\neq 0$, we have
\begin{equation}\label{eq:dege-vector}
\mathbf{V}_i(\lambda,\xi^{[i]})=\begin{bmatrix}
1\\[5pt]
0 \\
0 \\[5pt]
\frac{\alpha_2}{\lambda+\xi^{[i]}-\beta_1-\beta_2}\\
\end{bmatrix},\,\,\, i=1,\,2,\,\,\,\,\, \mathbf{V}_i(\lambda,\xi^{[i]})=\begin{bmatrix}
0\\[5pt]
1 \\[5pt]
\frac{\alpha_2}{\lambda+\xi^{[i]}+\beta_1+\beta_2}\\[5pt]
0\\
\end{bmatrix},\,\,\,i=3,\,4
\end{equation}
where $\xi^{[i]}$, $i=1,2$ satisfy the quadratic equation $(\beta_1-\beta_2+\lambda-\xi^{[i]})(\beta_1+\beta_2-\lambda-\xi^{[i]})-\kappa\alpha_2^2=0$, and $\xi^{[i]}$, $i=3,4$ satisfy the quadratic equation $(\beta_2-\beta_1+\lambda-\xi^{[i]})(\beta_1+\beta_2+\lambda+\xi^{[i]})+\kappa\alpha_2^2=0$. We assume that the roots $\xi^{[i]}$ are simple roots. Vector solutions corresponding to the multiple roots can be obtained as coalescence of simple roots.

By the way, the solutions can be constructed through the B\"acklund transformation \eqref{eq:backlund1} and the vector solutions \eqref{eq:fund-cd} and \eqref{eq:dege-vector}. Note that if we choose the combination of two or more vector solutions in \eqref{eq:fund-cd}, we can obtain breathers or resonant breathers solutions. If we choose the combination of two or more vector solutions in \eqref{eq:dege-vector}, we can obtain bright-dark solutions, breather solutions or combination thereof. These solutions can be obtained by inserting the vector solutions \eqref{eq:fund-cd} into formulas \eqref{eq:backlund1}.

Next, we turn to analyze the modulational instability for the plane wave solution \eqref{eq:cd-bc}. Firstly, we construct the solutions $\mathbf{L}(y,s;\lambda)$ in equations \eqref{eq:station-eq}. If the solution $\Phi_i(y,s;\lambda)$ satisfies the Lax pair \eqref{eq:laxcd} with the plane wave solution \eqref{eq:cd-bc}, then the solution $\Phi_j^{\dag}(y,s;\lambda^*)\Sigma$ satisfies the following adjoint Lax pair
\begin{equation}\label{eq:ad-laxcd}
\begin{split}
 -\Psi_y&=\Psi\mathbf{U}(y,s;\lambda), \\
 -\Psi_s&=\Psi\mathbf{V}(y,s;\lambda). \\
\end{split}
\end{equation}
It follows that $\mathbf{L}_{i,j}(y,s;\lambda)=\Phi_i(y,s;\lambda)\Phi_j^{\dag}(y,s;\lambda^*)\Sigma$ satisfies the stationary zero curvature equation \eqref{eq:station-eq}. In general, the solution $\mathbf{L}_{i,j}(y,s;\lambda)$ does not satisfy the symmetry relation \eqref{eq:symmetry-station}. Because the linear partial differential equations \eqref{eq:linear-stability} are independent of $\lambda$, linear combinations of $\mathbf{L}_{i,j}(y,s;\lambda)$ with distinct $\lambda$ still satisfy the equation \eqref{eq:linear-stability}. Through the above analysis, we can construct the matrix function
\begin{multline}\label{eq:abcd}
\begin{bmatrix}
\mathbf{A}(y,s) & \mathbf{B}(y,s) \\
\mathbf{C}(y,s) & \mathbf{D}(y,s) \\
\end{bmatrix}=\Phi_i(y,s;\lambda)\Phi_j^{\dag}(y,s;\lambda^*)\Sigma+\Lambda\Phi_j^*(y,s;\lambda^*)\Phi_i^{\T}(y,s;\lambda)\Lambda^{-1}\Sigma\\
-\Phi_j(y,s;\lambda^*)\Phi_i^{\dag}(y,s;\lambda)\Sigma-\Lambda\Phi_i^*(y,s;\lambda)\Phi_j^{\T}(y,s;\lambda^*)\Lambda^{-1}\Sigma
\end{multline}
satisfying the symmetry relation \eqref{eq:symmetry-station}.
Suppose the functions $c_1(y,s)$, $c_2(y,s)$, $h(y,s)$ have the form:
\begin{equation}\label{eq:fourier-mode}
\begin{split}
h(y,s)&=f\ee^{\ii \eta (y+\mu s)}+f^*\ee^{-\ii \eta^* (y+\mu^* s)}, \\
c_i(y,s)&=\left(g_i\ee^{\ii \eta (y+\mu s)}-g_{-i}^*\ee^{-\ii \eta^* (y+\mu^* s)} \right)\ee^{\ii \theta_i},\,\,\,\, i=1,2.
\end{split}
\end{equation}
By choosing the solutions \eqref{eq:fund-cd}, we can determine $f$, $g_i$, $g_{-i}$, $i=1,2$ exactly. For convenience, we introduce the following notation
\[
\mathbf{V}_i(\lambda,\xi^{[i]})=\begin{bmatrix}
\mathbf{v}(\lambda,\xi^{[i]}) \\
\mathbf{w}(\lambda,\xi^{[i]}) \\
\end{bmatrix},\,\,\,\, \mathbf{v}(\lambda,\xi^{[i]})=\begin{bmatrix}
v_1(\lambda,\xi^{[i]}) \\
v_2(\lambda,\xi^{[i]}) \\
\end{bmatrix},\,\,\,\, \mathbf{w}(\lambda,\xi^{[i]})=\begin{bmatrix}
w_1(\lambda,\xi^{[i]}) \\
w_2(\lambda,\xi^{[i]}) \\
\end{bmatrix}
\]
which yields
\begin{multline}
\mathbf{A}(y,s)=\left(\mathbf{v}(\lambda,\xi^{[i]})\mathbf{v}^{\dag}(\lambda^*,\xi^{[j]*})-
\sigma_2\mathbf{v}^*(\lambda^*,\xi^{[j]*})\mathbf{v}^{\T}(\lambda,\xi^{[i]})\sigma_2\right)
\ee^{\frac{\ii}{4}(\xi^{[i]}-\xi^{[j]})\left(s+(\xi^{[i]}+\xi^{[j]})\frac{\delta y}{\beta_1\beta_2\lambda}\right)}\\
+\left(\sigma_2\mathbf{v}^*(\lambda,\xi^{[i]})\mathbf{v}^{\T}(\lambda^*,\xi^{[j]*})\sigma_2-
\mathbf{v}(\lambda^*,\xi^{[j]*})\mathbf{v}^{\dag}(\lambda,\xi^{[i]})\right)\ee^{-\frac{\ii}{4}(\xi^{[i]*}-\xi^{[j]*})\left(s+(\xi^{[i]*}+\xi^{[j]*})\frac{\delta y}{\beta_1\beta_2\lambda^*}\right)}\\
=\left(v_1(\lambda,\xi^{[i]})v_1^*(\lambda^*,\xi^{[j]*})+v_2(\lambda,\xi^{[i]})v_2^*(\lambda^*,\xi^{[j]*})\right)\mathbb{I}_2\ee^{\frac{\ii}{4}(\xi^{[i]}-\xi^{[j]})\left(s+(\xi^{[i]}+\xi^{[j]})\frac{\delta y}{\beta_1\beta_2\lambda}\right)}\\
-\left(v_1^*(\lambda,\xi^{[i]})v_1(\lambda^*,\xi^{[j]*})+v_2^*(\lambda,\xi^{[i]})v_2(\lambda^*,\xi^{[j]*})\right)\mathbb{I}_2\ee^{-\frac{\ii}{4}(\xi^{[i]*}-\xi^{[j]*})\left(s+(\xi^{[i]*}+\xi^{[j]*})\frac{\delta y}{\beta_1\beta_2\lambda^*}\right)}
\end{multline}
\begin{multline}
\mathbf{C}(y,s)=\left(\mathbf{w}(\lambda,\xi^{[i]})\mathbf{v}^{\dag}(\lambda^*,\xi^{[j]*})-
\sigma_2\mathbf{w}^*(\lambda^*,\xi^{[j]*})\mathbf{v}^{\T}(\lambda,\xi^{[i]})\sigma_2\right)
\ee^{\frac{\ii}{4}(\xi^{[i]}-\xi^{[j]})\left(s+(\xi^{[i]}+\xi^{[j]})\frac{\delta y}{\beta_1\beta_2\lambda}\right)}\\
+\left(\sigma_2\mathbf{w}^*(\lambda,\xi^{[i]})\mathbf{v}^{\T}(\lambda^*,\xi^{[j]*})\sigma_2-
\mathbf{w}(\lambda^*,\xi^{[j]*})\mathbf{v}^{\dag}(\lambda,\xi^{[i]})\right)\ee^{-\frac{\ii}{4}(\xi^{[i]*}-\xi^{[j]*})\left(s+(\xi^{[i]*}+\xi^{[j]*})\frac{\delta y}{\beta_1\beta_2\lambda^*}\right)}\\
=\begin{bmatrix}
\left[w_1(\lambda,\xi^{[i]})v_1^*(\lambda^*,\xi^{[j]*})+w_2^*(\lambda^*,\xi^{[j]*})v_2(\lambda,\xi^{[i]})\right]\ee^{\theta_1} & \left[w_1(\lambda,\xi^{[i]})v_2^*(\lambda^*,\xi^{[j]*})-w_2^*(\lambda^*,\xi^{[j]*})v_1(\lambda,\xi^{[i]})\right]\ee^{\theta_2}   \\[5pt]
\left[w_2(\lambda,\xi^{[i]})v_1^*(\lambda^*,\xi^{[j]*})-w_1^*(\lambda^*,\xi^{[j]*})v_2(\lambda,\xi^{[i]})\right]\ee^{-\theta_2} & \left[w_2(\lambda,\xi^{[i]})v_2^*(\lambda^*,\xi^{[j]*})+w_1^*(\lambda^*,\xi^{[j]*})v_1(\lambda,\xi^{[i]})\right]\ee^{-\theta_1} \\
\end{bmatrix}\\
\ee^{\frac{\ii}{4}(\xi^{[i]}-\xi^{[j]})\left(s+(\xi^{[i]}+\xi^{[j]})\frac{\delta y}{\beta_1\beta_2\lambda}\right)}\\
-\begin{bmatrix}
\left[w_2^*(\lambda,\xi^{[i]})v_2(\lambda^*,\xi^{[j]*})+w_1(\lambda^*,\xi^{[j]*})v_1^*(\lambda,\xi^{[i]})\right]\ee^{\theta_1} & \left[w_1(\lambda^*,\xi^{[j]*})v_2^*(\lambda,\xi^{[i]})-w_2^*(\lambda,\xi^{[i]})v_1(\lambda^*,\xi^{[j]*})\right]\ee^{\theta_2}   \\[5pt]
\left[w_2(\lambda^*,\xi^{[j]*})v_1^*(\lambda,\xi^{[i]})-w_1^*(\lambda,\xi^{[i]})v_2(\lambda^*,\xi^{[j]*})\right]\ee^{-\theta_2} & \left[w_1^*(\lambda,\xi^{[i]})v_1(\lambda^*,\xi^{[j]*})+w_2(\lambda^*,\xi^{[j]*})v_2^*(\lambda,\xi^{[i]})\right]\ee^{-\theta_1}   \\
\end{bmatrix}\\\ee^{-\frac{\ii}{4}(\xi^{[i]*}-\xi^{[j]*})\left(s+(\xi^{[i]*}+\xi^{[j]*})\frac{\delta y}{\beta_1\beta_2\lambda^*}\right)}
\end{multline}
\begin{multline}
\mathbf{D}(y,s)=\kappa\left(\mathbf{w}(\lambda,\xi^{[i]})\mathbf{w}^{\dag}(\lambda^*,\xi^{[j]*})-
\sigma_2\mathbf{w}^*(\lambda^*,\xi^{[j]*})\mathbf{w}^{\T}(\lambda,\xi^{[i]})\sigma_2\right)
\ee^{\frac{\ii}{4}(\xi^{[i]}-\xi^{[j]})\left(s+(\xi^{[i]}+\xi^{[j]})\frac{\delta y}{\beta_1\beta_2\lambda}\right)}\\
+\kappa\left(\sigma_2\mathbf{w}^*(\lambda,\xi^{[i]})\mathbf{w}^{\T}(\lambda^*,\xi^{[j]*})\sigma_2-
\mathbf{w}(\lambda^*,\xi^{[j]*})\mathbf{w}^{\dag}(\lambda,\xi^{[i]})\right)\ee^{-\frac{\ii}{4}(\xi^{[i]*}-\xi^{[j]*})\left(s+(\xi^{[i]*}+\xi^{[j]*})\frac{\delta y}{\beta_1\beta_2\lambda^*}\right)}\\
=\kappa\left(w_1(\lambda,\xi^{[i]})w_1^*(\lambda^*,\xi^{[j]*})+w_2(\lambda,\xi^{[i]})w_2^*(\lambda^*,\xi^{[j]*})\right)\mathbb{I}_2\ee^{\frac{\ii}{4}(\xi^{[i]}-\xi^{[j]})\left(s+(\xi^{[i]}+\xi^{[j]})\frac{\delta y}{\beta_1\beta_2\lambda}\right)}\\
-\kappa\left(w_1^*(\lambda,\xi^{[i]})w_1(\lambda^*,\xi^{[j]*})+w_2^*(\lambda,\xi^{[i]})w_2(\lambda^*,\xi^{[j]*})\right)\mathbb{I}_2\ee^{-\frac{\ii}{4}(\xi^{[i]*}-\xi^{[j]*})\left(s+(\xi^{[i]*}+\xi^{[j]*})\frac{\delta y}{\beta_1\beta_2\lambda^*}\right)}
\end{multline}
and $\mathbf{B}(y,s)=-\kappa\mathbf{C}^{\dag}(y,s)$. Furthermore, we obtain
\begin{equation}\label{eq:lin-sol}
\begin{split}
f=&2\ii\left[\left(v_1(\lambda,\xi^{[i]})v_1^*(\lambda^*,\xi^{[j]*})+v_2(\lambda,\xi^{[i]})v_2^*(\lambda^*,\xi^{[j]*})\right)\right.\\
&\left.-\kappa\left(w_1(\lambda,\xi^{[i]})w_1^*(\lambda^*,\xi^{[j]*})+w_2(\lambda,\xi^{[i]})w_2^*(\lambda^*,\xi^{[j]*})\right)\right], \\
g_1=&\left[w_1(\lambda,\xi^{[i]})v_1^*(\lambda^*,\xi^{[j]*})+w_2^*(\lambda^*,\xi^{[j]*})v_2(\lambda,\xi^{[i]})\right], \\
g_{-1}=&\left[w_2(\lambda,\xi^{[i]})v_2^*(\lambda^*,\xi^{[j]*})+w_1^*(\lambda^*,\xi^{[j]*})v_1(\lambda,\xi^{[i]})\right], \\
g_2=&\left[w_1(\lambda,\xi^{[i]})v_2^*(\lambda^*,\xi^{[j]*})-w_2^*(\lambda^*,\xi^{[j]*})v_1(\lambda,\xi^{[i]})\right],\\
g_{-2}=&\left[w_1^*(\lambda^*,\xi^{[j]*})v_2(\lambda,\xi^{[i]})-w_2(\lambda,\xi^{[i]})v_1^*(\lambda^*,\xi^{[j]*})\right],\\
\end{split}
\end{equation}
and
\[
\eta=\frac{\xi^{[i]}-\xi^{[j]}}{4},\,\,\,\,\, \mu=\frac{\delta}{\beta_1\beta_2\lambda}\left(\xi^{[i]}+\xi^{[j]}\right),
\]
where $\xi^{[i]}$ and $\xi^{[j]}$ satisfies the quartic equation \eqref{eq:chara-eq}. If we regard $y$ as the direction of evolution, then the variable $\eta$ should be real: i.e.  $\xi^{[i]}$, $\xi^{[j]}$ have the same image part. If $\frac{\delta}{\beta_1\beta_2\lambda}\left(\xi^{[i]}+\xi^{[j]}\right)$ is also real, the Fourier mode is modulationally stable; if $\frac{\delta}{\beta_1\beta_2\lambda}\left(\xi^{[i]}+\xi^{[j]}\right)$ is not real, the Fourier mode is modulational unstable. If $\eta\to 0$ and $\frac{\delta}{\beta_1\beta_2\lambda}\left(\xi^{[i]}+\xi^{[j]}\right)$ is not real, this corresponds to resonant or based-band modulational instability. For the integrable partial differential equation, one can use the exact solution to describe the dynamics of modulational instability (MI) or resonant MI \cite{BaronioCDLOW14,ZhaoL16}. For the MI, the corresponding solutions are the so called Akhmediev breathers. For the base-band MI, the corresponding solutions are rogue wave solutions.

In what follows, we show how to construct the rogue waves of CCSP equations \eqref{eq:ccsp} which correspond to the base-band MI.

Firstly, we need to consider how to find roots which possess the same image part, i.e. $\xi^{[j]}=\xi^{[i]}+a$, where $a$ is a real constant. We rewrite equation \eqref{eq:chara-eq} with a shift as
\begin{equation}\label{eq:chara-eq-1}
F(\xi+a)\equiv(\xi+a)^4-2(\kappa |\alpha|^2+|\beta|^2+\lambda^2)(\xi+a)^2+8\beta_1\beta_2\lambda(\xi+a)+\Delta(\lambda)=0,
\end{equation}
where \[
\Delta(\lambda)=\lambda^4+2(\kappa |\alpha|^2-|\beta|^2)\lambda^2+|\alpha|^4+2\kappa(\alpha_1^2-\alpha_2^2)(\beta_1^2-\beta_2^2)+(\beta_1^2-\beta_2^2)^2.
\]
So one has
\begin{equation}\label{eq:deri-chara}
G(\xi,a)\equiv\frac{F(\xi+a)-F(\xi)}{a}\equiv 4\xi^3+6a\xi^2+4a^2\xi+a^3-2(\kappa |\alpha|^2+|\beta|^2+\lambda^2)(2\xi+a)+8\beta_1\beta_2\lambda=0.
\end{equation}
The quartic equation $F(\xi)=0$ and the cubic equation $G(\xi,a)=0$ have one common root. It follows that the resultant of $F(\xi)$ and $G(\xi)$ should equal to zero:
\begin{equation}\label{eq:resultant}
{\rm Rest}(F(\xi),G(\xi,a))=0
\end{equation}
which is an eighth order equation with respect to $\lambda$. When $a=0$, the resultant ${\rm Rest}(F(\xi),G(\xi,0))$ is the discriminant of $F(\xi)=0$. Therefore, to seek the roots with $\xi^{[i]}$ and $\xi^{[i]}+a$, we firstly solve the resultant equation \eqref{eq:resultant} for a fixed $a$. Moreover, we can obtain two groups of roots of the form $\{\pm \lambda_i,\pm\lambda_i^*\}$, $i=1,2$ if $\kappa=1$; while for $\kappa=-1$ we can obtain a quadruple $\{\pm \lambda_1,\pm\lambda_1^*\}$ and four other real roots which are not related to MI. Especially, when $a=0$ and $\kappa=1$, there are two groups of roots which corresponds to two branches of MI, which are associated to two distinct rogue waves. For the defocusing case, when $a=0$ and $\kappa=-1$, there is one group of roots, corresponding to one branch of MI, which is associated with rogue wave solution. The MI spectrum is shown in the phase diagram (Fig. \ref{MI}). 

\begin{figure}[tbh]
\centering
\includegraphics[height=50mm,width=140mm]{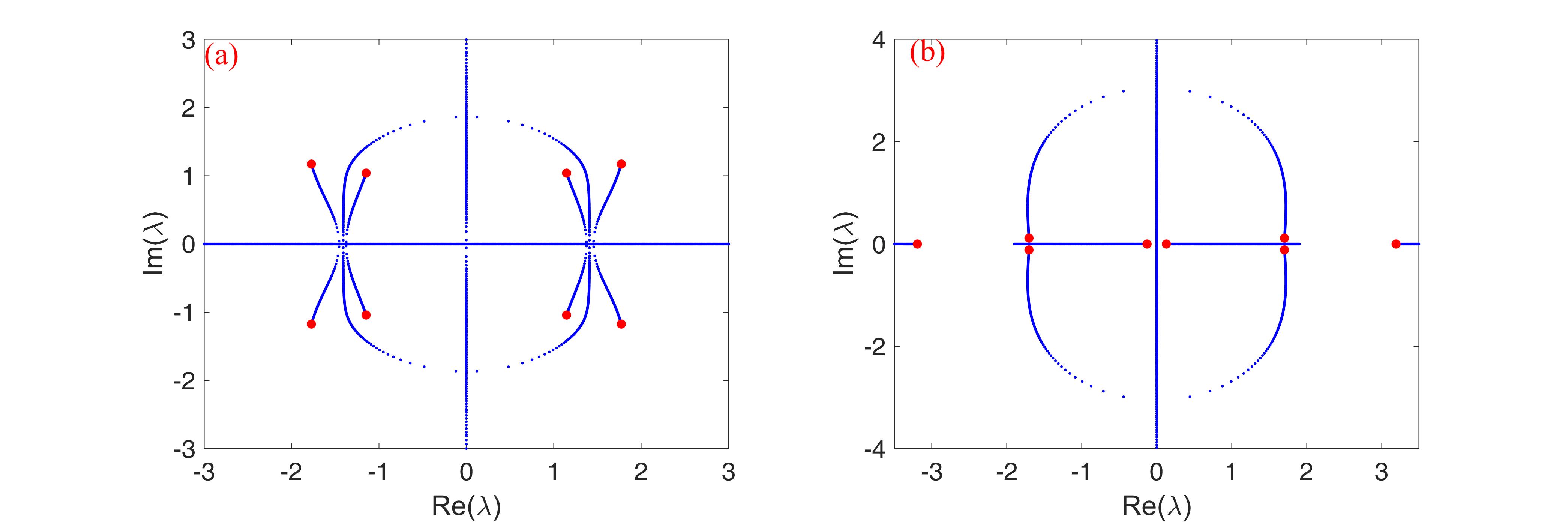}
\caption{(color on line): MI of focusing and defocusing type of CCSP equations. Parameters: $\delta=1$, $\alpha_1=\alpha_2=1$, $\beta_1=1$, $\beta_2=2$; (a) The eight red points represent the spectral points of the rogue waves. (b) The four red points located off the real axis represent the spectral points of the rogue waves.}
\label{MI}
\end{figure}

\subsection{Rogue waves of focusing and defocusing CCSP equation}
The case $a=0$ corresponds to $\xi^{[i]}=\xi^{[j]}$. In this case, the vector solutions are not involved in the expressions \eqref{eq:fund-cd}. We can seek the vector solutions by formal series expansions. Since the vector solutions \eqref{eq:fund-cd} involve the algebraic curve \eqref{eq:chara-eq},
by the Riemann-Hurwitz formula we deduce that the genus of the algebraic curve $F(\xi)=0$ is equal to $g=-n+1+B/2=1$, where $n=4$, $B=8$ is the total branching number. By the discriminant analysis, we know the discriminant equation of $F(\xi)=0$ does not possess multiple roots, so the characteristic equation \eqref{eq:chara-eq} $F(\xi)=0$ just involves double roots. In the neighborhood of a branch point $(\xi,\lambda)=(\xi_1,\lambda_1)$, the algebraic curve could have the following expansion:
\begin{equation}\label{eq:expan-xi-lambda}
\xi=\xi_1+\sum_{i=1}^{\infty} \xi_1^{[i]}\epsilon^i,\,\,\,\,\, \lambda=\lambda_1+\epsilon^2.
\end{equation}
Inserting the above formal series into the characteristic equation \eqref{eq:chara-eq} $F(\xi)=0$, we can solve for the coefficients $\xi_1^{[i]}$ recursively.

Following the steps given in \cite{guo12}, we can obtain the vector solutions with rational expression at the branch point:
\begin{equation}\label{eq:vect-rw}
\Phi_1(y,s;\lambda_1)=\mathbf{D}\begin{bmatrix}
\theta_1,&
\chi_1,&
\phi_1,&
\psi_1
\end{bmatrix}^\T \ee^{\frac{\ii}{4}\xi_1\left(s+\frac{\delta \xi_1y}{\beta_1\beta_2\lambda_1}+k_1\right)-\frac{\ii\delta(\kappa|\alpha|^2+|\beta|^2+\lambda_1^2)y}{4\beta_1\beta_2\lambda_1}}
\end{equation}
where $k_1$ is a complex parameter,
\begin{equation*}
\begin{split}
\theta_1&=\frac{\ii}{4}\left(s+\frac{2\delta \xi_1 y}{\beta_1\beta_2\lambda_1}+k_1\right)l_1+l_0, \\
\chi_1&=\frac{\ii}{4}\left(s+\frac{2\delta \xi_1 y}{\beta_1\beta_2\lambda_1}+k_1\right)\alpha_1\alpha_2l_2, \\
\phi_1&=\alpha_1\left[\frac{\ii}{4}\left(s+\frac{2\delta \xi_1 y}{\beta_1\beta_2\lambda_1}+k_1\right)l_3+2\kappa(\beta_2-\xi_1)\right], \\
\psi_1&=\alpha_2\left[\frac{\ii}{4}\left(s+\frac{2\delta \xi_1 y}{\beta_1\beta_2\lambda_1}+k_1\right)l_4-2\kappa(\beta_1+\xi_1)\right], \\
\end{split}
\end{equation*}
and
\begin{equation*}
\begin{split}
l_0(\lambda_1,\xi_1)=&\kappa\left[(\beta_2+\lambda_1)^2-(\beta_1+\xi_1)^2\right]+2\kappa(\beta_1+\xi_1)(\beta_1+\beta_2-\lambda_1-\xi_1)+|\alpha|^2, \\
l_1(\lambda_1,\xi_1)=&\kappa(\beta_1-\beta_2-\lambda_1+\xi_1)\left[(\beta_1+\beta_2)^2-(\lambda_1+\xi_1)^2\right]+(\alpha_1^2-\alpha_2^2)(\beta_1+\beta_2)+|\alpha|^2(\lambda_1+\xi_1),\\
l_2=&2(\beta_1+\beta_2),\,\,\,\,\, l_3(\lambda_1,\xi_1)=|\alpha|^2+\kappa(\lambda_1-\beta_1)^2-\kappa(\xi_1-\beta_2)^2, \\
l_4(\lambda_1,\xi_1)=&|\alpha|^2+\kappa(\lambda_1+\beta_2)^2-\kappa(\xi_1+\beta_1)^2. \\
\end{split}
\end{equation*}
Then by the Darboux transformation \eqref{eq:backlund}, the rogue wave solutions for the CCSP equations \eqref{eq:ccsp} are
\begin{equation}\label{eq:rw-sol}
\begin{split}
q_1[1]&=\left[\frac{\alpha_1}{2}-(\lambda_1-\lambda_1^*)\frac{\phi_1\theta_1^*+\psi_1^*\chi_1}{|\theta_1|^2+|\chi_1|^2+\kappa|\phi_1|^2+\kappa|\psi_1|^2}\right]\ee^{\theta_1},  \\
q_2[1]&=\left[\frac{\alpha_2}{2}-(\lambda_1-\lambda_1^*)\frac{\phi_1\chi_1^*-\psi_1^*\theta_1}{|\theta_1|^2+|\chi_1|^2+\kappa|\phi_1|^2+\kappa|\psi_1|^2}\right]\ee^{\theta_2},  \\
\rho[1]&=\frac{\delta}{2}-2 \ln_{y,s}\left(\frac{|\theta_1|^2+|\chi_1|^2+\kappa|\phi_1|^2+\kappa|\psi_1|^2}{\lambda_1-\lambda_1^*}\right),\\
x&=\frac{\delta}{2}y-\frac{\kappa}{8}(\alpha_1^2+\alpha_2^2)s-2 \ln_{s}\left(\frac{|\theta_1|^2+|\chi_1|^2+\kappa|\phi_1|^2+\kappa|\psi_1|^2}{\lambda_1-\lambda_1^*}\right),\,\,\,\, t=-s. \\
\end{split}
\end{equation}
{Fig. \ref{ccsp-rw} and Fig. \ref{dccsp-rw} display the rogue wave solutions for the focusing and defocusing CCSP equations respectively.}
\begin{figure}[tbh]
\centering
\includegraphics[height=50mm,width=140mm]{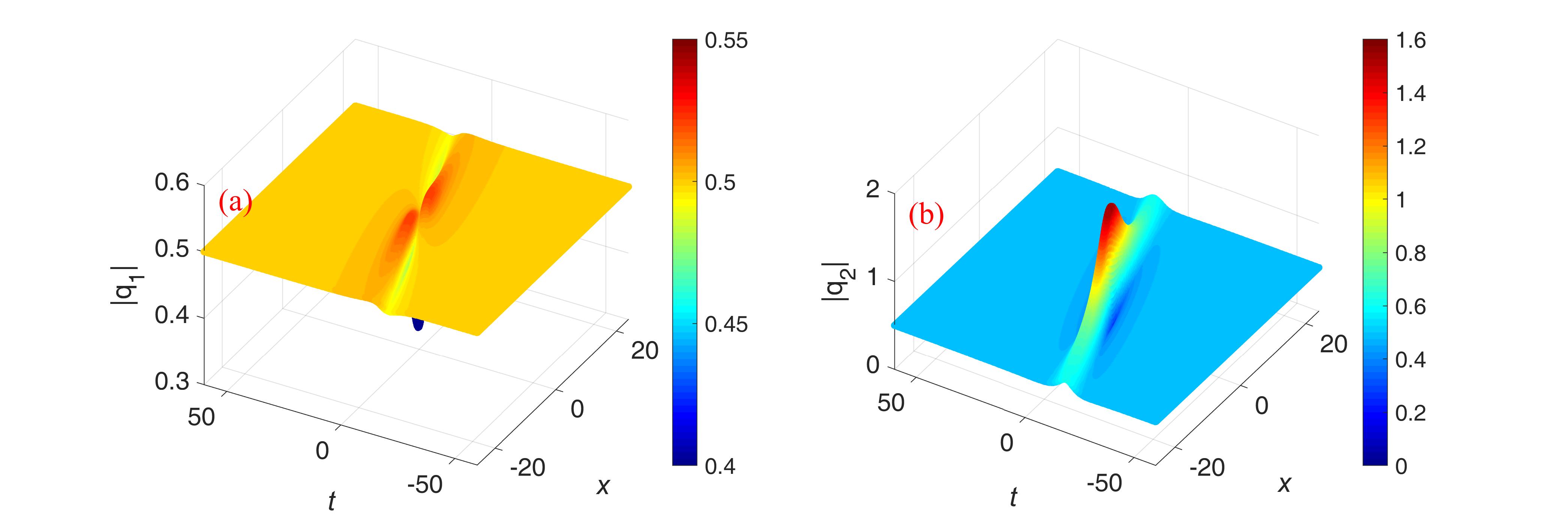}
\caption{(color on line) The rogue waves of the focusing CCSP equation. Parameters: $\delta=1$, $\alpha_1=\alpha_2=1$, $\beta_1=1$, $\beta_2=2$, $\lambda_1\approx -1.772-1.172\ii$, $\xi_1\approx-0.960-0.116\ii$}
\label{ccsp-rw}
\end{figure}
\begin{figure}[tbh]
\centering
\includegraphics[height=50mm,width=140mm]{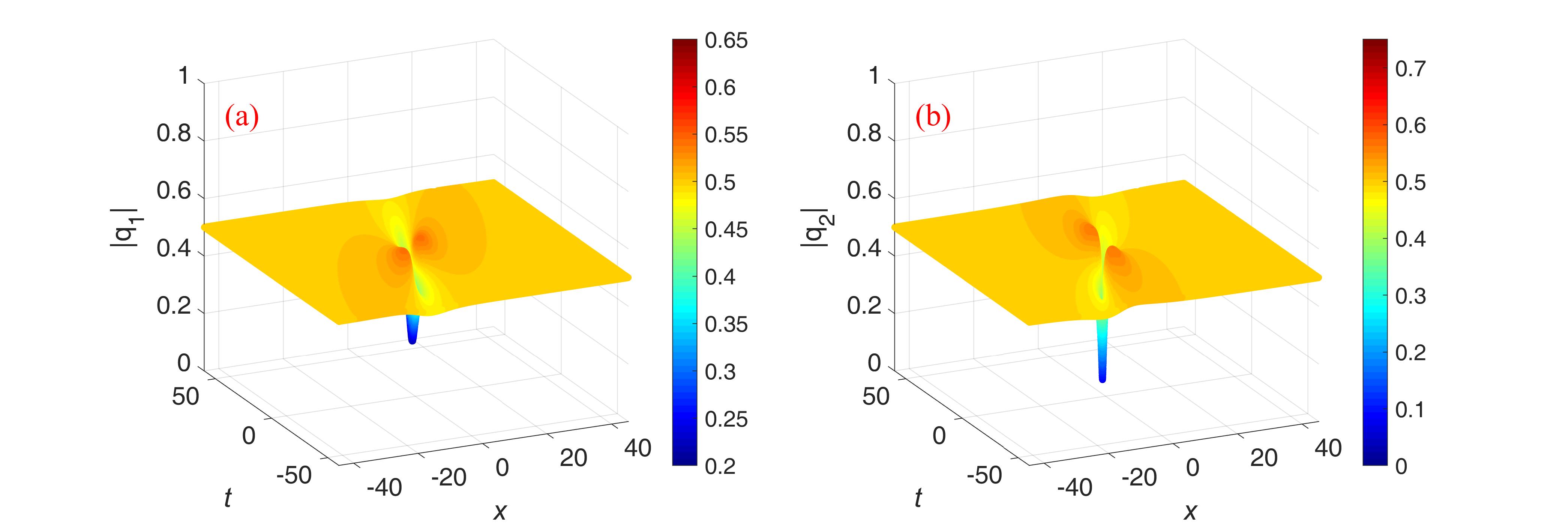}
\caption{(color on line) The rogue waves of the de-focusing CCSP equation. Parameters: Parameters: $\delta=1$, $\alpha_1=\alpha_2=1$, $\beta_1=1$, $\beta_2=2$, $\lambda_1\approx 1.705+0.115\ii$, $\xi_1\approx1.465+0.597\ii$}
\label{dccsp-rw}
\end{figure}

\subsection{Other localized wave solutions on a plane wave background}
Finally, we outline how one can construct additional exact solutions from the Darboux transformation.
\begin{description}
  \item[(1): $\alpha_1=0$, $\alpha_2\neq 0$ and $\kappa=1$]  Degenerate breather solutions, bright-dark solitons, degenerate rogue wave solutions and their combinations. For degenerate solutions we mean that solutions can be seen as the solutions of the scalar complex short-pulse equation. Combination means the nonlinear superpositions of different types of solutions. Choosing the special combinations of solutions in \eqref{eq:dege-vector} with the form:
      \begin{equation}\label{eq:ccsp-dege-b}
      \Phi_1(y,s;\lambda_1)=\Gamma_1(y,s;\lambda_1)+\gamma_1 \Gamma_2(y,s;\lambda_1),
      \end{equation}
      and inserting the vector solutions \eqref{eq:ccsp-dege-b} into formulas \eqref{eq:backlund}, we obtain that the solution $q_1$ is still a zero solution and $q_2$ is a breather solution. Especially, if $\lambda_1=\beta_2+\ii\alpha_2$ with the special limit technique as in \cite{guo12}, then $q_2$ can be shown to be the rogue wave solution.

      Choosing the special combinations of solutions in \eqref{eq:dege-vector} in the form:
      \begin{equation}\label{eq:ccsp-dege-bd}
      \Phi_1(y,s;\lambda_1)=\Gamma_1(y,s;\lambda_1)+\gamma_1 \Gamma_3(y,s;\lambda_1),
      \end{equation}
      and inserting the vector solutions \eqref{eq:ccsp-dege-b} into formulas \eqref{eq:backlund}, we obtain that the solution $q_1$ is a bright-soliton solution, and $q_2$ is a dark-soliton solution (Fig. \ref{ccsp-bd}).
      The multi-ones and higher order-ones can be constructed through the general formulas \eqref{eq:general-dt}.
\begin{figure}[tbh]
\centering
\includegraphics[height=50mm,width=140mm]{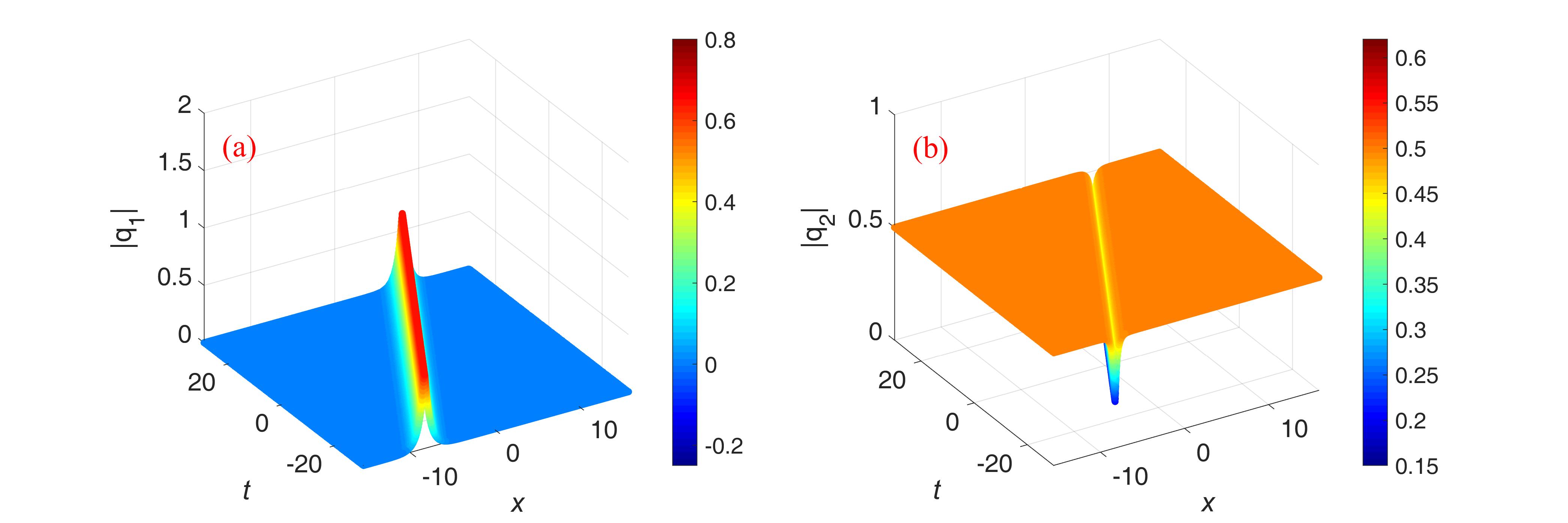}
\caption{(color on line) A bright-dark soliton of the focusing CCSP equation. Parameters: $\alpha_1=0$, $\alpha_2=1$, $\beta_1=1$, $\beta_2=-1$, $\delta=1$, $\lambda_1=-\frac{3}{4}+\frac{3\ii}{4}$, $\xi_1=\frac{1}{4}-\frac{\ii}{4}$, $\xi_3=-1-\frac{3\sqrt{7}}{4}+\frac{\ii \sqrt{7}}{4}$, $\gamma_1=1$.}
\label{ccsp-bd}
\end{figure}
  \item[(2): $\alpha_1=0$, $\alpha_2\neq 0$ and $\kappa=-1$] Degenerate dark solitons and bright-dark solitons and combinations thereof.
      Choosing the special combinations of solutions in \eqref{eq:dege-vector} with the form:
      \begin{equation}\label{eq:dcsp-dege-bd}
      \Phi_1(y,s;\lambda_1)=\Gamma_1(y,s;\lambda_1)+\gamma_1 \Gamma_3(y,s;\lambda_1),
      \end{equation}
      and inserting the vector solutions \eqref{eq:ccsp-dege-b} into formulas \eqref{eq:backlund}, we obtain that the solution $q_1$ is a bright-soliton solution, and $q_2$ is a dark-soliton solution (Fig. \ref{dccsp-bd}). The bright-dark solitons and their interactions were given by the Hirota bilinear method in \cite{GuoW-16}.
\begin{figure}[tbh]
\centering
\includegraphics[height=50mm,width=140mm]{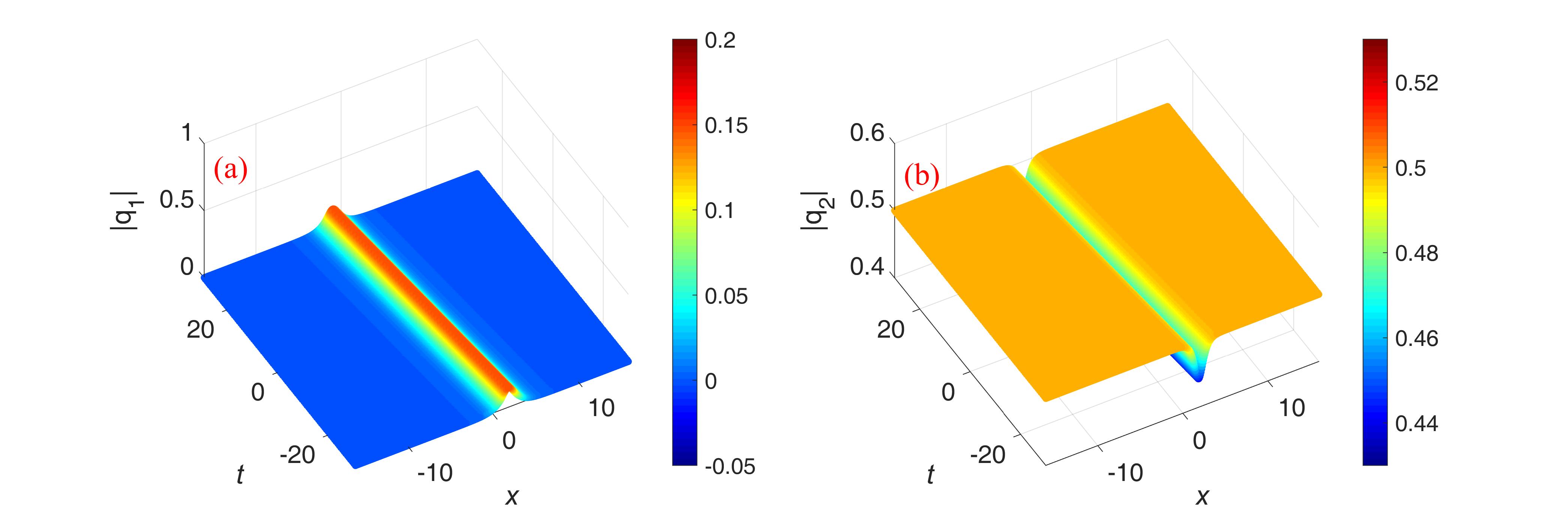}
\caption{A bright-dark soliton of the de-focusing CCSP equation. Parameters: $\alpha_1=0$, $\alpha_2=1$, $\beta_1=1$, $\beta_2=-1$, $\delta=1$, $\lambda_1=-\frac{1}{4}+\frac{\ii}{4}$, $\xi_1=\frac{3}{4}-\frac{3\ii}{4}$, $\xi_3=-1+\frac{\sqrt{\sqrt{41}+4}}{4}-\frac{\ii\sqrt{\sqrt{41}-4}}{4}$, $\gamma_1=1$.}
\label{dccsp-bd}
\end{figure}

      The degenerate dark soliton can be constructed as shown in equation \eqref{eq:dt-dark}:
      \begin{equation}\label{eq:dege-dark-soliton}
\begin{split}
q_1[1]&=0,  \,\,\,\,\,\,
q_2[1]=\frac{\alpha_2}{2}\left[1+\frac{\xi_1-\xi_1^{*}}{\lambda_1+\xi_1^{*}-\beta_1-\beta_2}\frac{E}{\Delta}\right]\ee^{\theta_2},  \\
\rho[1]&=\frac{\delta}{2}-2 \ln_{y,s}\left(\Delta\right)\geq 0,\\
x&=\frac{\delta}{2}y+\frac{\alpha_2^2}{8}s-2 \ln_{s}\left(\Delta\right),\,\,\,\, t=-s, \\
\Delta&=(E+1),\,\,\,\,\, E=\ee^{\frac{\ii}{4}\xi_1\left(s+\frac{\delta \xi_1y}{\beta_1\beta_2\lambda_1}\right)-
      \frac{\ii}{4}\xi_1^{*}\left(s+\frac{\delta \xi_1^{*}y}{\beta_1\beta_2\lambda_1}\right)+\gamma_1}\\
\end{split}
\end{equation}
where $\gamma_1$ is a real constant, $\xi_1=\beta_1+\ii \sqrt{\alpha_2^2-(\beta_2-\lambda_1)^2}$, $|\beta_2-\lambda_1|<\alpha_2.$
  \item[(3): $\alpha_1\alpha_2\neq 0$ and $\kappa=1$] Breather solutions, rogue wave solution and combinations thereof. Choosing the special solutions \eqref{eq:vector}:
      \begin{equation}\label{eq:breather}
      \Phi_1(y,s;\lambda_1)=\Gamma_k(y,s;\lambda_1)+\gamma_1\Gamma_j(y,s;\lambda_1),\,\,\,\, k\neq j,
      \end{equation}
      and inserting the vector solutions \eqref{eq:breather} into formulas \eqref{eq:backlund}, we obtain the breather solutions (Fig. \ref{ccsp-breathers}). If we choose three or more vectors, resonant breathers are obtained.
 \begin{figure}[tbh]
\centering
\includegraphics[height=50mm,width=140mm]{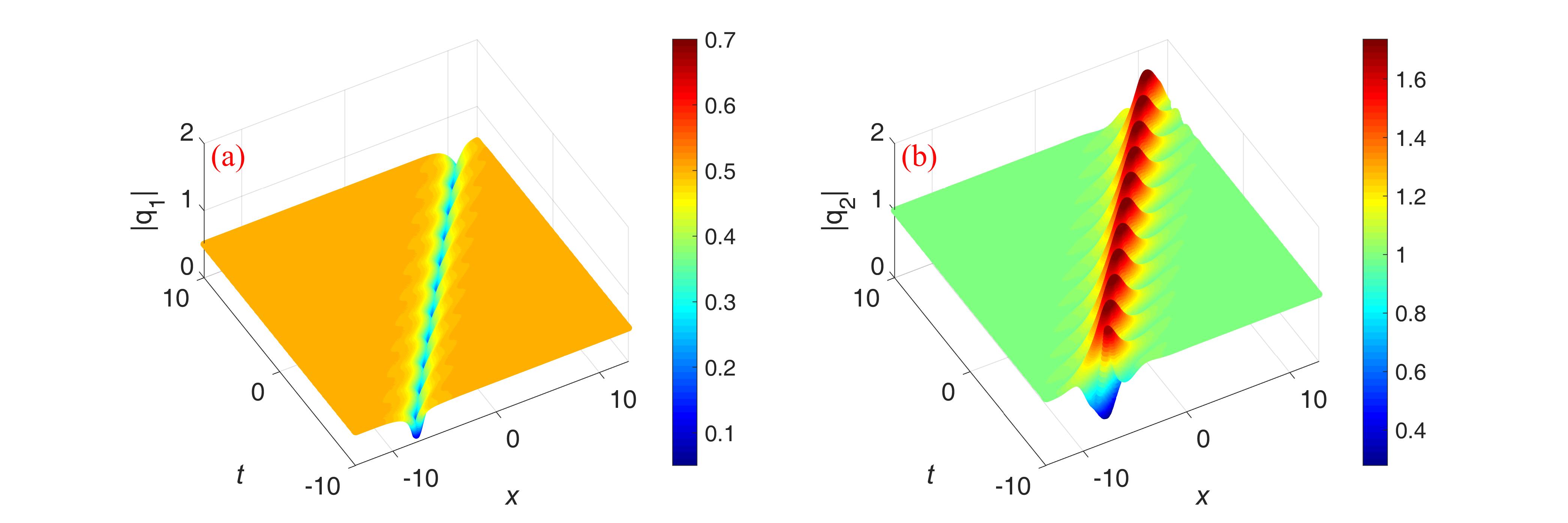}
\caption{Breathers of the focusing CCSP equation. Parameters: $\alpha_1=1$, $\alpha_2=2$, $\beta_1=1$, $\beta_2=8$, $\delta=1$, $\lambda_1=1+\ii$, $\xi_1\approx8.423-0.162\ii$, $\xi_2\approx-6.090+0.876\ii$, $\gamma_1=1$.}
\label{ccsp-breathers}
\end{figure}
      The rogue waves are given in the above subsection. The higher order rogue waves can be obtained through a similar procedure as in \cite{LingZ19}.
  \item[(4): $\alpha_1\alpha_2\neq 0$ and $\kappa=-1$] Breather solutions, dark-dark solitons \cite{LingZG15}, rogue wave solutions and combinations thereof. In the defocusing case, the breather solutions (Fig. \ref{dccsp-breathers}) can be obtained by inserting the vector solutions \eqref{eq:breather} into formulas \eqref{eq:backlund}. But the imaginary part of the roots $\xi_1^{[k]}$ and $\xi_1^{[l]}$ should have the same sign to keep the Darboux matrix positive or negative definite.
\begin{figure}[tbh]
\centering
\includegraphics[height=50mm,width=140mm]{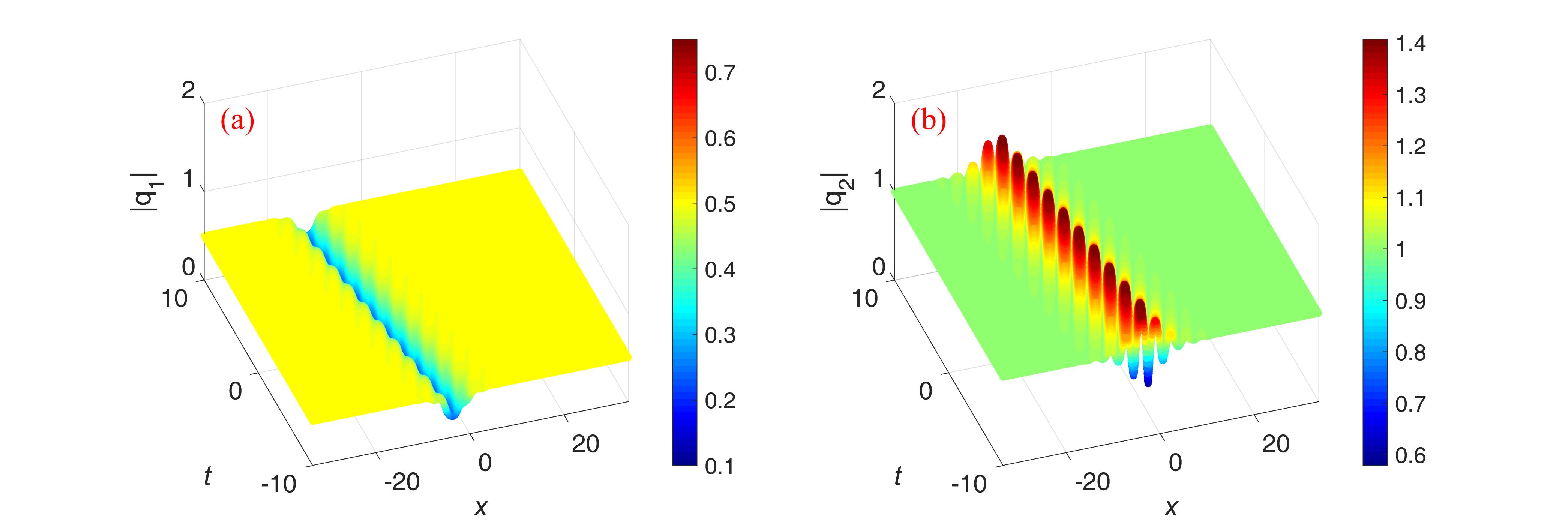}
\caption{Breathers of the defocusing CCSP equation. Parameters: $\alpha_1=1$, $\alpha_2=2$, $\beta_1=1$, $\beta_2=8$, $\delta=1$, $\lambda_1=1+\ii$, $\xi_1\approx7.772+1.452\ii$, $\xi_2\approx-5.929+1.1599\ii$, $\gamma_1=1$.}
\label{dccsp-breathers}
\end{figure}      
If $\lambda_1\in \mathbb{R}$, and there is a pair of complex roots $\xi_1^{[1]}=\xi_1^{[2]*}\equiv \xi_1$ in the characteristic equation \eqref{eq:chara-eq-1},
 \begin{equation}
      \lim_{\lambda\to\lambda_1}\frac{\Gamma_1^{\dag}(y,s;\lambda_1)\Sigma_3 \Gamma_1(y,s;\lambda)}{\lambda-\lambda_1}=\frac{2\left(|V_{1,1}(\lambda_1,\xi_1)|^2+|V_{1,2}(\lambda_1,\xi_1)|^2\right)}{\xi_1-\xi_1^*}\ee^{\frac{\ii}{4}(\xi_1-\xi_1^*)\left(s+\frac{\delta (\xi_1+\xi_1^*)y}{\beta_1\beta_2\lambda_1}\right)}
 \end{equation}
then the dark soliton can be constructed as shown in equation \eqref{eq:dt-dark}:
      \begin{equation}\label{eq:dark-soliton}
\begin{split}
q_1[1]&=\left[\frac{\alpha_1}{2}-\frac{(V_{1,3}V_{1,1}^*+V_{1,4}^*V_{1,2})E}{\Delta}\right]\ee^{\theta_1},  \,\,\,\,\,\,
q_2[1]=\left[\frac{\alpha_2}{2}-\frac{(V_{1,3}V_{1,2}^*-V_{1,4}^*V_{1,1})E}{\Delta}\right]\ee^{\theta_2},  \\
\rho[1]&=\frac{\delta}{2}-2 \ln_{y,s}\left(\Delta\right)\geq 0,\\
x&=\frac{\delta}{2}y+\frac{1}{8}(\alpha_1^2+\alpha_2^2)s-2 \ln_{s}\left(\Delta\right),\,\,\,\, t=-s, \\
\Delta&=\frac{2\left(\left|V_{1,1}\right|^2+\left|V_{1,2}\right|^2\right)(E+1)}{\xi_1-\xi_1^{*}},\,\,\,\,\, E=\ee^{\frac{\ii}{4}(\xi_1-\xi_1^*)\left(s+\frac{\delta (\xi_1+\xi_1^*)y}{\beta_1\beta_2\lambda_1}\right)+\gamma_1}\\
\end{split}
\end{equation}
where $\gamma_1$ is a real constant, $V_{1,i}$ represents the $i-$th component of $\mathbf{V}_1$ \eqref{eq:vector}. We plot the solution in Fig. \ref{dccsp-dd} for a special choice of the parameters.

\begin{figure}[tbh]
\centering
\includegraphics[height=50mm,width=140mm]{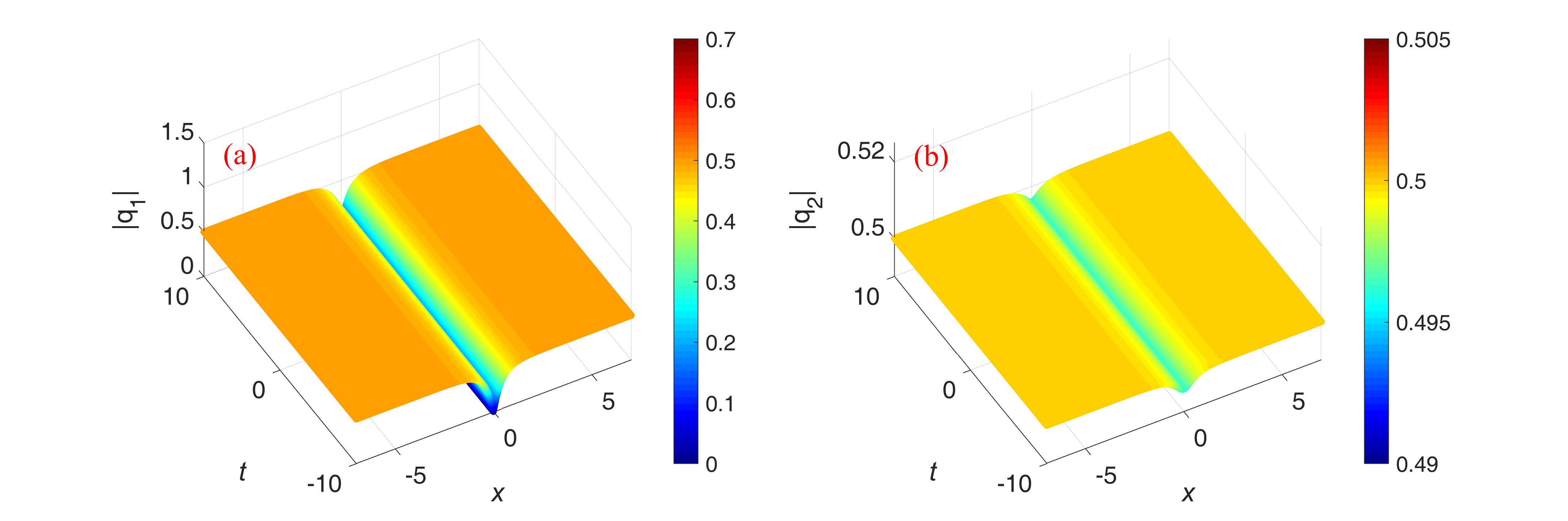}
\caption{A dark-dark soliton of the defocusing CCSP equation. Parameters: $\alpha_1=1$, $\alpha_2=2$, $\beta_1=1$, $\beta_2=8$, $\delta=1$, $\lambda_1=1$, $\xi_1\approx7.937270056+1.007999195\ii$, $\gamma_1=0$.}
\label{dccsp-dd}
\end{figure}
\end{description}

\section{Discussions and conclusions}\label{sec5}
We have constructed various localized wave solutions to the CCSP equation \eqref{eq:ccsp} by the Darboux transformation. For the focusing CCSP equation with VBC, by the scattering and inverse scattering analysis, the soliton solutions are characterized from the viewpoint of spectrum. The elementary solitons are classified into  three types: bright soliton, ${SU}(2)$ soliton, the breather and double-hump soliton.  Moreover, higher order soliton solutions are constructed.
For NVBC, we have developed a new method to analyze the modulational instability. In this way, we firstly explain the mechanism for the formation of rogue waves of two-component complex coupled dispersionless (TCCD) equation. We find that for the focusing TCCD equation, there exist two types of rogue waves which correspond to two distinct branches of MI; for the defocusing TCCD equation, there exists a type of rogue wave. This is similar to the cases of focusing and defocusing NLS equation. Since the rogue waves of the TCCD equation tend to a plane wave solution when $y,s\to \pm\infty$, rogue waves of the CCSP equation tend to a plane wave solution when $x,t\to \pm\infty$ by virtue of the inverse hodograph transformation. Then the CCSP equations exhibits three types of rogue waves: the regular, the cuspon and the loop ones. Since the parameters of the rogue waves depend on the roots of an eighth order and fourth order algebraic equation, the exact conditions for distinguishing three types of rogue waves can be only be obtained numerically.
For the mixed CCSP equation, the soliton and rogue wave solutions were constructed by the Hirota bilinear method \cite{YangZ-18}. The long time asymptotics for CSP equation can be performed by the Deift-Zhou method \cite{LiTY-21,LiuG-21,XuF-20}, so it is hopeful to consider the long time asymptotics for the CCSP equation.

\section*{Acknowledgments}
BF's work is
partially supported by National Science Foundation (NSF) under Grant No. DMS-1715991 and U.S. Department of Defense (DoD), Air Force for Scientific
Research (AFOSR) under grant No. W911NF 2010276.
LL's work is supported by National Natural Science Foundation of China (Contact Nos. 11771151, 12122105), Guangzhou Science and Technology Program (No. 201904010362). The authors appreciate Professor P.D. Miller for the useful discussion on the formal inverse scattering analysis.  The authors also thank Professor Prinari for her careful reading on the original manuscript and useful suggestions.

\end{document}